\documentclass[11pt]{article}

\usepackage{amsthm}
\usepackage{graphicx} 
\usepackage{array} 

\usepackage{amsmath, amssymb, amsfonts, verbatim}
\usepackage{hyphenat,epsfig,subcaption,multirow}
\usepackage[font=small,labelfont=bf]{caption}
\usepackage{nicefrac}
\usepackage{paralist}

\usepackage{tocloft}

\usepackage[usenames,dvipsnames]{xcolor}
\usepackage[ruled]{algorithm2e}

\DeclareFontFamily{U}{mathx}{\hyphenchar\font45}
\DeclareFontShape{U}{mathx}{m}{n}{
      <5> <6> <7> <8> <9> <10>
      <10.95> <12> <14.4> <17.28> <20.74> <24.88>
      mathx10
      }{}
\DeclareSymbolFont{mathx}{U}{mathx}{m}{n}
\DeclareMathSymbol{\bigtimes}{1}{mathx}{"91}

\usepackage{tcolorbox}
\tcbuselibrary{skins,breakable}
\tcbset{enhanced jigsaw}

\usepackage[normalem]{ulem}
\usepackage[compact]{titlesec}

\definecolor{DarkRed}{rgb}{0.5,0.1,0.1}
\definecolor{DarkBlue}{rgb}{0.1,0.1,0.5}

\usepackage{nameref}
\definecolor{ForestGreen}{rgb}{0.1333,0.5451,0.1333}
\definecolor{Red}{rgb}{0.9,0,0}
\usepackage[linktocpage=true,
	pagebackref=true,colorlinks,
	linkcolor=DarkRed,citecolor=ForestGreen,
	bookmarks,bookmarksopen,bookmarksnumbered]
	{hyperref}
\usepackage[noabbrev,nameinlink]{cleveref}
\crefname{property}{property}{Property}
\creflabelformat{property}{(#1)#2#3}
\crefname{equation}{eq}{Eq}
\creflabelformat{equation}{(#1)#2#3}

\usepackage{bm}
\usepackage{url}
\usepackage{xspace}
\usepackage[mathscr]{euscript}

\usepackage{tikz}
\usetikzlibrary{arrows}
\usetikzlibrary{arrows.meta}
\usetikzlibrary{shapes}
\usetikzlibrary{backgrounds}
\usetikzlibrary{positioning}
\usetikzlibrary{decorations.markings}
\usetikzlibrary{decorations.pathreplacing} 
\usetikzlibrary{patterns}
\usetikzlibrary{calc}
\usetikzlibrary{fit}
\usetikzlibrary{decorations}

\usepackage[framemethod=TikZ]{mdframed}

\usepackage[noend]{algpseudocode}
\makeatletter
\def\BState{\State\hskip-\ALG@thistlm}
\makeatother

\usepackage{cite}
\usepackage{enumitem}
\setlist[itemize]{leftmargin=20pt}
\setlist[enumerate]{leftmargin=20pt}

\usepackage[margin=1in]{geometry}

\usepackage{thmtools}
\usepackage{thm-restate}

\newtheorem{theorem}{Theorem}
\newtheorem{lemma}{Lemma}[section]
\newtheorem{proposition}[lemma]{Proposition}

\newtheorem{claim}[lemma]{Claim}
\newtheorem{fact}[lemma]{Fact}

\newtheorem{question}{Question}
\crefname{question}{Question}{Questions}

\newtheorem{definition}[lemma]{Definition}
\newtheorem{problem}{Problem}

\newtheorem*{claim*}{Claim}
\newtheorem*{assumption*}{Assumption}
\newtheorem*{proposition*}{Proposition}
\newtheorem*{lemma*}{Lemma}

\newtheorem{observation}[lemma]{Observation}

\newtheorem*{theorem*}{Theorem}

\crefname{lemma}{Lemma}{Lemmas}
\crefname{claim}{claim}{claims}
\crefname{property}{Property}{Properties}
\crefname{invariant}{Invariant}{Invariants}

\newtheorem{mdresult}{Result}
\newenvironment{result}{\begin{mdframed}[backgroundcolor=lightgray!40,topline=false,rightline=false,leftline=false,bottomline=false,innertopmargin=5pt]\begin{mdresult}}{\end{mdresult}\end{mdframed}}

\newtheorem{remark}[lemma]{Remark}

\theoremstyle{definition}

\newtheorem*{mdproblem*}{Problem}
\newenvironment{Problem*}{\begin{mdframed}[hidealllines=false,innerleftmargin=10pt,backgroundcolor=gray!10,innertopmargin=5pt,innerbottommargin=5pt,roundcorner=10pt]\begin{mdproblem*}}{\end{mdproblem*}\end{mdframed}}
\newtheorem{mddefinition}[lemma]{Definition}
\newenvironment{Definition}{\begin{mdframed}[hidealllines=false,innerleftmargin=10pt,backgroundcolor=white!10,innertopmargin=5pt,innerbottommargin=5pt,roundcorner=10pt]\begin{mddefinition}}{\end{mddefinition}\end{mdframed}}
\newtheorem*{mddefinition*}{Definition}
\newenvironment{Definition*}{\begin{mdframed}[hidealllines=false,innerleftmargin=10pt,backgroundcolor=white!10,innertopmargin=5pt,innerbottommargin=5pt,roundcorner=10pt]\begin{mddefinition*}}{\end{mddefinition*}\end{mdframed}}
\newtheorem{mdremark}{Remark}

\newenvironment{ourbox}{\begin{mdframed}[hidealllines=false,innerleftmargin=10pt,backgroundcolor=white!10,innertopmargin=5pt,innerbottommargin=5pt,roundcorner=10pt]}{\end{mdframed}}

\newtheorem{mdalgorithm}{Algorithm}
\newenvironment{Algorithm}{\begin{ourbox}\begin{mdalgorithm}}{\end{mdalgorithm}\end{ourbox}}

\newtheorem{mddistribution}{Distribution}
\newenvironment{Distribution}{\begin{ourbox}\begin{mddistribution}}{\end{mddistribution}\end{ourbox}}

\newtheorem{mddistributionpart}{Part}
\newenvironment{Part}{\begin{ourbox}\begin{mddistributionpart}}{
		\end{mddistributionpart}
\end{ourbox}}

\allowdisplaybreaks

\renewcommand{\qed}{\nobreak \ifvmode \relax \else
      \ifdim\lastskip<1.5em \hskip-\lastskip
      \hskip1.5em plus0em minus0.5em \fi \nobreak
      \vrule height0.75em width0.5em depth0.25em\fi}

\newcommand{\Qed}[1]{\rlap{\qed$_{\textnormal{~~\Cref{#1}}}$}}

\newcommand{\logstar}[1]{\ensuremath{\log^{*}\!{#1}}}

\renewcommand{\leq}{\leqslant}
\renewcommand{\geq}{\geqslant}

\renewcommand{\ge}{\geq}

\newcommand{\rs}{{{Ruzsa-Szemerédi}}\xspace}

\newcommand{\tvd}[2]{\ensuremath{\norm{#1 - #2}_{\mathrm{tvd}}}}
\newcommand{\Ot}{\ensuremath{\widetilde{O}}}
\newcommand{\eps}{\ensuremath{\varepsilon}}
\newcommand{\Paren}[1]{\Big(#1\Big)}
\newcommand{\Bracket}[1]{\Big[#1\Big]}
\newcommand{\bracket}[1]{\left[#1\right]}
\newcommand{\paren}[1]{\ensuremath{\left(#1\right)}\xspace}
\newcommand{\card}[1]{\left\vert{#1}\right\vert}
\newcommand{\Omgt}{\ensuremath{\widetilde{\Omega}}}

\newcommand{\IN}{\ensuremath{\mathbb{N}}}

\newcommand{\norm}[1]{\ensuremath{\|#1\|}}

\newcommand{\set}[1]{\ensuremath{\left\{ #1 \right\}}}
\newcommand{\poly}{\mbox{\rm poly}}

\DeclareMathOperator*{\Exp}{\ensuremath{{\mathbb{E}}}}
\DeclareMathOperator*{\Prob}{\ensuremath{\textnormal{Pr}}}
\renewcommand{\Pr}{\Prob}

\newcommand{\Ex}{\Exp}

\newcommand{\event}{\ensuremath{\mathcal{E}}}

\newcommand{\rv}[1]{\ensuremath{{\mathsf{#1}}}\xspace}
\newcommand{\rA}{\rv{A}}
\newcommand{\rB}{\rv{B}}
\newcommand{\rC}{\rv{C}}
\newcommand{\rD}{\rv{D}}

\newcommand{\rX}{\rv{X}}
\newcommand{\rY}{\rv{Y}}
\newcommand{\rL}{\rv{L}}

\newcommand{\supp}[1]{\ensuremath{\textnormal{\text{supp}}(#1)}}
\newcommand{\distribution}[1]{\ensuremath{\textnormal{dist}(#1)}\xspace}

\newcommand{\kl}[2]{\ensuremath{\mathbb{D}(#1~||~#2)}}
\newcommand{\II}{\ensuremath{\mathbb{I}}}
\newcommand{\HH}{\ensuremath{\mathbb{H}}}
\newcommand{\mi}[2]{\ensuremath{\def\mione{#1}\def\mitwo{#2}\mireal}}
\newcommand{\mireal}[1][]{
  \ifx\relax#1\relax%
    \II(\mione \,; \mitwo)%
  \else%
    \II(\mione \,; \mitwo\mid #1)%
  \fi
}
\newcommand{\en}[1]{\ensuremath{\HH(#1)}}

\newcommand{\itfacts}[1]{\Cref{fact:it-facts}-(\ref{part:#1})\xspace}

\newcommand{\chrom}{\ensuremath{\chi}}

\newcommand{\player}[1]{\ensuremath{\mathcal{P}_{#1}}}

\newcommand{\cG}{\ensuremath{\mathcal{G}}}
\newcommand{\cD}{\ensuremath{\mathcal{D}}}

\newcommand{\cliqnum}[2]{\ensuremath{C^{(#1)}_{#2}}}

\newcommand{\ans}{\ensuremath{\textnormal{\textsc{ans}}}}

\newcommand{\spec}{\ensuremath{\textnormal{\textsc{spec}}}}

\newcommand{\istar}{\ensuremath{i^{\star}}}

\newcommand{\prot}{\ensuremath{\pi}}

\newcommand{\estar}{\ensuremath{e^{\star}}}

\newcommand{\Prot}{\ensuremath{\Pi}}

\newcommand{\rProt}{\ensuremath{\rv{\Prot}}}

\newcommand{\ic}{\ensuremath{\textnormal{\textsc{IC}}}}

\newcommand{\cDSI}{\ensuremath{\mathcal{D}_{\textnormal{\textsc{SI}}}}}

\newcommand{\rR}{\ensuremath{\rv{R}}}

\newcommand{\rT}{\ensuremath{\rv{T}}}
\newcommand{\rS}{\ensuremath{\rv{S}}}

\newcommand{\protsetint}{\ensuremath{\prot_{\textnormal{\textsc{SI}}}}}

\newcommand{\protindex}{\ensuremath{\prot_{\textnormal{\textsc{IND}}}}}

\newcommand{\rsigma}{\ensuremath{\bm{\sigma}}}

\newenvironment{tbox}{\begin{tcolorbox}[
		enlarge top by=5pt,
		enlarge bottom by=5pt,
		 breakable,
		 boxsep=0pt,
                  left=4pt,
                  right=4pt,
                  top=10pt,
                  arc=0pt,
                  boxrule=1pt,toprule=1pt,
                  colback=white
                  ]
	}
{\end{tcolorbox}}

\title{Coloring Graphs with Few Colors in the Streaming Model}
\author{Sepehr Assadi\footnote{(sepehr@assadi.info) School of Computer Science, University of Waterloo. 
Supported in part by a Sloan Research Fellowship, an NSERC
Discovery Grant, and a Faculty of Math Research Chair grant. \smallskip} \\ {\small University of Waterloo} \and 
Janani Sundaresan\footnote{(jsundaresan@uwaterloo.ca) School of Computer Science, University of Waterloo. Supported in part by a Cheriton Scholarship from School of Computer Science, a Faculty of Math Graduate Research Excellence Award, and Sepehr Assadi's NSERC Discovery Grant. \smallskip} \\ {\small University of Waterloo}
\and Helia Yazdanyar\footnote{(hyazdanyar@uwaterloo.ca) Cheriton School of Computer Science, University of Waterloo. Supported in part by Sepehr Assadi's NSERC Discovery Grant.} \\ {\small University of Waterloo}
}

\date{}

\begin{document}

	\maketitle
	
	\pagenumbering{roman}
	

\begin{abstract}

\bigskip

We study graph coloring problems in the streaming model, wherein the goal is to process an $n$-vertex graph whose edges arrive in a stream, using a limited space that is much smaller than the trivial $O(n^2)$ bound. 
While prior work has largely focused on coloring graphs with a large number of colors---typically as a function of the maximum degree---we explore the opposite end of the spectrum: 
deciding whether the input graph can be colored using only a few, say, a constant number of colors. We are interested in each of the adversarial, random order, or dynamic streams, and---as is the standard in this model---focus
 solely on the \textbf{space complexity} rather than running time.  

\smallskip

Our work lays the foundation for this new direction by establishing both upper and lower bounds on space complexity of key variants of the problem. Some of our main results include:

\smallskip

\begin{itemize}[leftmargin=10pt]
	\item \textbf{Adversarial:} for distinguishing between $q$- vs $2^{\Omega(q)}$-colorable graphs, lower bounds of $n^{2-o(1)}$ space  for $q$ up to $(\log{n})^{1/2-o(1)}$, and $n^{1+\Omega(1/\log\log{n})}$ space for $q$ further up to $(\log{n})^{1-o(1)}$. 
	\smallskip
	
	\item \textbf{Random order:} for distinguishing between $q$- vs $q^t$-colorable graphs for $q,t \geq 2$, an upper bound of $\Ot(n^{1+1/t})$ space.  
	Specifically, distinguishing between $q$-colorable graphs vs ones that are not even $\poly(q)$-colorable can be done in $n^{1+o(1)}$ space unlike in adversarial streams. Although, 
	distinguishing between $q$-colorable vs $\Omega(q^2)$-colorable graphs requires $\Omega(n^2)$ space even in random order streams for constant $q$. 
	
	\smallskip
	
	\item \textbf{Dynamic:} for distinguishing between $q$- vs $q \cdot t$-colorable graphs for any $q \geq 3$ and $t \geq 1$, nearly optimal upper and lower bounds of $\widetilde{\Theta}(n^2/t^2)$ space.  
\end{itemize}

\smallskip
In establishing our results, we develop several new technical tools that may be of independent interest. These include cluster packing graphs, a generalization of \rs graphs; a player elimination framework based on cluster packing graphs; and new edge and vertex sampling lemmas tailored to graph coloring.


\end{abstract}

	\clearpage
	
	\setcounter{tocdepth}{3}
	\tableofcontents
	
	\clearpage
	
	\pagenumbering{arabic}
	\setcounter{page}{1}
	

\section{Introduction}\label{sec:intro}

The \emph{chromatic number} of a graph $G=(V,E)$, denoted by $\chi(G)$, is the smallest integer $k \geq 1$ such that vertices of $G$ can be colored from $k$ colors
without creating any monochromatic edges. The algorithmic study of the chromatic number of a given graph, the graph coloring problem, dates back to at least half a century ago\footnote{Graph coloring also has 
a much older history in graph theory, dating back to the ``Four Color Problem''~\cite{AppelH76,AppelHK77} conjectured first in the nineteen century.}. Since then, it has been studied 
in numerous algorithmic models including approximation, exponential time, online, parallel, dynamic, or distributed algorithms, to name a few. In this work, we study graph coloring in the \textbf{graph streaming} model. 

In the graph streaming model, introduced by~\cite{FeigenbaumKMSZ04}, the vertices of the graph are known in advance and are denoted by $[n] := \set{1,2,\ldots,n}$, and 
the edges are presented to the algorithm in a stream. The algorithm needs to make a single pass (or sometimes, a few passes) over this stream and use a limited memory---much smaller than $O(n^2)$ that is sufficient for storing
all the edges explicitly---to solve a given problem. The stream can be presented to the algorithm in an \emph{adversarial order} or a \emph{random order}~\cite{McGregor14}, or may even contain both insertion and deletion of edges called 
the \emph{dynamic} streaming model~\cite{AhnGM12}. 

Starting from~\cite{BeraG18,AssadiCK19a}, there has been 
a growing body of work on graph coloring in the streaming model~\cite{CormodeDK19,BeraCG20,AlonA20,BhattacharyaBMU21,AssadiKM22,HalldorssonKNT22,AssadiCS22,ChakrabartiGS22,AssadiCGS23,FlinGHKN24}. This line of work has primarily focused on coloring graphs with a ``large'' number of colors, often
aiming to match \emph{combinatorial} upper bounds on the chromatic number via streaming algorithms: e.g., $(\Delta+1)$ coloring of max-degree $\Delta$ graphs~\cite{AssadiCK19a,AssadiY25}, streaming Brooks' Theorem for $\Delta$-coloring~\cite{AssadiKM22}, $O(\Delta/\ln(\Delta))$ coloring of triangle-free graphs~\cite{AlonA20}, degeneracy coloring~\cite{BeraCG20}, $(\deg+1)$-list coloring~\cite{HalldorssonKNT22}, and so on. 
There is also work on understanding possibility of finding $O(\Delta)$ or even $\poly(\Delta)$ coloring using deterministic algorithms~\cite{AssadiCS22} or adversarially robust ones~\cite{ChakrabartiGS22,AssadiCGS23}. 

For all aforementioned problems, 
existence of such a coloring is guaranteed using standard graph theory results and the challenge is in \emph{finding} the coloring via a streaming algorithm. 
However, we can also go beyond these combinatorial bounds and ask the following \emph{purely algorithmic} question: 
\begin{question}\label{question}
	How well can we \textbf{estimate} the chromatic number of a graph $G$ presented in a stream? Specifically, for what values of $q < Q$, there are non-trivial
	streaming algorithms that \textbf{distinguish} between $\chi(G) \leq q$ versus $\chi(G) \geq Q$?
\end{question}

With a slight abuse of notation, we call the task of deciding whether $\chi(G)\leq q$ or $\chi(G)\geq Q$ as the problem of distinguishing $q$-colorable graphs vs $Q$-colorable graphs.

In this work, we are specifically interested in the case \textbf{when $q$ is a small number}, say, a constant or polylogarithmic. This choice is motivated
by the goal of complementing the large body of work in the streaming model for ``large'' chromatic number mentioned earlier, as well as the rich literature on this particular range
in the classical setting; see, e.g.,~\cite{Furer95,KhannaLS00,Khot01,GuruswamiK04,DinurMR06,Huang13,BrakensiekG16,KrokhinOWZ23} and references therein.

Not much is known about~\Cref{question} at this point: we know 
that testing two colorability (bipartiteness) can be done in $\Theta(n\log{n})$ space~\cite{FeigenbaumKMSZ04,SunW15} ($q=2$, $Q \geq 3$), while three colorability cannot be tested in $o(n^2)$ space~\cite{AbboudCKP21} ($q=3$, $Q=4$); 
more generally, distinguishing $q$-colorable graphs from $\Omega(q^2)$-colorable ones requires $\Omgt(n^2)$ space for $q = \poly\!\log{\!(n)}$~\cite{CormodeDK19}. 

Our goal in this work is to initiate a systematic study of~\Cref{question} and lay the foundation for this fundamental problem in the streaming model. We emphasize that---as is the standard in this model---we focus
 solely on the \textbf{space complexity} rather than running times: this means our algorithms may use exponential-time and conversely, all our lower bounds hold regardless of the running time of the algorithms and runtime considerations
 such as NP-hardness. 
 

\subsection{Our Contribution} 

\paragraph{Adversarial streams.} Our first result---and main technical contribution---is a space lower bound in adversarial streams for an exponentially larger range of the approximability gap, namely, the gap between $Q$ and $q$ in~\Cref{question}, compared to~\cite{CormodeDK19}. 

\begin{result}\label{res:adv}
	For $3 \leq q \leq o(\sqrt{\log{n}})$, any streaming algorithm on \textbf{adversarial streams} that distinguishes between $q$-colorable graphs and $2^{\Omega(q)}$-colorable ones requires $n^{2-o(1)}$ space. 
	
	Moreover, for the larger range of $3 \leq q \leq o(\frac{\log{n}}{\log\log{n}})$, 
	a weaker lower bound of $n^{1+\Omega(1/\log\log{n})}$ space holds for the same problem. 
\end{result}

\Cref{res:adv} implies that for~\Cref{question} in adversarial streams, almost nothing non-trivial can be done even when $Q = 2^{\Omega(q)}$. 
This is conceptually inline with the state-of-the-art hardness of approximation for polynomial time algorithms, wherein distinguishing $q$- from $o(2^{q}/\sqrt{q})$-colorable graphs is proven to be NP-hard~\cite{KrokhinOWZ23} (but, we again emphasize that~\Cref{res:adv} holds unconditionally and is for streaming \emph{space} and not time); however, to our knowledge, the NP-hardness results only
hold for a much more limited range of $q$ and certainly not $q$ close to $(\log{n}/\log\log{n})$ in the second part of the result. In fact, while proving a weaker space lower bound, this part still rules out algorithms with $\Ot(n)$ space---often referred to as \emph{semi-streaming} 
algorithms---that are of special interest in this model~\cite{FeigenbaumKMSZ04}, for almost the entire possible range of $q$ for this problem: for $q \geq \omega(\log{n})$, we have $2^{\Omega(q)} > n$ 
and since every graph is $n$-colorable, distinguishing between $q$- and $2^{\Omega(q)}$-colorable graphs becomes trivial. 



Before moving on, we remark that the $n^{1+\Omega(1/\log\log{n})}$ term in~\Cref{res:adv} is a common bound in streaming lower bounds, see, e.g.~\cite{GoelKK12,Kapralov13,AssadiKL17,CormodeDK19,Kapralov21,KonradT25},
and is often related to techniques relying on \emph{\rs} (RS) graphs~\cite{RuzsaS78}. In our case, we construct a new family of graphs (\Cref{prop:cpg-large-r}), termed \emph{cluster packing} graphs, that generalize RS graphs by replacing induced matchings in RS graphs with induced clusters of cliques and use them in our lower bound. We discuss these graphs and their origin from~\cite{AssadiKNS24} in~\Cref{sec:stronger-cluster-packing},
and for now only mention that it is entirely \emph{plausible} for the second part of~\Cref{res:adv} to also hold for $n^{2-o(1)}$ space algorithms and $q = \Theta(\log{n})$ -- proving such a result ``only'' requires constructing denser RS graphs and cluster 
packing graphs with certain parameters\footnote{Constructing (or ruling out possibility of) denser RS graphs (with linear size induced matchings) has been a notoriously challenging problem; see, e.g.~\cite{FoxHS17,AssadiS23} 
and~\Cref{app:cpg-parameters} for some pointers.}.

\paragraph{Random order streams.} \Cref{res:adv} shows a negative answer to~\Cref{question} for $Q = 2^{\Omega(q)}$ on adversarial streams. 
To bypass this impossibility result, we consider the problem in random order streams, wherein the graph is still chosen arbitrary but its edges are arriving in a random order in the stream. 
We present the algorithm in \Cref{res:rand} in this model. 

\Cref{res:rand} shows that we can already have non-trivial streaming algorithms for distinguishing between $q$- and $>\!\!q^2$-colorable graphs in $\Ot(n^{3/2})$ space in random order streams;
the space can then be reduced to $n^{1+o(1)}$ for distinguishing between $q$-colorable and graphs that are not even $\poly(q)$-colorable. This is in sharp contrast with~\Cref{res:adv} in adversarial streams. 
We also note that our algorithm in~\Cref{res:rand} naturally takes exponential-time given the NP-hardness of distinguishing between $q$- vs $2^{o(q)}$-colorable graphs~\cite{KrokhinOWZ23}. 

At the same time, in light of this algorithmic improvement, we could have hoped to reduce the approximability gap in random order streams even below the $q$- vs $q^2$-bound (the smallest gap where the upper bound of~\Cref{res:rand} kicks in); 
for instance, what about $q$ vs $2q$ or even $q+o(q)$? The second part of this result shows no better results are possible and our algorithm is asymptotically optimal in this respect for constant (and even polylogarithmic) values of $q$. 

\begin{result}\label{res:rand}
	There exists a streaming algorithm that given any integers $q,t  \geq 2$, and an input graph in \textbf{random order stream}, with (exponentially) high probability distinguishes between $q$-colorable graphs and $Q$-colorable ones for any $Q > q^t$ in $\Ot(n^{1+1/t})$ space. 
	
	On the other hand, distinguishing $q$-colorable graphs from $\Omega(q^2)$-colorable ones with high constant probability in the same model requires $\Omega(n^2)$ space
	for any {constant} $q \geq 3$. 
\end{result}

\paragraph{Dynamic streams.} \Cref{res:rand} relaxes the model in \Cref{res:adv} to obtain an algorithm that bypasses the lower bounds in the latter result. 
We now consider the opposite case: strengthening~\Cref{res:adv} to rule out streaming algorithms for even larger approximability gaps. Specifically, 
while the gap of $q$ vs $2^{\Omega(q)}$ in~\Cref{res:adv} cannot be improved much further when $q \approx \log{n}$, one can ask what happens when $q$ is small, say, only a constant. 
While we do not yet know the answer to this question in adversarial streams, we can fully resolve it in dynamic streams: here, the stream consists of edge insertions and deletions, with the guarantee 
that no edge is deleted more times than it has been inserted, and the goal is to solve the problem on the final set of edges present (namely, the ones inserted more than deleted) at the end of the stream. 

\begin{result}\label{res:dynamic}
	For any $3 \leq q \leq \poly\!\log{\!(n)}$, and any $t > 1$, the space complexity of distinguishing between $q$-colorable graphs and $q\cdot t$-colorable graphs	in \textbf{dynamic streams} is $\widetilde{\Theta}(n^2/t^2)$ space. 
\end{result}

We emphasize that~\Cref{res:dynamic} consists of both an algorithm as well as a lower bound, which match each other up to some $\poly\!\log\!{(n)}$ factors. 

An interesting corollary of~\Cref{res:dynamic} is that
semi-streaming algorithms---the ones with $\Ot(n)$ space---cannot even
distinguish between 3- vs $n^{1/2-o(1)}$-colorable graphs in dynamic streams. This lower bound is much stronger than 
computational lower bounds (NP-hardness or stronger conditional lower bounds) for $3$-coloring in the classical setting~\cite{DinurMR06} or even the existing polynomial time
algorithms for distinguishing between 3- vs $O(n^{0.19...})$-colorable graphs~\cite{KawarabayashiTY24,KawarabayashiT17}. 

Moreover, our lower bound in~\Cref{res:dynamic} is \emph{not} entirely restricted to dynamic streams: it also rules out many standard techniques like sampling and sketching algorithms that are also used quite frequently in insertion-only streams. It thus suggests 
that if there are better algorithms in insertion-only (adversarial order) streams---compared to the one implied by~\Cref{res:dynamic} already---, they should rely on techniques that are ``inherently insertion-only based''. 

\paragraph{Perspective: NP-hard problems in the streaming model.} There is a large and rapidly growing body of work on studying NP-hard problems from the streaming perspective including highly fruitful lines of work on streaming coverage problems~\cite{DemaineIMV14,Har-PeledIMV16,AssadiKL16,Assadi17,AssadiK18,IndykV19,McGregorV19,ChakrabartiMW24} and streaming constrained satisfaction problems~\cite{GuruswamiVV17,ChouGV20,ChouGSV21,SingerSV21,ChouGSVV22,SaxenaSSV23,ChouGSV24} or work 
on specific problems such as TSP~\cite{ChenKT23,AlipourFM25}, Steiner forests~\cite{CzumajJKV22}, facility location~\cite{CzumajJKVY22}, and longest paths~\cite{KonradT25}. 

This body of work has both advanced the development of new techniques within the streaming model and revealed novel insights about powers and limitations of the model itself. At the same time, it has also 
deepened our understanding of the underlying problems in the classical setting by viewing them through a new lens. Our work continues along this trajectory by introducing new techniques and constructions that we 
hope will find further applications in the streaming model and beyond; we shall elaborate more on our techniques in~\Cref{sec:techniques}. 

\subsection{Open Problems}\label{sec:open} 

Our results make progress on~\Cref{question} from different angles. However, we are still quite far from having a complete answer to this question especially in adversarial and random-order streams. 
In the following, we list some immediate next steps as main questions 
left open by our work: 
\begin{enumerate}
	\item \Cref{res:adv} implies that there is effectively no non-trivial single-pass streaming algorithm in adversarial streams for distinguishing between $q$-colorable and $2^{\Omega(q)}$-colorable graphs. 
	But, can we distinguish between $q$-colorable and $f(q)$-colorable graphs for any function $f$ in this model? Note that by~\Cref{res:dynamic} and the discussion after that, any such algorithm---if it exists---should necessarily
	exploit the ``insertion-only'' aspects of the model and cannot solely rely on techniques such as sampling and sketching that also apply to dynamic streams. 
	
	\item We can ask the above question even more concretely for $3$-colorable graphs and semi-streaming algorithms: 
	for what values of $Q$ (say, even as a function of $n$), can we distinguish between $3$-colorable and $Q$-colorable graphs in adversarial streams via $\Ot(n)$-space algorithms? Currently, the best bounds 
	we know only imply $Q = \Omgt(n^{1/2})$ by~\Cref{res:dynamic} which even hold for dynamic streams (and is optimal in that model). Is it possible to obtain an $n^{1/2-\eps}$-coloring of $3$-colorable graphs, for some constant $\eps > 0$, via 
	semi-streaming algorithms in adversarial streams? 
	
	\item Is the tradeoff of $n^{1+1/t}$ space for distinguishing between $q$-colorable and $q^t$-colorable graphs in random-order streams optimal? The second part of~\Cref{res:rand} addresses the ``base case'' of this tradeoff curve
	for $t < 2$, but can we prove this for larger values of $t$ as well (even a lower bound for a fixed choice of $t$ will be interesting)? Alternatively, are there better algorithms in this model? 
\end{enumerate}

Throughout this paper, we primarily focused on single-pass streaming algorithms. It will also be very interesting to examine the power of multi-pass algorithms, even if with just two or three passes, for~\Cref{question}. 

Finally, we introduced cluster packing graphs (\Cref{def:cpg}) as a generalization of RS graphs by replacing induced matchings in RS graphs with induced clusters of cliques. Similar to RS graphs, it will be very interesting 
to determine what range of parameters---clique size, cluster size, and number of clusters---do admit such graphs. We note that addressing this question has been notoriously challenging for RS graphs already with many basic questions 
left wide open. However, given cluster packing graphs are ``harder'' to construct---as in, they do generalize RS graphs---it is plausible that proving upper bounds on their density in principle can be easier than in the case of RS graphs. 
Alternatively, can one relate the density of cluster packing graphs, in a black-box way, to the density of RS graphs with similar parameters? (See a recent work of~\cite{Pratt25} that addresses the same question for a different generalization
of RS graphs).

\clearpage

\section{Technical Overview}\label{sec:techniques}

We now elaborate on the techniques behind each of our results individually. We note that the techniques underlying our results in each section are entirely independent---with the exception of the lower bound in~\Cref{res:rand}, which builds on ideas from~\Cref{res:adv}---and thus, readers may freely skip to any subsection that interests them the most.


\subsection{Adversarial Streams (\Cref{res:adv})}\label{sec:over-adv}

The standard way of proving (single-pass) streaming lower bounds is via (one-way) communication complexity: suppose we partition edges of a graph $G=(V,E)$ between two (or more) players whose goal is to decide if $\chi(G) \leq q$ or $\chi(G) \geq Q$. 
To do this, in some pre-determined order, each player sends a single message to the next, and the last player is responsible for outputting the answer. It is easy to show that the minimum length of messages by each player needed to solve the problem in the worst case, namely, the one-way communication complexity of the problem, lower bounds the space of single-pass streaming algorithms (\Cref{prop:stream-cc} and~\Cref{prop:adv-communication-to-stream}). 

\subsubsection{Warm-up: Two-player Communication Complexity} 

As a warm-up, we prove that distinguishing $q$- from $(q^2/4)$-colorable graphs needs $\Omega(n^2)$ communication when we only have two players, say, Alice and Bob (\Cref{thm:two-player}). 
This proof is based on constructing a simple family of graphs and then a reduction from the Index problem~\cite{KNR95}. 

For integers $r,k \geq 1$, define an \textbf{$\bm{(r,k)}$-cluster} as a collection of $r$ vertex-disjoint $k$-cliques. 
An \textbf{induced} $(r,k)$-cluster in a graph $G$ is defined as a subgraph of $G$ which is a $(r,k)$-cluster and moreover, there are no other edges between vertices of this cluster in $G$. 
The first step of the lower bound is to construct, for each constant $k \geq 1$, a graph $G_{cluster}$ which consists of $\Omega(n^2)$ distinct induced $(k,k)$-clusters. The construction
is elementary and relies on a combinatorial interpretation of the basic geometric fact that two distinct lines can only intersect on a single point. 

Having this construction, the lower bound can be proven easily. Fix $k \geq 1$ and consider the graph $G_{cluster}$ for this $k$ such that $G_{cluster}$ consists of $t = \Theta(n^2)$ distinct induced $(k,k)$-clusters. Create graph $G$ as input to Alice and Bob as follows (see~\Cref{fig:over-two-player}): 
\begin{itemize}
	\item \textbf{Alice}: for each $(k,k)$-cluster $C_i$, for $i\in [t]$, with probability half, drop all edges of $C_i$  and otherwise keep all of them;
	\item \textbf{Bob}: for a randomly chosen $\istar$ from $[t]$, insert all edges \emph{between} vertices of different $k$-cliques of the cluster $C_{\istar}$. 
\end{itemize}

\begin{figure}[t!]
  \centering
  \subcaptionbox{In this case, $\istar=4$ and both Bob's edges and the $(3,3)$-cluster $C_4$ in Alice's input form a $9$-clique. \label{fig:over-two-player-1}}[0.45\linewidth]{%
    \includegraphics[width=0.58\linewidth]{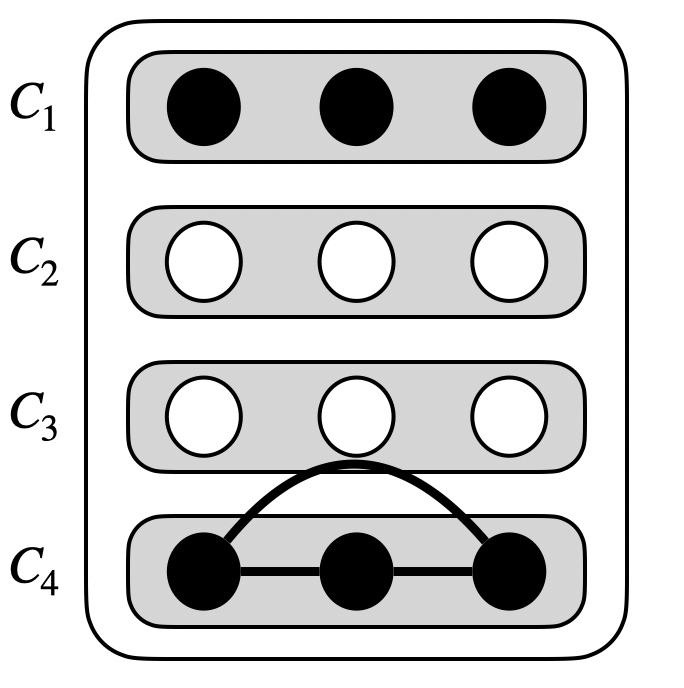}}
  \hspace{0.05\linewidth}
  \subcaptionbox{In this case, $\istar=2$ and the vertices of the $(3,3)$-cluster $C_2$ can be colored using $3$ colors as shown.\label{fig:over-two-player-1}}[0.45\linewidth]{%
    \includegraphics[width=0.615\linewidth]{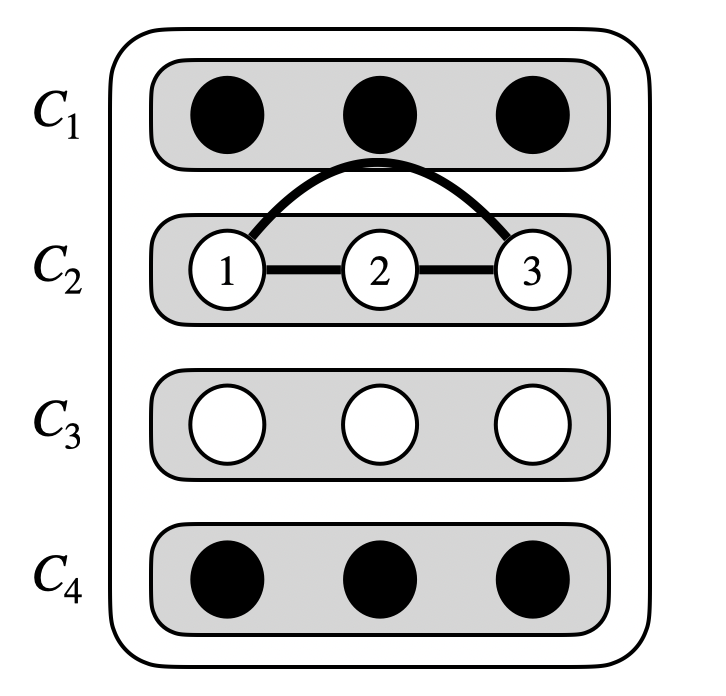}}
  \caption{An illustration of our two-player communication lower bound. Each circle denotes the vertices of a $3$-clique in $G_{cluster}$ and each small box (gray) denotes an induced $(3,3)$-cluster. The filled
  in circles mean the edges of those cliques are not dropped from Alice's graph and the empty circles shows the one that were dropped. The solid lines inside the cluster denote the edges of Bob between vertices of all cliques of the cluster.
  We shall note that this figure is not accurate in the sense of the clusters in $G_{cluster}$ are not necessarily vertex-disjoint (unlike what is depicted) as we have $\Omega(n^2)$ many of them.}
  \label{fig:over-two-player}
\end{figure}

If edges of the cluster $C_{\istar}$ for index $\istar$ of Bob are \emph{not} dropped from Alice's input, then this cluster plus Bob's edges form a $k^2$-clique in $G$, so $\chi(G) \geq k^2$. 
But, if these edges are dropped, then \emph{at least}, the cluster $C_{\istar}$ can be $k$-colored by coloring vertices of each \emph{original} clique with a different color. This is not enough to argue $\chi(G) \leq k$, 
but we can construct $G_{cluster}$ to be $k$-colorable also, and then color $G_{clustar} \setminus C_{\istar}$ with another $k$ colors, ensuring $\chi(G) \leq 2k$. By re-parameterizing $q = 2k$ and using 
standard lower bounds for the Index problem~\cite{KNR95}, we conclude the proof. 

\paragraph{Coloring vs clique.} By our description above, our lower bound holds also for the problem of deciding if $\chi(G) \leq q$ vs $\omega(G) \geq (q^2/4)$, where $\omega(G)$ is the size of the \emph{largest clique} in $G$. Previously,~\cite{HalldorssonSSW12} (see also~\cite{RashtchianW0Z21}) proved that deciding if $\omega(G) \leq q$ or $\omega(G) = 2^{\Omega(q)}$ requires $\Omega(n^2)$ communication even when Alice and Bob can communicate back and forth with each other. 
So, we can ask if their approach also works for us to prove something much stronger than our warm-up result. The answer is \emph{No} for multiple reasons: firstly, their graph constructions---which are random graphs---in both cases have effectively the same value of $\chi(G)$
and thus cannot be used as is for our purpose. But, more importantly, the $q$- vs $\Omega(q^2)$- bound in our warm-up result for two players is asymptotically optimal for a trivial reason: if both Alice and Bob's graphs are individually $q$-colorable, then their combination
will be $q^2$-colorable, and if one of them is not $q$-colorable, then that player can detect it on their own. Thus, there is no hope of obtaining results like~\cite{HalldorssonSSW12,RashtchianW0Z21} for the coloring problem (see also the next part for another reason as well). 


\subsubsection{Multi-player Communication Complexity}

As stated earlier, to prove~\Cref{res:adv}, we need to consider more than two players in our communication model; in fact, for the same reason, with $p$ players, distinguishing between $q$ vs $>q^p$-colorable graphs is again trivial.  
We will show that for any constant number of players, this bound is asymptotically optimal and build on this, with some more loss in the parameters, to prove~\Cref{res:adv}.

One of the most promising directions for proving \emph{multi}-player communication lower bounds (in the context of streaming algorithms) are reductions from the \emph{Promise Set Disjointness} problem (see~\cite{KamathPW21} and references therein). 
This problem also lower bounds the space complexity of \emph{multi-pass} streaming algorithms. 
However, as a corollary of~\Cref{res:rand}, we obtain a multi-pass streaming algorithm for graph coloring (\Cref{prop:multi-pass}) with much stronger guarantees than the lower bound we aim to prove in~\Cref{res:adv}. This suggests
proving~\Cref{res:adv} requires an ``inherently single-pass'' approach (this is another difference with the clique lower bounds of~\cite{HalldorssonSSW12,RashtchianW0Z21}). For such lower bounds, a candidate problem for reductions is the \emph{Chain} problem 
(see~\cite{Chakrabarti07,CormodeDK19,FeldmanNSZ20,Sundaresan25}), which is a somewhat direct
extension of the Index problem to multiple players. A reduction from chain has been used previously in~\cite{CormodeDK19} to prove a lower bound for $q$- vs $q^2$-coloring problem in single-pass streams (in vertex
arrival streams wherein our warm-up lower bound does not apply). Alas, chain also appears to be a too simplistic problem to allow us to for a stronger reduction 
needed to prove~\Cref{res:adv}. 

As \textbf{our main technical contribution}, we present a new approach for proving multi-player communication lower bounds on graphs that uses a combination of ideas from both Set Disjointness and Chain at the same time, 
together with a novel graph theoretic construction. 

Let us fix the number of players to be some \emph{constant} $p \geq 3$ and pick an integer $k \in \IN$. The input to players will be graphs $G_1,G_2,\ldots,G_p$ such that $\chi(G_i) \leq k$ for all $i \in [p]$. However, 
these graphs are put together in a way that their union $G=G_1 \cup \ldots \cup G_p$ can have very different chromatic numbers: in one case, we have $\chi(G) = O(k)$ whereas in the other case, even $\omega(G) \geq k^p$ (as before, 
our ``certificate'' of large chromatic number is a large clique). To picture how a combination of $p$ $k$-colorable graphs can create a $k^p$-clique, it helps to think of vertices of the $k^p$-clique as being in a discrete $p$-dimensional \emph{grid} on $[k]$, namely, $[k]^p$,
and each graph $G_i$ is responsible for connecting edges of this clique along one dimension. See~\Cref{fig:k-kp} for an illustration. 

\medskip

\begin{figure}[h!]
  \centering
  \subcaptionbox{A $27$-clique (many edges are omitted here).\label{fig:k-kp}}[0.45\linewidth]{%
    \includegraphics[width=1\linewidth]{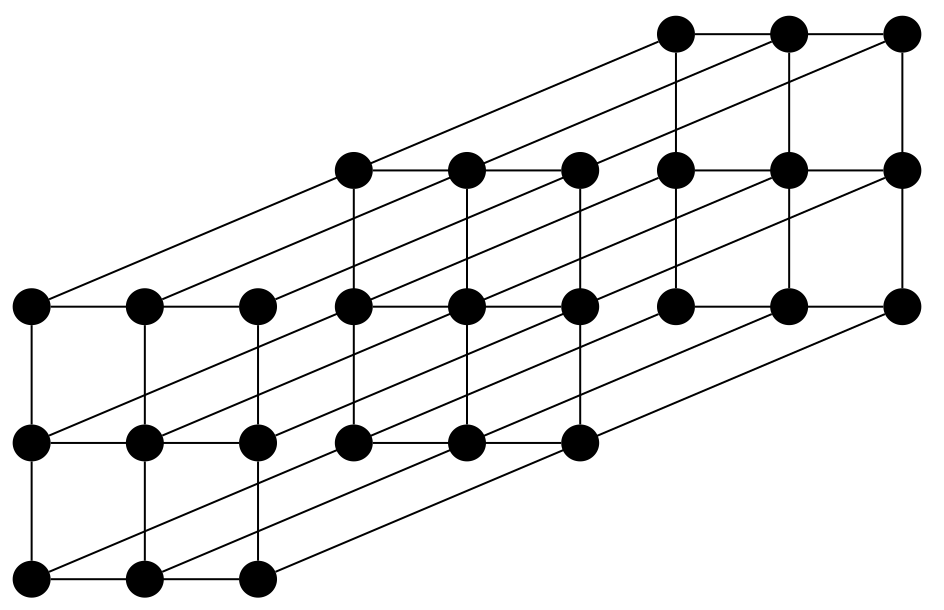}}
  \hspace{5pt}
  \subcaptionbox{A $3$-colorable graph of edges in the $x$-dimension. 
  \label{fig:k-kp-x}}[0.45\linewidth]{%
    \includegraphics[width=1\linewidth]{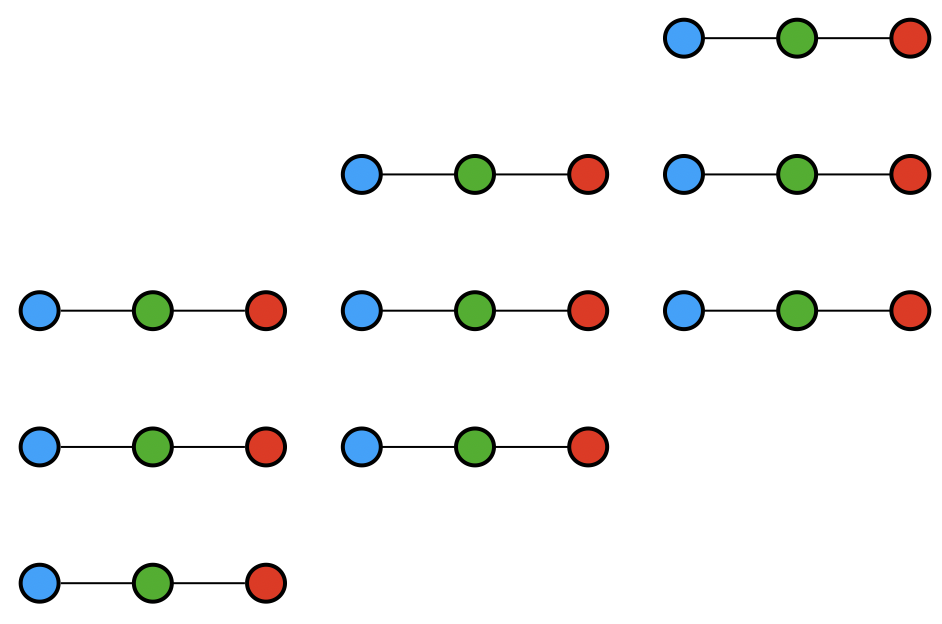}}
    
    \vspace{10pt}
    
  \subcaptionbox{A $3$-colorable graph of edges in the $y$-dimension. \label{fig:k-kp-y}}[0.45\linewidth]{%
    \includegraphics[width=1\linewidth]{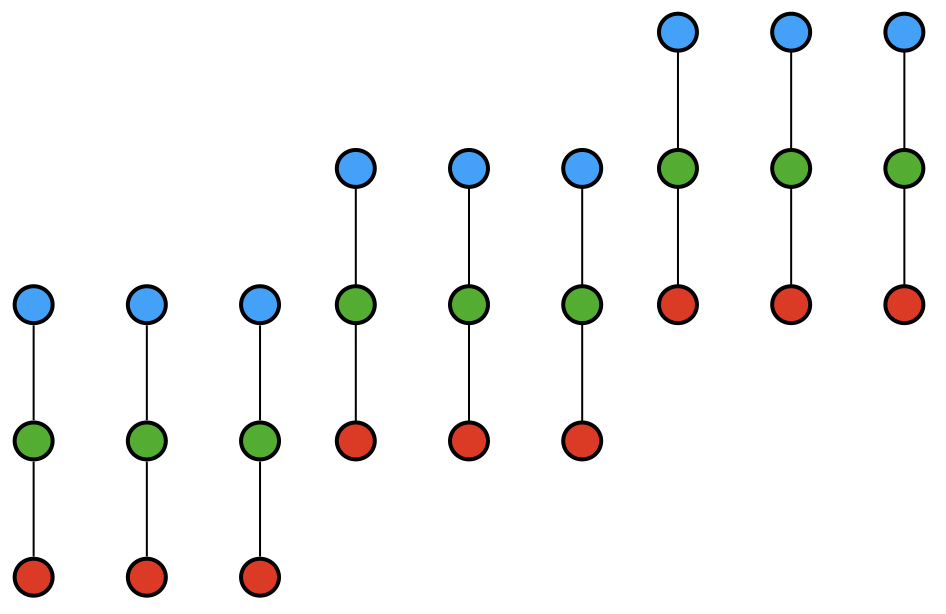}}
  \hspace{5pt}
  \subcaptionbox{A $3$-colorable graph of edges in the $z$-dimension. 
  \label{fig:k-kp-z}}[0.45\linewidth]{%
    \includegraphics[width=1\linewidth]{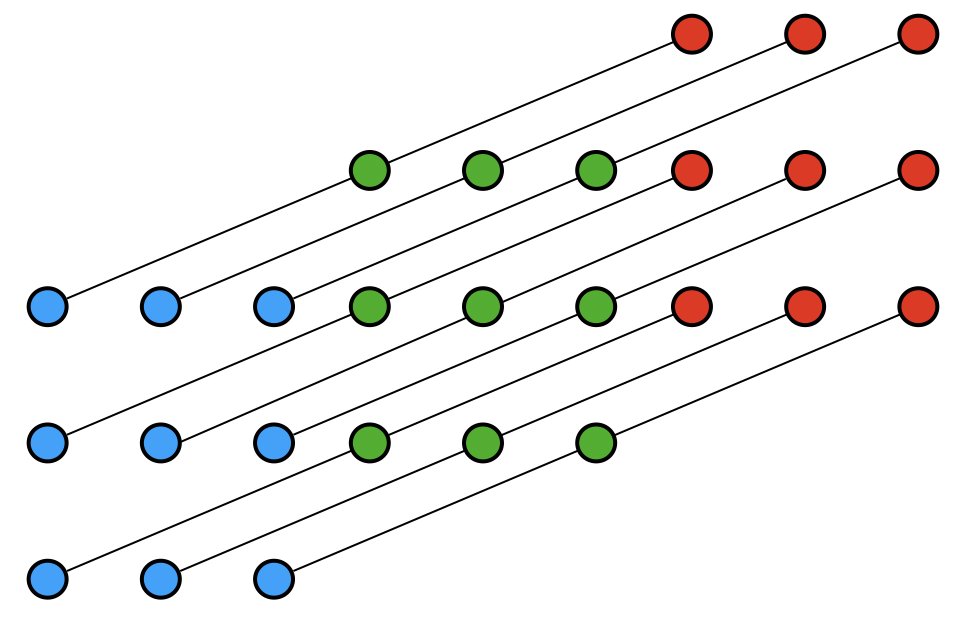}}  
  \caption{An illustration of forming a $k^p$-clique by combining $p$ $k$-colorable graphs according to a discrete $p$-dimensional grid; here, $k=p=3$. To avoid cluttering the figures, not all edges of the clique are shown. }
  \label{fig:k-kp}
\end{figure}

In our lower bound instances, there is a $p$-dimensional grid and each player's graph is responsible for providing edges alongside one dimension of this grid. In one case, all edges
of this grid appear in the inputs of players, and in the other case, none of the edges are there. The edges in each dimension of this grid are \emph{hidden} in an $n$-vertex $k$-colorable graph provided to the corresponding player 
as input. To be able to establish the lower bound, we need to ensure that: 
\begin{enumerate}
	\item the $p$-dimensional grid remains hidden from each player even after receiving messages of prior players (as otherwise they can check if any edge inside the grid appear in the graph or not); 
	\item no player should be able to detect if any edge in the grid has appeared in the graph or not.
\end{enumerate}
In addition, for this lower bound to be applicable to our original coloring problem, we also need:
\begin{enumerate}
	\item[3.] $G_1 \cup \ldots \cup G_p$ should have a ``small'' chromatic number, say, $O(k)$, in the absence of grid's edges. 
\end{enumerate}

We address these by designing a new family of graphs---called \textbf{cluster packing} graphs---that vastly generalize our warm-up construction (addressing challenge 3), 
using ``Set Disjointness type'' arguments for hiding the grid (addressing challenge 1), and using ``Chain type'' arguments for hiding existence of edges of the grid's from any player (addressing challenge 2). 
We discuss these ideas briefly in the following. 

\paragraph{Cluster packing graphs.} Recall \emph{induced} $(r,k)$-clusters from earlier. Based on this, we define: 
\begin{itemize}
	\item For integers $r,t,k \geq 1$, we say a graph $G=(V,E)$ is an \textbf{$\bm{(r,t,k)}$-cluster packing graph} iff its edges can be partitioned into $t$ distinct induced $(r,k)$-clusters $C_1,\ldots,C_t$. 
\end{itemize}
See~\Cref{fig:cpg} for an example of a cluster packing graph. 

\medskip

\begin{figure}[h!]
  \centering
  \subcaptionbox{\label{fig:cpg-1}}[0.32\linewidth]{%
    \includegraphics[scale=0.35]{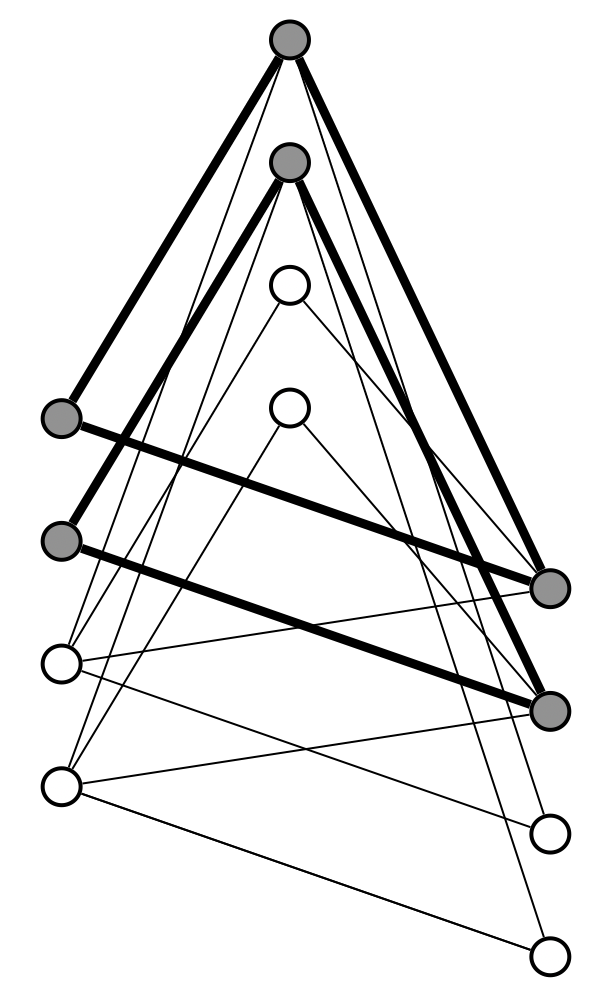}}
  \hspace{1pt}
  \subcaptionbox{ 
  \label{fig:cpg-2}}[0.32\linewidth]{%
    \includegraphics[scale=0.35]{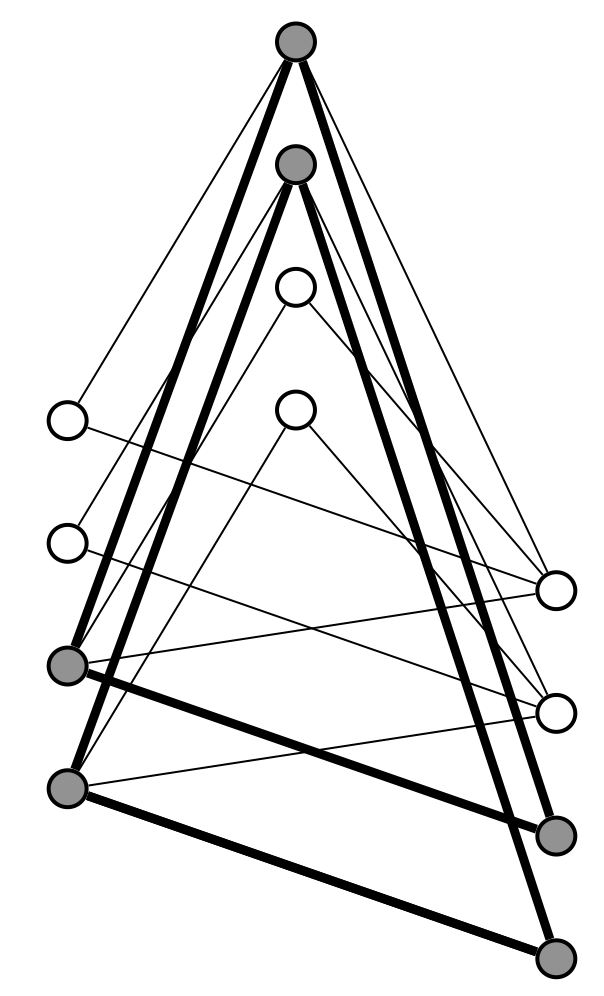}}
    \hspace{1pt}
  \subcaptionbox{ \label{fig:cpg-3}}[0.32\linewidth]{%
    \includegraphics[scale=0.35]{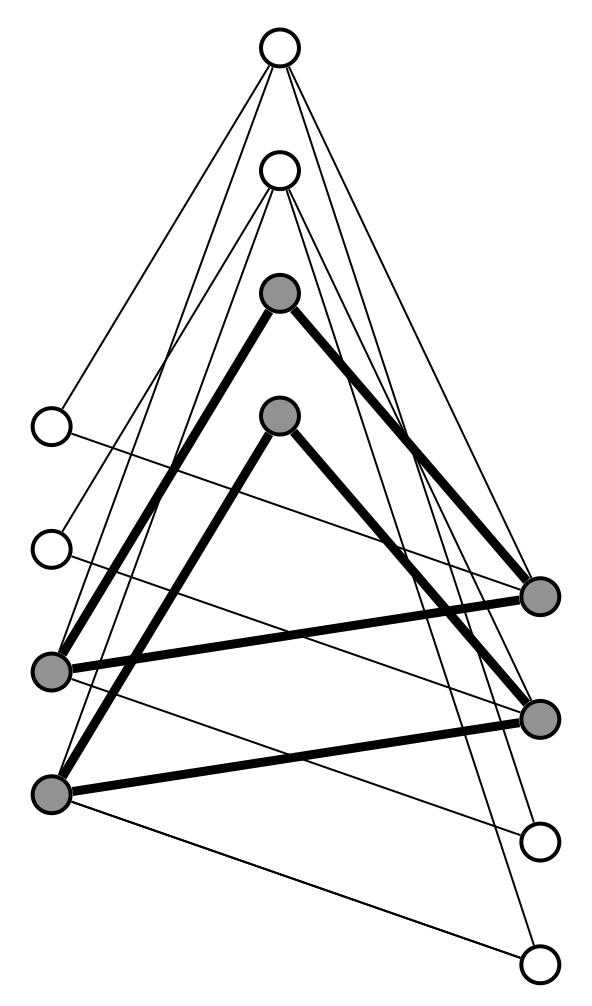}}
  \caption{A $(2,3,3)$-cluster packing graph ($r=2$, $t=3$, $k=3$). Each one of the three $(2,3)$-clusters are shown separately; notice that there are no 
  other edges between vertices of each cluster (they are induced).}
  \label{fig:cpg}
\end{figure}

Our warm-up lower bound could work with a simple cluster packing graph with parameters $r=k$ and $t=\Omega_k(n^2)$. However, for our multi-player lower bound, 
we need these graphs when $r$ is quite large, at least $n^{1-o(1)}$ or even $r=\Theta_k(n)$, while keeping the graph almost dense (or at least not too sparse); we will describe this necessity later.  
The existence of such a graph is rather counter intuitive apriori given that induced $(r,k)$-clusters for large $r$ and small $k$ form highly \emph{sparse} induced subgraph in the graph, and yet 
we need the entire graph to be \emph{dense}. 

Cluster packing graphs generalize \rs (RS) graphs~\cite{RuzsaS78} that in recent years have become a crucial building block for various graph streaming lower bounds among many other applications; see~\Cref{app:cpg-parameters}. 
An $(r,t)$-RS graph is a graph whose edges can be partitioned into $t$ \emph{induced matchings}, each of size $r$. 
Thus, an $(r,t)$-RS graph is a $(r,t,2)$-cluster packing graph, and generally, cluster packing graphs can be seen as ``switching'' edges of RS graphs with larger cliques. 

A  similar generalization of this type from RS graphs is implicit in~\cite{AssadiKNS24},\footnote{In fact, the graphs in~\cite{AssadiKNS24} switch edges of RS graphs with arbitrary graphs, but while we focus on constant-size cliques (parameter $k$), 
they focus on much larger subgraphs (say, $n^{1/3}$ vertices) but instead require the subgraphs to have small chromatic number.} and we can use their results to construct cluster packing graphs with both $r , t = n^{1-o(1)}$ for any $k=n^{o(1)}$. We do this for proving the first part of~\Cref{res:adv} but for the second part of this result, we need cluster packing graphs with $r=\Theta_k(n)$. A key contribution of our work is to construct such graphs with parameter $t=n^{\Omega(1/\log\log{n})}$ 
by generalizing RS graph constructions of~\cite{FischerLNRRS02} with $r=\Theta(n)$ to cluster packing graphs also. We postpone the details of this construction to~\Cref{sec:stronger-cluster-packing} since, while short, it is quite technical. 

We believe cluster packing graphs can be of their own independent interest and find applications elsewhere, specifically to other graph streaming lower bounds; as such, in~\Cref{app:cpg}, 
we study further aspects of cluster packing graphs---that are orthogonal to the topics of our paper---in hope of shedding more light on them. 

\paragraph{The hard input distribution.} Equipped with cluster packing graphs, we design our hard input distribution as follows. 
We start by providing the first player with a $(r,t,k)$-cluster packing graph $G_1$ (for either $r,t = n^{1-o(1)}$ for proving a $n^{2-o(1)}$ communication lower bound for a more restricted range of $k$, 
or $r=\Theta_k(n)$ and $t = n^{\Omega(1/\log\log{n})}$ for a wider range for $k$). We then randomly pick subsets $S_1,\ldots,S_t$ from the $(r,k)$-clusters $C_1,\ldots,C_t$ of $G_1$, each containing a constant fraction of the cliques.  
We drop edges of all $k$-cliques \emph{outside} these subsets, and in each subset, we also drop edges of a random half of the $k$-cliques entirely. 

The input of the remaining players is then constructed as follows. We pick $\istar \in [t]$ randomly to point to one of the $(r,k)$-clusters $C_{\istar}$ of $G_1$ and then pick a set $T$ of $k$-cliques from this cluster; the choice of $T$ is such
that $S_{\istar} \cap T$ has precisely $k^{p-1}$ different $k$-cliques, and moreover, either edges of all these $k$-cliques are dropped from $G_1$ or none of them are. See~\Cref{fig:adv-dist} for an illustration. 

\medskip

\begin{figure}[h!]
  \centering
    \includegraphics[width=1\linewidth]{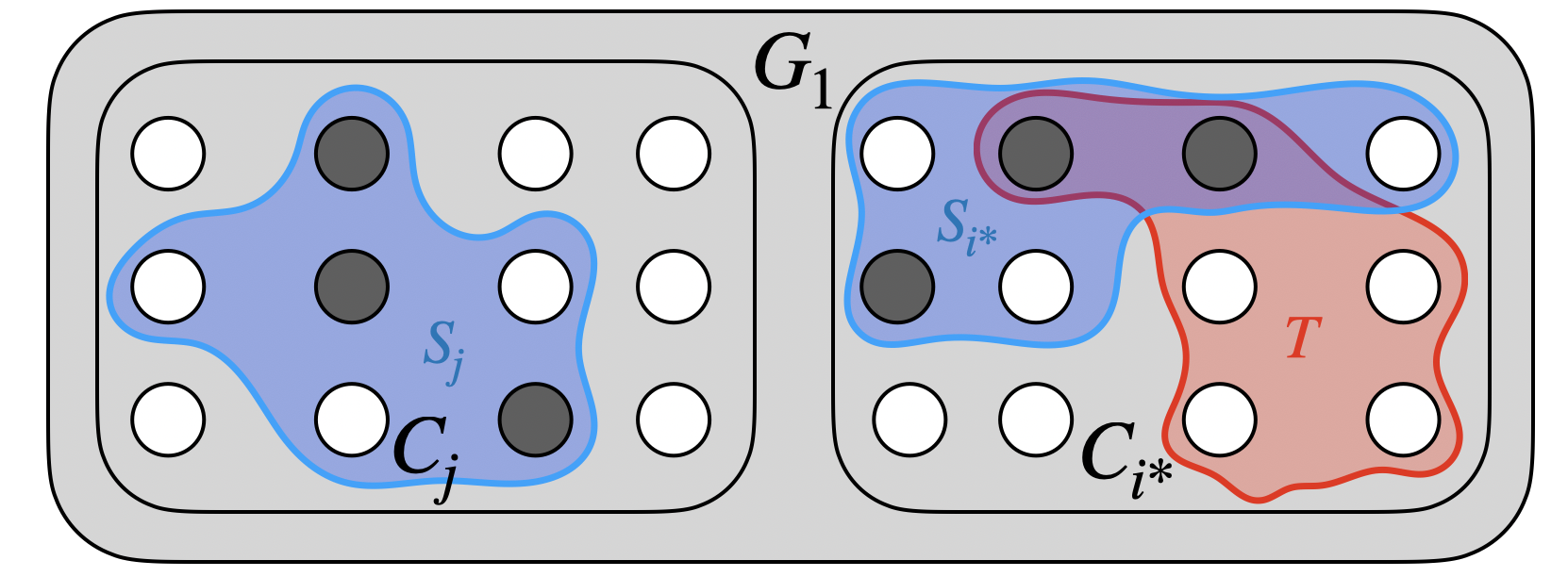}
  \caption{An illustration of input of the first player in our multi-player communication lower bound. Here, the bigger boxes show two different induced clusters in the graph $G_1$. Each circle is a $k$-clique in the cluster
  and empty circles mean the edges of the clique are dropped from $G_1$ while filled in one means the clique is present. The graph $G_1$ can only have remaining cliques inside sets $S_1,\ldots,S_t$ for each induced cluster.
  The set $T$ that defines vertices of the subsequent players is also depicted. After player one, all other players only have edges inside the set $T$. 
  }
  \label{fig:adv-dist}
\end{figure}

We think of the vertices in $S_{\istar} \cap T$ as the $p$-dimensional grid we talked about earlier. This way, the first player can handle one dimension of this grid. 
To create the input of the subsequent players, we abstractly think of contracting each $k$-cliques inside $T$ to become a single vertex; and then recursively
create a hard instance for the remaining players on vertices of $T$ with the following promise: \emph{all} subsequent players also use vertices of $S_{\istar} \cap T$ as their ``grid vertices'', and that, either all 
players have all required edges inside the grid (in the corresponding dimension) or none of them have. This allows us to argue that when the ``grid edges'' are present, there is a $k^p$-clique inside the entire graph. 
When these edges are not present, we can show that we can color the input of each player with a \emph{new} set of $k$ colors---by ensuring that $(r,t,k)$-cluster packing graphs are $k$-colorable---and
thus obtain a $(k \cdot p)$-coloring of the graph. See~\Cref{subsec:adv-construct} for a formal description of this distribution. 

\paragraph{Analysis of hard instances.}  The analysis of these hard instances uses various tools developed in graph streaming lower bounds and beyond such as 
``almost solving'' Set Disjointness~\cite{AssadiCK19b,AssadiR20,AssadiGLMM24}, ``elimination arguments'' for eliminating communicated messages and simulating them instead, e.g.,~\cite{MiltersenNSW95,GuruswamiO13,AssadiKNS24,AssadiBKNS25}\footnote{We note that majority of these work are applying this idea to \emph{round} elimination wherein the goal is to eliminate one round of a multi-round protocol (and prove lower bounds for multi-pass algorithms). However, we use these ideas for a \emph{player} elimination by removing players one at a time to prove the lower bound. Player elimination (a.k.a. party elimination) ideas have also appeared in the past although less frequently, 
e.g., in~\cite{Chakrabarti07,ViolaW07}.}, and ``direct-sum type'' arguments using information complexity~\cite{ChakrabartiSWY01,Bar-YossefJKS02,BarakBCR10}. 


At a high level, the analysis follows
both the analysis of the Promise Set Disjointness problem (to keep track of the vertices of the abstract grid as only ones ``shared'' between the players) and the Chain problem (to keep track of how much information is revealed to players
on whether edges of the abstract grid are all present or all absent). Specifically, we first prove that after the first player's message---assuming it is small enough---, the location of the special grid, namely, the identity of the set $S_{\istar} \cap T$ is 
still almost uniform over its support from the perspective of the second player. This argument is similar to~\cite{AssadiR20} but needs to additionally take into account that the support of this distribution is quite restricted (given the ``rigid'' structure of cluster packing graphs). We then argue that the first message---again, if small enough---cannot reveal much on whether the edges in the cliques in $S_{\istar} \cap T$ are all present or they are all absent. This argument is similar to the lower bound
for the Index problem~\cite{KNR95}. 

Putting the above two results together, we argue that the message of the first player is neither correlated much with the identity of $S_{\istar} \cap T$ nor with whether edges inside this set are present or absent. We use this in a \emph{player elimination} argument: given an instance $I_{p-1}$ of the $(p-1)$-player problem, we can embed this instance on vertices of the set $T$ (after proper expansion of each vertex of $I_{p-1}$ to become a $k$-clique), independently \emph{sample} the 
message of an artificial first player---using the low correlation argument above---, and then run a $p$-player protocol from its second player on this instance. This allows us to eliminate the first player of the protocol
and recurse like this until we get to a two-player protocol and use our warm-up argument. These parts can be seen as generalizations of 
some prior work on the Chain problem, e.g., in~\cite{Chakrabarti07,FeldmanNSZ20}. We note that since in this step we are recursing on instances with vertices solely in $T$ which is of size at most $\Theta_k(r)$ (in the cluster packing graph), 
we need the value of $r$ to be as close as possible to $n$; in other words, the ratio of $r/n$ determines how many times we can repeat these recursive steps and thus governs the choice of the number of players $p$ in the lower bound.  

Putting all these together, we can prove that for any $p \geq 2$, distinguishing $(k \cdot p)$-colorable from $k^p$-colorable graphs via $p$-player protocols requires a large communication. 
For our multi-player communication problem for \emph{constant} $p \geq 2$, this gives us the desired lower bound between $O(k)$- vs $\Omega(k^p)$-colorability. For establishing~\Cref{res:adv} for streaming algorithms, 
we instead set $k = \Theta(1)$ and let $p \rightarrow \infty$
and obtain a lower bound for $q$- vs $2^{\Omega(q)}$-colorability for $q = k \cdot p$. 

This concludes the high level overview of the proof of~\Cref{res:adv}. The formal proof is  presented in~\Cref{sec:two-player} (the warm-up),~\Cref{sec:stronger-cluster-packing} (cluster packing graphs), and~\Cref{sec:adversarial} (main lower bound). 

\subsection{Random Order Streams (\Cref{res:rand})}\label{sec:tech-rand}

\subsubsection{The Algorithm}
Our algorithm in~\Cref{res:rand} is similar in spirit to the ``sample and solve'' framework of streaming algorithms~\cite{LattanziMSV11,KumarMVV13}. Unlike most prior applications of this technique that are for multi-pass algorithms, we apply this to random order streams (although it also does have multi-pass implications; see~\Cref{prop:multi-pass}). The algorithm, at a high level, is as follows: 
\begin{Algorithm}\label{alg:over-rand}
Given a graph $G=(V,E)$ and integers $q,t \geq 2$, decide $\chi(G) \leq q$ or $\chi(G) > q^t$: 
	\begin{enumerate}
	\item Read and store the first $\Ot(n^{1+1/t})$ edges of the stream, and call it subgraph $H_1$. Check if $\chi(H_1) > q$
	and if so output `large' chromatic number for $G$.\footnote{This step takes exponential time but can be done in the same space as that of storing the subgraph.} Otherwise, let $C_1$ be a $\chi(H_1)$-coloring of $H_1$ and continue. 
	\item Read and store the next $\Ot(n^{1+1/t})$ edges of the stream that are \emph{monochromatic} under $C_1$, and call it subgraph $H_2$.  Again, check if $\chi(H_2) > q$ and define $C_2$ as a coloring of $H_2$. 
	\item Continue as before by storing edges that are monochromatic under \emph{both} $C_1$ and $C_2$, and go on like this by finding subgraphs $H_3,H_4,\ldots,$ and corresponding colorings $C_3,C_4,\ldots$; output `large' if $\chi(H_i) > q$ for 
	any $H_i$, otherwise, output `small' at the end of the stream. 
	\end{enumerate}
\end{Algorithm}

The space complexity of~\Cref{alg:over-rand} can be bounded by $\Ot(n^{1+1/t})$ by reusing the space for storing each $H_i$'s. It is also easy to see that if $\chi(G) \leq q$, the algorithm always outputs correctly (as each $H_i \subseteq G$, we have 
$\chi(H_i) \leq \chi(G) \leq q$ in each step). Thus, it remains to show that if $\chi(G) > q^t$, then, the algorithm actually outputs `large' with high probability. 

The proof is by showing that if the algorithm outputs `small', it actually has found a $q^t$-coloring of $G$ with high probability. To prove this, we show that with high probability, 
the number of subgraphs $H_1,H_2,\ldots,$ found by the algorithm is upper bounded by $t$. This in turn implies that if each of these subgraphs is $q$-colorable and no edge is monochromatic under all of their colorings \emph{simultaneously}, 
then $G$ can be $q^t$-colored by taking the product of the (at most) $t$ different $q$-colorings $C_1,C_2,\ldots,C_t$. The main part---and the step that is reminiscent of sample and solve framework---is an edge sampling lemma
that states that proper coloring of random subgraphs is also an ``almost proper'' coloring of the entire graph: 
\begin{itemize}
	\item[] \textbf{Edge Sampling Lemma (informal):} Given a graph $G=(V,E)$, if $H$ is subgraph of $G$ by sampling $\gtrsim n^{1+1/t}\log{n}$ \emph{edges} uniformly at random and $C$ is a proper coloring of $H$,
	then, with high probability, $C$ has $\lesssim \card{E} \cdot n^{-1/t}$ monochromatic edges in $G$. 
\end{itemize}
We can then use the randomness of the stream to argue in each step, $H_i$ is chosen randomly from existing monochromatic edges and apply the above lemma repeatedly to argue the algorithm finish reading the entire stream
before needing to pick more than $t$ subgraphs. The proof of the lemma itself is a simple probabilistic analysis by union bounding over all potential proper colorings of $H$.  

\subsubsection{The Lower Bound} 

Our lower bound in~\Cref{res:rand} is a simple adaptation of our two-player lower bound from the warm-up of the previous subsection. This is done by working with the \emph{robust} communication complexity of the problem~\cite{ChakrabartiCM08} instead, 
which is shown to even lower bound the space of random order streaming algorithms. We create the input graph of players exactly as before, but now we partition each edge independently and uniformly between the two players. 
Using a robust communication lower bound of~\cite{ChakrabartiCM08} for a generalization of the Index problem, 
we can argue that even under this random partitioning, $\Omega(n^2)$ communication is needed by the players to solve the problem for constant $q$. The implication for streaming algorithms then follows immediately from~\cite{ChakrabartiCM08}. 

\subsection{Dynamic Streams (\Cref{res:dynamic})}\label{sec:tech-dynamic}

\subsubsection{The Algorithm} 
Our dynamic stream algorithm follows from a purely combinatorial lemma that we prove: if we sample vertices of a graph with a ``large'' chromatic number, then, the chromatic number of the sample
cannot be ``too small'', as in, reduce ``too much'' below than the sampling probability: 
\begin{itemize}
	\item[]\textbf{Vertex Sampling Lemma (informal):} Given a graph $G=(V,E)$, if $H$ is subgraph of $G$ by sampling each \emph{vertex} of $G$ independently with probability $p \in (0,1)$, 
	then, $\chi(H) \gtrsim \dfrac{p}{\ln{n}} \cdot \chi(G)$ with some constant probability. 
\end{itemize}
\noindent
The proof of this lemma is more elaborate than our edge sampling lemma and goes as follows. Suppose the lemma is not true and sample $k \lesssim (\ln{n})/p$ subgraphs $H_1,\ldots,H_k$
independently as in the lemma. Define $\event$ as the event that for all $i \in [k]$, $\chi(H_i)$ is much smaller than $(p/\ln{n}) \cdot \chi(G)$. By picking the leading constant in the definition of $k$ properly, we will have: 
\begin{itemize}
	\item on one hand, \emph{without} conditioning on $\event$, we can argue that with large enough probability, every vertex of $G$ will appear in some subgraph $H_1,\ldots,H_k$; 
	\item on the other hand, \emph{with} conditioning on $\event$, we can color all vertices of $G$ that appear in some subgraph $H_1,\ldots,H_k$ with $k \cdot (p/\ln{n}) \cdot \chi(G) < \chi(G)$ colors without creating monochromatic 
	edges between those vertices. 
\end{itemize}
Our goal is to show that with non-zero probability, the conclusion of both steps above happen simultaneously, which mean we find a proper coloring of the entire $G$ with less than $\chi(G)$ colors, a contradiction. The crux of the analysis 
is to show that conditioning on $\event$ cannot change the \emph{marginal} sampling probability for each vertex by much---since inclusion or exclusion of a single vertex cannot alter chromatic number of
a graph by more than one---and use this to relate the probability distributions of both parts above to each other. 

Obtaining the algorithm in~\Cref{res:dynamic} is now straightforward. Sample $\Ot(n/t)$ vertices of the graph at the beginning of the stream, store all edges between them---by maintaining a counter for each pair and counting
insertions and deletions of that pair explicitly---; at the end of the stream, compute the chromatic number of this subgraph and return `large' iff it is larger than $q$. The correctness follows from the vertex sampling lemma since when $\chi(G) \geq q \cdot t$, 
it is unlikely that the sampled subgraph has chromatic number as low as $q$, concluding the proof. 

\subsubsection{The Lower Bound}

Our lower bound is proven using a ``sketching characterization'' of dynamic streaming algorithms due to~\cite{LiNW14,AiHLW16} which allows one to use \emph{simultaneous} communication lower bounds 
to prove dynamic stream lower bounds (see~\Cref{rem:sim-dynamic}). In the simultaneous model, the edges of the input graph are partitioned between $p$ players, but unlike in~\Cref{sec:over-adv}, the players no longer can communicate with each other; 
instead, they all, \emph{simultaneously} with each other, send a message to a referee who will then output the answer based on the received messages. 

The communication lower bound follows the pattern of prior work, e.g., in~\cite{Konrad15,AssadiKLY16,AssadiKL17}, by differentiating between \emph{local} view of the players versus \emph{global} structure of the graph. 
Suppose we have $\binom{k}{2}$ players for some integer $k$, and locally, the input of each player is a bipartite graph. However, globally, in input of each player, there is a fixed pair of \emph{special} vertices and these pairs, 
across the $\binom{k}{2}$ players, ``cover'' all edges of a $k$-clique. The remaining vertices of the players, in the global structure of the graph, follow a consistent bipartition across the players (this way, the input can become 
a multi graph, not a simple one). See~\Cref{fig:over-sim} for an illustration. 

\medskip

\begin{figure}[h!]
  \centering
  \subcaptionbox{\emph{Local} view of the six players. We note that the identity of ``special'' vertices (the differently colored vertices) are not apparent to the players locally. Whether or not special vertices are connected
  to each other are coordinated, but all other edges appear independently across the players.\label{fig:sim-local}}[1\linewidth]{%
    \includegraphics[width=1\linewidth]{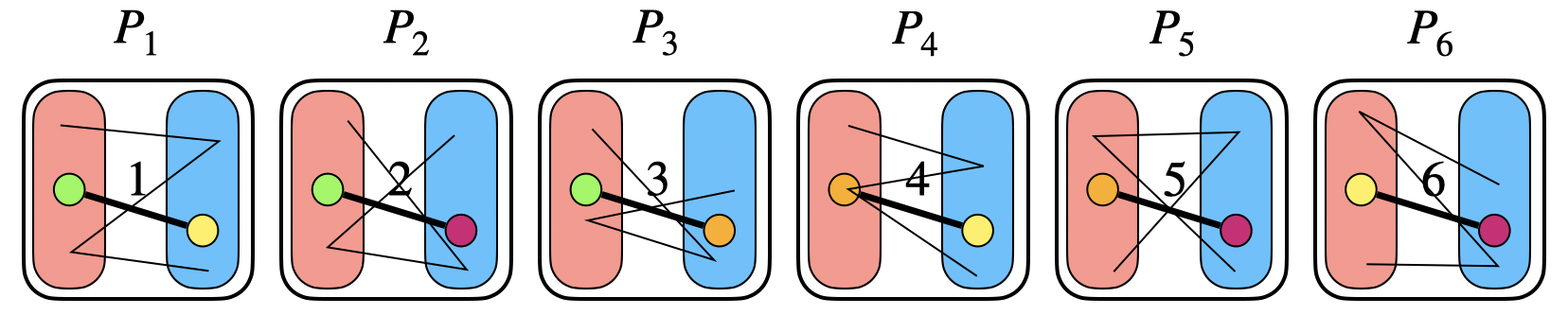}}
  \\
  \subcaptionbox{\emph{Global} structure of the multi-graph. The two circles in the left denote the missing ``special'' vertices. There will be many edges between the bipartite ``non-special'' vertices, while the 
  special vertices either form a $4$-clique or an empty graph (there are also edges between special and non-special vertices, thus the entire graph is never bipartite.  
  \label{fig:sim-global}}[1\linewidth]{%
    \includegraphics[width=0.5\linewidth]{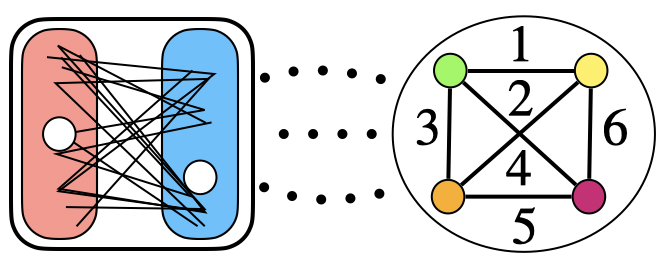}}
  \caption{An illustration of our simultaneous communication lower bound for $k=4$ and $\binom{4}{2} = 6$ players.}
  \label{fig:over-sim}
\end{figure}

In these graphs, with probability half, \emph{all} special vertices are connected to each other, and so form a $k$-clique in the global structure. With the remaining probability half, 
there are \emph{no} edges between special vertices---in this case, the global structure will become a $3$-colorable graph by coloring the bipartite subgraph with two colors, and all the special vertices with a third. 

The intuition behind the analysis is that the players are oblivious to the identity of their special vertices; in other words, in their \emph{local} view, any pair of vertices (respecting the bipartition), have the same 
probability of being the special pair. As such, it is unlikely that the message of a single player, if short, reveals much information about the connection of special pair of vertices in their input. We can then
exploit certain independencies between inputs of players to argue that the information revealed by them about such connections is \emph{sub-additive} across them. The conclusion is 
that their combined short messages still does not allow the referee to determine if special vertices are connected to each other or not; hence, the referee 
cannot distinguish  $\chi(G) \leq 3$ from $\chi(G) \geq k$ with small messages from the players. The proofs in this part uses elementary information theory facts that by-now are standard in streaming lower bounds. 







\section{Preliminaries}

\paragraph{Notation.}
We use $[n]$ to denote the set $\{1, 2, \ldots, n\}$ for any positive integer $n$. We use $\log(\cdot)$ to denote the logarithm base $2$ and $\ln{(\cdot)}$ to use the natural logarithm. 
For any string $x \in \{0,1\}^{\ell}$ for integer $\ell \geq 1$, we use $x_i$ to denote the bit at position $i \in [\ell]$. For any matrix $x \in \{0,1\}^{\ell \times m}$, we use $x_{i,j}$ for $i \in [\ell], j \in [m]$ to denote the bit at row $i$ and column $j$ inside $x$. We also use $x_i \in \{0,1\}^{m}$ to denote the string of length $m$ that row $i$ of matrix $x $ contains for $i \in [\ell]$.

\paragraph{Graph notation.}
For any graph $G = (V, E)$, and any subgraph $H$ of $G$, we use $V(H)$ to denote the vertex set of $H$. For any vertex set $S \subseteq V$, we use $G[S]$ to denote the subgraph of $G$ induced by the vertex set $S$.

\paragraph{Information theory notation.} Throughout this text, we use information-theoretic tools extensively in all our lower bounds. We use \textsf{sans-serif} to denote random variables. Sometimes, we use bold font for random variables also, and this will be clear from context.  We use $\en{\rA}$ and $\mi{\rA}{\rB}$ to denote Shannon entropy and mutual information, respectively. Moreover, 
$\tvd{\rA}{\rB}$ and $\kl{\rA}{\rB}$ denote the total variation distance and KL-divergence between the distribution of corresponding random variables. Finally, we use $H_2(\eps)$ to denote the binary entropy function; that is, for every $\eps \in (0,1)$,
\[
	H_2(\eps):= \eps \cdot \log(\frac1{\eps}) + (1-\eps) \cdot \log(\frac1{1-\eps}).
\]
All the basic definitions and standard results from information theory that we use can be found in \Cref{app:info}.

\paragraph{Communication models.} Our paper is comprised of multiple results in different streaming models, and we use different models of communication to prove them correspondingly.  
To make each section as self-contained as possible, we defer the specific communication model and impossibility results used to these sections themselves (see \Cref{sec:two-player-background,subsec:p-player-back,sec:sim}).

	\clearpage
	

\section{Warmup: Two-Player (One-Way) Communication Model}\label{sec:two-player}

As is the case for most streaming problems, a good starting point for understanding the complexity of the problem at hand is to 
consider it first in the communication complexity model. In this section, we first define this model and state its connection to streaming, and then present and prove our results in this model, as a 
warm-up for our main arguments in subsequent sections. 

\subsection{Background on the Model}\label{sec:two-player-background} 

In the two-player one-way communication model, there are two players, Alice and Bob, with inputs $x\in \mathcal{X}$ and $y \in \mathcal{Y}$, respectively. The goal is to compute
$f(x,y)$ for some known function $f$ on the domain $\mathcal{X} \times \mathcal{Y}$. To do this, Alice is allowed to send a single message $M(x)$ to Bob 
and Bob should output the answer (hence, the term `one-way' in the model). The players have access to a shared tape of randomness, referred to as \textit{public randomness}, in addition to their own \emph{private randomness}.
We use $\pi$ to denote the protocol used by the players to compute the message and the answer. 

\begin{definition}
    For any protocol $\pi$, the \textnormal{\textbf{communication cost}} of $\pi$, denoted by $CC(\pi)$, is defined as the worst-case length of message $M$ in bits, communicated from Alice to Bob on any input.
\end{definition}

The goal in this model is to design protocols for a problem with minimal communication. The following standard proposition---dating back to~\cite{AlonMS96}---relates the communication cost of protocols and space complexity of streaming algorithms. 

\begin{proposition}\label{prop:stream-cc}
	Any $s$-space streaming algorithm $A$ for computing $f(x,y)$ on the stream $x \circ y$, namely, concatenation of $x$ with $y$ (with elements $x$ and $y$ being ordered arbitrarily), implies a communication
	protocol $\pi$ with $CC(\pi) = s$ with the same success probability. 
\end{proposition}
\begin{proof}
	Alice runs $A$ on her input $x$ and communicates the memory state of the algorithm as the message $M(x)$ to Bob. Since $A$ is a streaming algorithm, Bob can continue running it on $x \circ y$ 
	using only its memory state at the end of $x$ and output the same answer as $A$. Communication cost of this protocol is the same as the memory size of the streaming algorithm and this is a lossless simulation in success probability, 
	concluding the proof. 
\end{proof}

\Cref{prop:stream-cc} allows us to prove streaming lower bounds by proving communication cost lower bounds for the same problem. 

\paragraph{Index problem.} We also have the following canonical problem in this model that we use in our reduction. 
Alice receives a string $x \in \set{0,1}^m$ and Bob receives an index $i \in [m]$. The goal is to compute $x_i \in \set{0,1}$, i.e., compute $f: \set{0,1}^{m} \times [m] \rightarrow \set{0,1}$ where $f(x,i) = x_i$. 
\begin{proposition}[\!\!\cite{KNR95}]\label{prop:INDEX-lowerbound}
    Any randomized protocol $\pi$ for the Index problem on $\set{0,1}^m$ that succeeds with probability at least $2/3$ has communication cost $CC(\pi) = \Omega(m)$ bits. 
\end{proposition}
    
\subsection{Our Results in the Two-Player Model}\label{sec:two-player-res}

We study the following graph coloring problem in the two-player communication model: Alice and Bob receive a partition of the edges of an $n$-vertex graph $G=(V,E)$ and two integers $q < Q$, and their goal is to decide $\chi(G) \leq q$ (`small' case) or $\chi(G) \geq Q$ (`large' case). 

Clearly, $O(n^2)$ communication suffices to solve this problem for any $q < Q$ by Alice sending all her input to Bob. In the other extreme, it is also easy to see that if $Q > q^2$, then $O(1)$ communication suffices to solve 
this problem: if the input of either Alice or Bob is not $q$-colorable itself, they can output `large', otherwise, we know the graph should be $q^2$-colorable by taking the product of their $q$-colorings\footnote{\label{footnote:product}By this, we mean
coloring each $v \in V$ with a pair of colors $(c_A(v),c_B(v))$ from $[q] \times [q]$ where $c_A(v)$ (resp. $c_B(v)$) is the color of $v$ in Alice's (resp. Bob's) $q$-coloring of their input; this way, an edge $e=(u,v)$ in Alice's (resp. Bob's) input 
cannot be monochromatic since $c_A(u) \neq c_A(v)$ (resp. $c_B(u) \neq c_B(v)$).}. We prove that there is almost nothing non-trivial possible between these two extremes when it comes to constant values of $q$. 

\begin{theorem}\label{thm:two-player}
	For any integer $q \in \IN$, any two-player one-way communication protocol for distinguishing between $q$-colorable and $(q^2/4)$-colorable graphs with probability of success at least $2/3$ 
	has communication cost $\Omega(n^2/q^5)$ bits. 
	
	As a corollary, any streaming algorithm for the same problem requires $\Omega(n^2/q^5)$ space. 
\end{theorem}
The smallest non-trivial case of~\Cref{thm:two-player} is for $q \geq 5$ as for smaller values the theorem vacuously holds given $q^2/4 \leq q$ for $q < 5$. We also note that in general, this theorem is most
interesting for large constant values of $q$ and proves distinguishing between $q$ vs $\Omega(q^2)$-colorable graphs requires $\Omega(n^2)$ communication. This exhibits a strong dichotomy that $q$ vs $q^2$ can be 
done in $O(1)$ communication whereas a slightly stronger separation requires $\Omega(n^2)$ communication. 

The second part of~\Cref{thm:two-player} follows immediately from the first part and~\Cref{prop:stream-cc}. In the rest of this section, we prove the first part. This involves presenting a combinatorial 
construction first (\Cref{sec:simple-cluster-packing}) and a reduction from the Index problem using this construction (\Cref{sec:two-player-reduction}).

\newcommand{\Li}[1]{\ensuremath{L^{(#1)}}\xspace}
\newcommand{\vij}[2]{\ensuremath{v^{(#1,#2)}}\xspace}

\subsection{A Combinatorial Construction}\label{sec:simple-cluster-packing}

We present a construction of a $k$-colorable graph which has many disjoint collections of $k$-cliques. This will then be used in our reduction to prove~\Cref{thm:two-player}. 

\begin{lemma}\label{lem:graph-simple}
	For infinitely many integers $n \geq 1$ and any integer $1 \leq k \leq \sqrt{n}$, there exists a $k$-colorable graph $G=(V,E)$ such that its edges can be partitioned into $t = \Omega(n^2/k^5)$ subgraphs $H_1 \sqcup \ldots \sqcup H_t$
	with the following property: for every $i \in [t]$, the induced subgraph of $G$ on $V(H_i)$ is a vertex-disjoint union of $k$ separate $k$-cliques. 
\end{lemma}

\begin{proof} We start with the construction of the graph and then prove it has the desired properties. 

	\paragraph{Construction.} 	Let $G=(V,E)$ be an $n$-vertex graph, initially with no edges (we define its edges later in the proof). Partition the vertex set into $k$ equal-sized \emph{layers} $\Li{1}, \ldots, \Li{k}$. For each $i\in [k]$, divide $\Li{i}$ into \emph{groups} of size $k$, 
	denoted by $\Li{i}_j$ for $j \in [n/k^2]$.   Fix an ordering of the vertices in each $\Li{i}_j$ as $(\vij{i}{j}_1,\ldots,\vij{i}{j}_k)$. For any choice of 
	\[
	b \in [\frac{n}{2k^2}]~, \qquad p\in [\frac{n}{2k^3}]~, \quad \text{and}\quad s\in [k],
	\]
	define a \emph{line} $P_{b,s,p}$ as
	\[
		  P_{b,p,s} = \paren{~\vij{1}{b}_s~,~\vij{2}{b+p}_s~,~\ldots~,~\vij{k}{b+(k-1) \cdot p}_s~};
	\]
	this is a ``(geometric) line'' on vertices in $G$ that starts from $\vij{1}{b}_s$ in $\Li{1}$ and moves with a step of $p$ from a group in one layer to a group in the next layer. See~\Cref{fig:paths-lem} for an illustration. 
	Since 
	\[
		b+(k-1) \cdot p \leq \frac{n}{2k^2} + (k-1) \cdot \frac{n}{2k^3} \leq \frac{n}{k^2},
	\]
	all vertices defined in $P_{b,s,p}$ do indeed belong to the graph $G$.

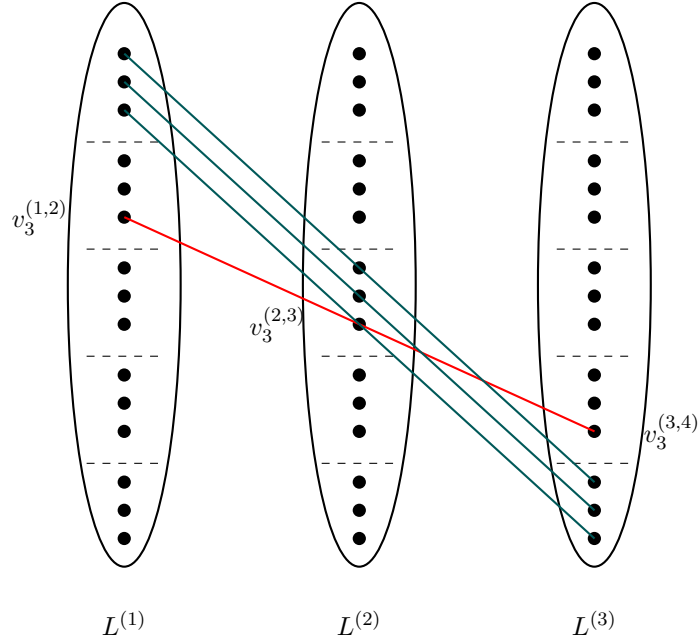
\begin{figure}[h!]
\centering

\begin{tikzpicture}[scale=1.25, every node/.style={scale=0.9}]
  \def\n{5} 
  \def\layers{3}
  \def\layersep{2.5}
  \def\height{3}

  \foreach \i in {1,...,\layers} {
    \draw[thick] ({\i*\layersep},0) ellipse (0.6 and \height);
    
    \foreach \j in {1,...,\n} {
      \pgfmathsetmacro{\y}{\height - (\j)*1.9*\height/\n}
      \ifnum\j=1
      \else
      \draw[dashed] ({\i*\layersep - 0.4}, \y+0.8) -- ({\i*\layersep + 0.4}, \y+0.8);
      \fi

    	\fill ({\i*\layersep}, \y) circle (2pt);
    	\fill ({\i*\layersep}, \y+0.3) circle (2pt);
		\fill ({\i*\layersep}, \y+0.6) circle (2pt);

    }
    
    \node at ({\i*\layersep}, -\height - 0.6) {$L^{(\i)}$};
  }

  \node[left] at ({1*\layersep - 0.5}, { \height - (2)*1.9*\height/\n  }) {$v_3^{(1,2)}$};

  \node[left] at ({2*\layersep -0.47}, { \height - (3)*1.9*\height/\n }){$v_3^{(2,3)}$};

  \node[right] at ({3*\layersep  + 0.43}, { \height - (4)*1.9*\height/\n }) {$v_3^{(3,4)}$};


  \path ({1*\layersep}, { \height - (2)*1.9*\height/\n }) coordinate (A);
  \path ({2*\layersep}, { \height - (3)*1.9*\height/\n }) coordinate (B);
  \path ({3*\layersep}, { \height - (4)*1.9*\height/\n }) coordinate (C);

  \path ({1*\layersep}, { \height - (1)*1.9*\height/\n +0.3}) coordinate (E);
  \path ({2*\layersep}, { \height - (3)*1.9*\height/\n +0.3}) coordinate (F);
  \path ({3*\layersep}, { \height - (5)*1.9*\height/\n +0.3}) coordinate (G);

 \path ({1*\layersep}, { \height - (1)*1.9*\height/\n }) coordinate (E1);
  \path ({2*\layersep}, { \height - (3)*1.9*\height/\n}) coordinate (F1);
  \path ({3*\layersep}, { \height - (5)*1.9*\height/\n }) coordinate (G1);

\path ({1*\layersep}, { \height - (1)*1.9*\height/\n +0.6}) coordinate (E2);
  \path ({2*\layersep}, { \height - (3)*1.9*\height/\n +0.6}) coordinate (F2);
  \path ({3*\layersep}, { \height - (5)*1.9*\height/\n +0.6}) coordinate (G2);

  \draw[thick, red] (A) -- (B) -- (C);

  \draw[thick, teal!70!black] (E) -- (F)-- (G);
                     
  \draw[thick, teal!70!black] (E1) -- (F1)-- (G1);

  \draw[thick, teal!70!black] (E2) -- (F2)-- (G2);

\end{tikzpicture}
\caption{The green edges represent the lines that construct $H_{1,2}$, and the red line represent $P_{2, 1, 3}$ from $H_{2, 1}$. This line intersects $H_{1, 2}$ in at most one vertex and therefore does not add any new edge to the induced subgraph on $H_{1, 2}$. This instance corresponds to a graph in~\Cref{lem:graph-simple} with $k=3$.}
\label{fig:paths-lem} 
\end{figure}

	For any choice of 
	\[
		b \in [\frac{n}{2k^2}] \quad \text{and}\quad p\in [\frac{n}{2k^3}],
	\]
	we define a subgraph $H_{b,p}$ in $G$ as follows: for every $s \in [k]$, create a clique $K_{b,p,s}$ on vertices of the line $P_{b,p,s}$ by adding all their edges to $G$.
	Since for different $s,s' \in [k]$, the two lines $P_{b,p,s}$ and $P_{b,p,s'}$ do not share a vertex, $H_{b,p}$ will be a vertex-disjoint union of $k$ separate $k$-cliques. 
	We also have that $G$ is $k$-colorable by coloring each $\Li{i}$ for $i \in [k]$ with a fixed color since there are no edges inside a layer. 
	Moreover, given the range of $b,p$, we will have $\Omega(n^2/k^5)$ subgraphs $H_{b,p}$. It thus remains to prove the inducedness property of these subgraphs. 
	
	\paragraph{Inducedness property.} Fix any subgraph $H_{b,p}$. We prove that there are no edges between $V(H_{b,p})$ in $G$ except for those of $H_{b,p}$ itself, which implies the desired
	inducedness property. Fix any two vertices $u,v \in H_{b,p}$ as 
	\[
		u = \vij{i\,}{\,b + (i-1) \cdot p}_s \qquad v = \vij{i'\,}{\,b + (i'-1) \cdot p}_{s'},
	\]
	and suppose towards a contradiction that an edge $(u,v)$ is inserted to the graph as a part of some other $H_{b',p'}$ and so we also have, for some $s^* \in [k]$, 
	\[
		u = \vij{i\,}{\,b' + (i-1) \cdot p'}_{s^*} \qquad v = \vij{i'\,}{\,b' + (i'-1) \cdot p'}_{s^*}. 
	\]
	This implies that $s = s' = s^*$, $i \neq i'$ (as there are no edges inside a layer), and additionally 
	\[
		b + (i-1) \cdot p = b' + (i-1) \cdot p' \qquad b+(i'-1) \cdot p = b' + (i'-1) \cdot p';
	\]
	but since $i \neq i'$, we need to have $b=b'$ and $p=p'$ for the above equations to hold (in other words, two different ``lines'' cannot intersect in more than one point). 
	But this contradicts the fact that the edge is inserted by some subgraph other than $H_{b,p}$ (to violate the inducedness property), concluding the proof. 
\end{proof}

\newcommand{\protcolor}{\pi_{\textsc{color}}}
\renewcommand{\protindex}{\pi_{\textsc{index}}}

\subsection{A Two-Player Lower Bound (\Cref{thm:two-player})}\label{sec:two-player-reduction}

We are now ready to prove~\Cref{thm:two-player}. Let $\protcolor$ be a protocol for graph coloring problem in the statement of~\Cref{thm:two-player}. We use this protocol to design
another protocol $\protindex$ for the Index problem defined in~\Cref{sec:two-player-background} and use the lower bound of~\Cref{prop:INDEX-lowerbound} to conclude the proof. 

For some sufficiently large integer $n \geq 1$ and an even integer $q \geq 2$, let $G^*$ be a graph obtained from~\Cref{lem:graph-simple} for the parameter $k=q/2$ and $t=\Omega(n^2/q^5)$ be 
the number of specified subgraphs in $G^*$. Now, consider the Index problem wherein Alice receives $x \in \set{0,1}^t$ and Bob receives  $i \in [t]$. We design the following protocol $\protindex$ for this Index problem: 
\begin{Algorithm}\label{alg:protindex-adv}
	The protocol $\protindex$ for Index of $x \in \set{0,1}^t $,$ i \in [t]$ given the protocol $\protcolor$: 
	\begin{enumerate}
		\item Let $G^*=(V,E^*)$ with specified subgraphs $H_1,\ldots,H_t$ from~\Cref{lem:graph-simple} be as defined above.
		\item Given $x \in \set{0,1}^{t}$, Alice defines the graph $G_A=(V,E_A)$ by starting with $G_A = G^*$ and then, for every $j \in [t]$, removing all edges of $H_j$ from $G$ iff $x_j= 0$.
		\item Bob defines $G_B = (V,E_B)$ as the set of all edges between distinct $k$-cliques inside $H_i$. 
		\item Alice runs $\protcolor$ on $G_A$ and sends the message $M(G_A)$ to Bob, and Bob outputs $x_i = 1$ if $\protcolor$ declares $\chi(G) \geq q^2/4$ and otherwise outputs $x_i = 0$. 
	\end{enumerate}
\end{Algorithm} 

The following lemma establishes the correctness of this reduction. 

\begin{lemma}\label{lem:two-player-reduction}
	Let $G=(V,E_A \cup E_B)$ be the input graph of $\protcolor$ in~\Cref{alg:protindex-adv}. If $x_i = 0$, then $\chi(G) \leq q$ and if $x_i = 1$, then $\chi(G) \geq q^2/4$. 
\end{lemma}
\begin{proof}
	Recall that we set $k=q/2$ above. 
     
    If $x_i=1$, then $V(H_i)$ in $G$ contains $k$ disjoint $k$-cliques in $E_A$ (from $H_i$) and the edges in $E_B$ fully connect every pair of these cliques, making $V(H_i)$ a $k^2$-clique in $G$. 
    Thus, in this case, we have $\chi(G) \geq k^2 = q^2/4$ by the choice of $k$. 
    
    If $x_i=0$, then by~\Cref{lem:graph-simple}, we know that $G_A$ is $k$-colorable. Since $E_B$ only contains edges between the $k$ vertex-disjoint cliques in $V(H_i)$, $G_B$ is also $k$-colorable. Moreover, since all edges of $H_i$ are removed
    from $G_A$, this means that we can color $V(H_i)$ with $k$ colors (to handle edges $E_B$) and $V \setminus V(H_i)$ with another $k$ colors  (to handle edges $E_A$) and have $\chi(G) \leq 2k = q$. 
\end{proof}

\Cref{thm:two-player} now follows easily. 

\begin{proof}[Proof of~\Cref{thm:two-player}]
	By~\Cref{lem:two-player-reduction}, protocol $\protindex$ outputs the correct answer to the Index problem whenever $\protcolor$ is correct in solving the coloring problem on $G = (V, E_A \cup E_B)$, 
	which happens with probability at least $2/3$ by the theorem statement. Since $CC(\protindex) = CC(\protcolor)$ on one hand, and $CC(\protindex) = \Omega(t)$ by~\Cref{prop:INDEX-lowerbound}, 
	we obtain $CC(\protcolor) = \Omega(t)$ also. Finally, since $t = \Omega(n^2/k^5) = \Omega(n^2/q^5)$ by~\Cref{lem:graph-simple}, we can conclude the proof. 
\end{proof}
	
	\clearpage


\newcommand{\LijS}[3]{\ensuremath{L^{(#1)}_{#2}(#3)}}

\section{Cluster Packing Graphs}\label{sec:stronger-cluster-packing}

The combinatorial construction we presented in~\Cref{sec:simple-cluster-packing} played a key role in our two-player communication lower bound. In order to strengthen this lower bound 
for streaming algorithms, by considering more players, we need a generalization of our initial construction. This section is dedicated to the introduction and construction of this family of graphs. 
The definition of these graphs is closely inspired by a similar definition in~\cite{AssadiKNS24}---in particular, their Disjoint Union Path graphs and embedding product---but our main construction in this section (\Cref{prop:cpg-large-r}) is 
new in this context and generalizes the RS graph construction of~\cite{FischerLNRRS02}. 

\subsection{Definition and Constructions}\label{sec:cpg-definition}
To continue, we need some notation. For any integer $k \geq 2$, we define a \textbf{$k$-cluster} as a vertex-disjoint union of $k$-cliques in a graph, and refer to the number of these cliques as the \emph{size} of the cluster. An \emph{induced} $k$-cluster in a graph $G$, 
is a $k$-cluster $C$ such that the induced subgraph of $G$ on vertices of $C$ has no edges other than those of $C$, i.e., $G[V(C)] = C$. 

\begin{definition}\label{def:cpg}
	For any $r,t,k \geq 1$, an \textbf{$(r,t,k)$-cluster packing} graph is a graph $G=(V,E)$ whose edges can be partitioned into $t$ separate induced $k$-clusters $C_1,\ldots,C_t$, each of size $r$. 
	Throughout, we additionally require a $(r,t,k)$-cluster packing graph to be $k$-colorable\footnote{As we shall show in~\Cref{app:cpg-k-color}, any arbitrary cluster packing graph can be modified to satisfy this requirement with minimal
	 loss in the parameters.}.	
\end{definition}

It is easy to see that the construction in~\Cref{sec:simple-cluster-packing} yields a $(k, t, k)$-cluster packing, i.e., creating $\Omega_k(n^2)$ induced $k$-clusters of size $r=k$. 
Furthermore, the same combinatorial structure can be extended directly to obtain a $(r, t, k)$-cluster packing graph for larger values of $r$ as well. 

\begin{proposition}\label{prop:generalized-simple-graph}
	For all integers $r,k,n \geq 1$ such that $r \cdot k \leq \sqrt{n}$, there exists a $(r,t,k)$-cluster packing graph $G$ with $t = \Omega(n^2/(r^2 \cdot k^3))$ clusters. 
\end{proposition}

\Cref{prop:generalized-simple-graph} is proved in \Cref{app:generalized-simple-graph}.
Nevertheless, this construction is quite ``weak'' for our purpose of proving lower bounds, as roughly speaking, we need both $r$ and $t$ to be at least $n^{1-o(1)}$. However, setting $r=n^{1-o(1)}$ only allows for $t=n^{o(1)}$ in~\Cref{prop:generalized-simple-graph}. To improve~\Cref{prop:generalized-simple-graph}, we can take inspiration from the construction of RS graphs which are in fact $(r,t,2)$-cluster packing graphs (since a $2$-cluster is simply a matching). 
Using this, we provide two ``stronger'' constructions of cluster packing graphs. 

The first construction generalizes the original RS graphs of~\cite{RuzsaS78} to larger values of $k$ to allow for $r = n/2^{\Theta_k(\sqrt{\log{n}})}$. The proof of this result 
is somewhat implicit in~\cite{AssadiKNS24}---itself based on~\cite{Alon01} and~\cite{AbboudB16}---and we provide it for completeness in \Cref{app:cpg-dense}. 

\begin{proposition}[cf.~\cite{AssadiKNS24}]\label{prop:cpg-dense}
	For all integers $k,n \geq 1$ such that $k = o(\sqrt{\log{n}})$, there exists a $(r,t,k)$-cluster packing graph $G$ with parameters 
	\[
		r = \frac{n}{2^{\Theta(\sqrt{\log{n} \cdot \log{k}})}} \quad and \quad t = \frac{n}{2^{\Theta(\sqrt{\log{n} \cdot \log{k}})}}. 
	\]
\end{proposition}

This construction will be used in the first part of~\Cref{res:adv}. However, even this construction is not strong enough for proving the second part of~\Cref{res:adv}, which requires $r$ to be $\Theta_k(n)$.  
The following construction, which is the main contribution of this section, achieves this at a cost of a (much) smaller, but still sufficiently large value of $t$. This construction generalizes the RS graphs of~\cite{FischerLNRRS02}
to larger values of $k > 2$. 

\begin{theorem}\label{prop:cpg-large-r}
	For all integers $k,n \geq 1$ such that $k = o(\frac{\log{n}}{\log\log{n}})$, there exists a $(r,t,k)$-cluster packing graph $G$ with parameters 
	\[
		r = \frac{n}{2k^2} \cdot (1-o(1)) \quad and \quad t = n^{\Omega(1/\log\log{n})}. 
	\]
\end{theorem}

We note that the $n/2k^2$ bound in~\Cref{prop:cpg-large-r} can be further improved to $n/k^2$ using the idea of~\cite{GoelKK12} for RS graphs; 
we can also show that the $n/k^2$ is the ``right'' threshold for size of clusters to have a ``large'' choice of $t$; since this part is tangent to the main results in the paper, we postpone it to~\Cref{app:cpg}, wherein we also study some further aspects of cluster packing graphs 
as a combinatorial structure of their own independent interest.


\subsection{An Almost Dense Construction with  Linear Size Clusters (\Cref{prop:cpg-large-r})}\label{sec:sparse-cluster-packing}

We now prove~\Cref{prop:cpg-large-r}. Before we start, we need the following standard claim on existence of an exponentially large 
family of sets with small pairwise intersection. This result is standard and can be proven easily for instance by picking the sets randomly, or using Gilbert-Varshamov bound from coding theory (for completeness, we provide
a short proof in~\Cref{app:missing}). 

\begin{claim}\label{claim:num-of-sets}
        There exist $t=2^{\Omega(d)}$ sets $S_1, \ldots, S_t\subseteq [d]$ such that for all $i \neq j \in [d]$, 
        \[
        		\card{S_i} = \frac{d}{3} \quad and \quad \card{S_i \cap S_j} \leq \frac{d}{7}.
	\]
\end{claim}

We can now provide the proof of our main construction of cluster packing graphs. 

\begin{proof}[Proof of~\Cref{prop:cpg-large-r}]
	Start with an empty graph $G=(V,E)$ with $n$ vertices for a sufficiently large $n$. We start with some parameters and notation for vertices and then define the edges. 
	The reader may want to refer to~\Cref{fig:cpg-large-r} for an illustration of different parts of the construction. 
	
	\paragraph{Parameters, vertices, and their weights and coloring.} Fix the following integer parameters 
	 \[
	 d=\Theta({\frac{\log n}{\log\log n}}) \quad and \quad p:=k\cdot d^{10},
	 \]
	 where the hidden constant in the $\Theta$-notation for $d$ will be chosen at the end of the proof. 
	 Partition vertices of $V$ into $k$ vertex sets $\Li{1},\ldots,\Li{k}$, referred to as \emph{layers}, where each $\Li{i} := [p]^d$. Let $\mathcal{F}$ be the collection of subsets $S \subseteq [d]$ 
     promised by~\Cref{claim:num-of-sets}. For each $S \in \mathcal{F}$, define a weight function $w_S: [p]^d \rightarrow \IN$ over vertices $x = (x_1, \ldots, x_d) \in [p]^d$ as
    \[
        w_S(x) := \sum_{i \in S} x_i.
    \]
    Note that $\card{S} \leq w_S(x) \leq \left|S\right|\cdot p$. We further partition each layer $L_i$ into $p$ \emph{groups} based on $w_S(x)$, where for any $i \in [k]$ and $j \in [p-1]$, we
    define the group $\LijS{i}{j}{S}$ as 
    \[
    	\LijS{i}{j}{S} := \set{x \in \Li{i} \mid  j \cdot |S| \leq w_S(x) < (j+1) \cdot |S|}. 
    \]
   Next, define a color tuple
    \[
        (c_1~,~ \text{`white'}~,~ c_2~,~ \text{`white'}~,~ \ldots~,~ c_k~,~ \text{`white'}),
    \]
    where we take `white' and $c_1,\ldots,c_k$ to be $k$ distinct colors. Assign colors to the groups in each layer $\Li{i}$ \emph{cyclically} according to this tuple: 
    color $\LijS{i}{1}{S}$ with $c_1$, $\LijS{i}{2}{S}$ with `white', $\LijS{i}{3}{S}$ with $c_2$, and so on until coloring $\LijS{i}{2k-1}{S}$ with $c_k$, $\LijS{i}{2k}{S}$ with white, and then repeat the same cycle until all groups in $\Li{i}$ are colored. 
    For any vertex $x$, let $C_S(x)$ denote the color of $x$ according to the set $S$.

   \paragraph{Lines and edges.} For a fixed $S \in \mathcal{F}$, and every vertex $x \in \Li{1}$ colored $c_1$ under $S$, as long as all coordinates of $x+ 2k\cdot \mathbf{1}_S$ are at most $p$, define a \emph{line} as a tuple of vertices: 
    \[
        P_{x,S} = (y_1, y_2, \ldots, y_k), \quad \text{where } y_1 := x \quad and \quad y_{i+1} := y_i + 2 \mathbf{1}_S \in \Li{i+1};
    \]
    here, $\mathbf{1}_S \in \{0,1\}^d$ is the indicator vector of $S$. 
    
    \begin{observation}\label{obs:colorful-line}
    	For a line $P_{x,S}$ with vertices $y_1,\ldots,y_k$, the color of vertices, with respect to $S$, will be $c_1,\ldots,c_k$, respectively. 
    \end{observation}
    \begin{proof}
    	Firstly, note that since the lines are only defined when all coordinates of $x+ k\cdot 2 \mathbf{1}_S$ are at most $p$, all vertices in $P_{x,S}$ do belong to the graph. 
	Moreover, since for every $i \in [k-1]$, we have $w_S(y_{i+1}) = w_S(y_i) + 2|S|$, the group-number of $y_{i+1}$ will be two more than the group-number of $y_i$ (as the gap between weight of each group is $\card{S}$). 
	As such, since $y_1$ is colored with $c_1$ and we cyclically color the vertices of the groups and the path goes through $k$ groups of vertices, we obtain that the coloring of vertices are $c_1,\ldots,c_k$, respectively. \Qed{obs:colorful-line} 
	
    \end{proof}
    
    The edges of the graph are then as follows: for any $S \in \mathcal{F}$ and any $x \in \Li{1}$ where the line $P_{x,S}$ is defined, we add all edges between vertices of each $P_{x,S}$ to turn it into a $k$-clique. 
    Finally, for any $S \in \mathcal{F}$, we define a $k$-cluster $H_S$ as the set of all $k$-cliques obtained from all qualifying lines $P_{x,S}$. The fact that $H_S$ is a $k$-cluster is because 
    the lines $P_{x,S}$ for different $x$, by definition, do not share vertices. In the following, we prove a lower bound on the size of these $k$-clusters and show their inducedness in the graph and then conclude the proof. 
    
    \begin{figure}[h!]
    \centering
    \includegraphics[width=0.56\textwidth]{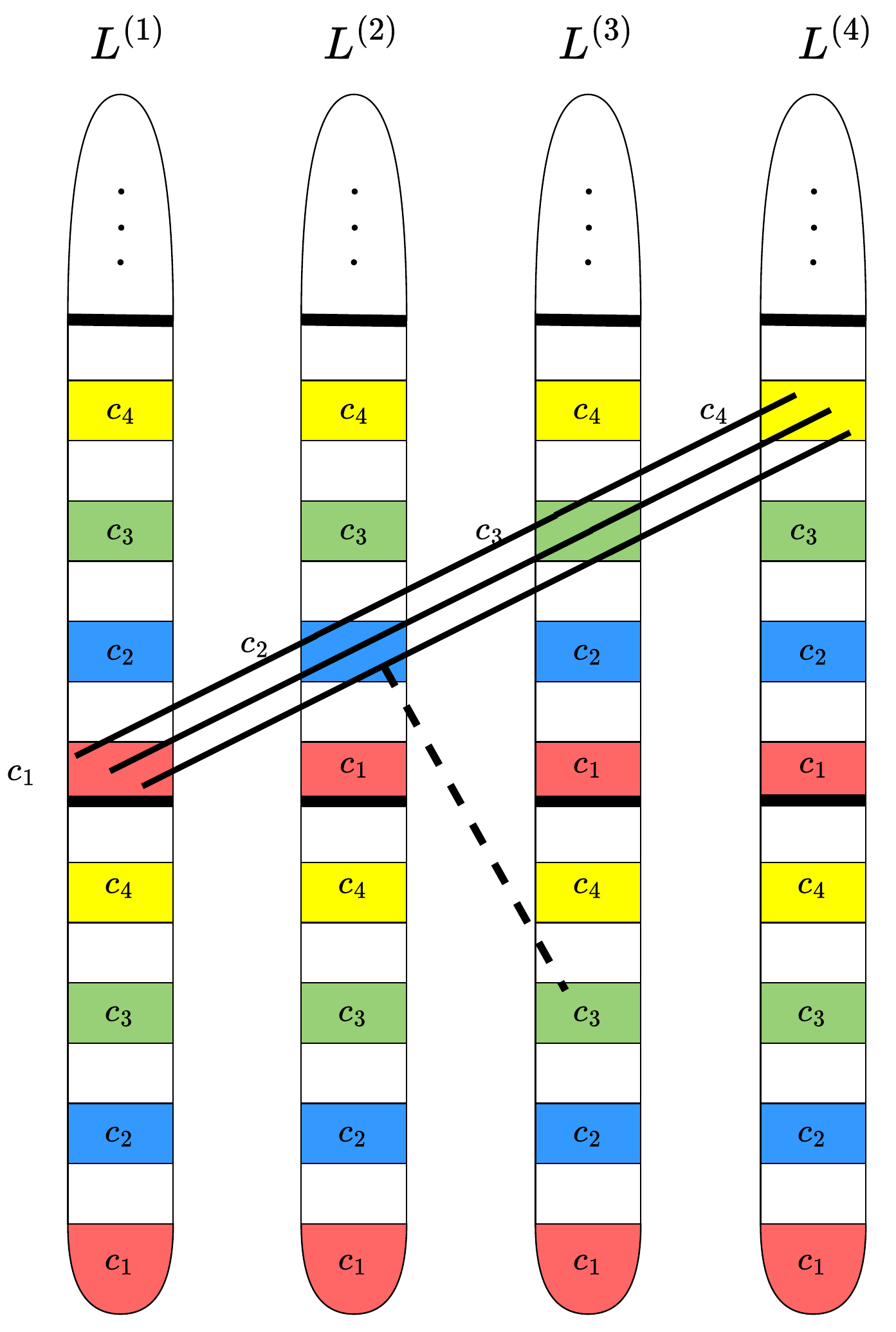}
    \caption{An illustration of the construction of the cluster packing graph in~\Cref{prop:cpg-large-r} for $k=4$ and a coloring of vertices in the layers based on a fixed set $S$. 
    The colors used in the definition of the construction are $c_1$ (red), $c_2$ (blue), $c_3$ (green), $c_4$ (yellow), and `white'. The thick lines represent the cliques in $H_S$. An edge corresponding to the dashed line
    would break the inducedness of the subgraph $H_S$, but it cannot exist as in any other subgraph $H_T$ for $T \neq S$, there are no lines between a vertex of color $c_2$ (blue) to any color $c_3$ (green) vertex (the coloring
    here is still with respect to $S$ and not $T$).}
    \label{fig:cpg-large-r}
\end{figure}
    
    \paragraph{Size of $k$-clusters.} The following claim lower bounds the size of each $k$-cluster defined in $G$. 
    
    \begin{claim}\label{claim:lower-bound-r}
        For a fixed $S\in \mathcal{F}$, the number of distinct lines $P_{x, S}$ is at least $\frac{n}{2k^2} \cdot (1-o(1))$.
    \end{claim}
    \begin{proof}
        The number of $P_{x,S}$ lines depends on the number of vertices in $\Li{1}$ that are colored $c_1$ and satisfy $x_i +2k \leq p$ for all $i \in S$.

        The total number of distinct weights is $\card{S} \cdot (p-1) + 1$, hence the number of groups in $\Li{1}$ is at most $p$. Since the color $c_1$ appears in $\Li{1}$ once every $2k$ groups, the number of $c_1$-colored groups is at least $\frac{p-1}{2k} - 1$, so the proportion of $c_1$-colored groups is at least
        \begin{align*}
            \frac{\frac{p-1}{2k}-1}{p}=\frac{1}{2k} \cdot (1-o(1)). \tag{since $p \gg k$} 
        \end{align*}
        The total number of vertices in $L_1$ is $p^d$. The number of vertices with all coordinates in $S$ smaller than $p-2k$ is at least $(p - 2k)^d$. Hence, the proportion of such vertices is at least
        \begin{align*}
            \frac{(p-2k)^d}{p^d}
            &=\left(1-\frac{2k}{p}\right)^d \simeq \exp\left(-\frac{2k\cdot d}{k\cdot d^{10}}\right) \tag{using $1-x\simeq e^{-x}$}\\
            & \simeq 1-\frac{2k\cdot d}{k\cdot d^{10}} = 1-\frac{2}{d^9} = 1- o(\frac1k) \tag{using $p = k \cdot d^{10}$ and $d > k$}. 
        \end{align*}
         Therefore, given $n= k \cdot p^d$ and $\card{L_1}=p^k$, the number of qualifying lines is at least
        \[
            \frac{p^d}{2k} \cdot \paren{1-o(1)} - o(\frac{p^d}{k}) \geq \frac{n}{2k^2} \cdot \paren{1-o(1)},
        \]
        completing the proof. \Qed{claim:lower-bound-r}
        
    \end{proof}

   \paragraph{Inducedness.} The following claim shows that every $k$-cluster $H_S$ is induced in $G$. 
   
    \begin{claim}\label{claim:H_S-induced}
        For each $S\in \mathcal{F}$, the induced subgraph of $G$ on $V(H_S)$ only contains edges of $H_S$. 
    \end{claim}
    \begin{proof}
        Fix some $S \in \mathcal{F}$ and suppose by contradiction that there exists an edge $(u, v)$ in $V(H_S)$ which is not part of $H_S$ itself.  
        By construction, such an edge could only exist if both $u$ and $v$ lie on the same line $P_{z, T}$ for some $z \in \Li{1}$, $T \neq S \in \mathcal{F}$ (so are added as part of $H_T$).
        Assume $u \in \Li{i}$ and $v \in \Li{j}$ with $i \neq j$ (since there are no intra-layer edges). 
        By~\Cref{obs:colorful-line}, we have $C_S(u)=c_i$ and $C_S(v)=c_j$, therefore there exist integers $\alpha$ and $\beta$ such that
         \begin{align*}
         (2k\cdot \alpha+2(i-1)) \card{S} &\leq w_S(u) < (2k\cdot\alpha+2(i-1)+1) \card{S};
         \\
         (2k\cdot\beta+2(j-1)) \card{S} &\leq w_S(v) < (2k\cdot\beta+2(j-1)+1) \card{S}.
         \end{align*}
         This is by the cyclic coloring of groups in each layer and the definition of each group. 
         Assume $i<j$, then with $u, v \in P_{z, T}$, we have $v=u+2(j-i) \mathbf{1}_T$. Therefore,
         \[
         w_S(v)= w_S(u+2(j-i) \mathbf{1}_T)= w_S(u) + w_S(2(j-i) \mathbf{1}_T) \geq w_S(u).
         \]
         Using $w_S(v) \geq w_S(u$), and as the weight range for different groups have no overlap, we have
         $2k\cdot\beta+2(j-1) \geq 2k\cdot \alpha+2(i-1)+1$ (and consequently, $\beta\geq \alpha$). Then, with respect to the weights on set $S$:
        \begin{align*}
            w_S(v-u) 
            &= w_S(v) - w_S(u)
            \\
            &\geq (2k\cdot\beta+2(j-1))\left|S\right| - (2k\cdot\alpha+2(i-1)+1)\left| S \right| \\
            &= \card{S}\cdot \paren{2k\cdot(\beta-\alpha)+2(j-i)-1}\\
            &= \frac{d}{3} \cdot \paren{2k\cdot(\beta-\alpha)+2(j-i)-1} \tag{as $\card{S} = d/3$} \\
            &\geq \frac{d}{3} \cdot \paren{2(j-i)-1} \tag{since $\beta\geq \alpha$} \\
            &\geq \frac{d}{3} \cdot (j-i) \tag{since $\card{j-i}\geq 1$} 
        \intertext{
        On the other hand, as $u, v \in P_{z, T}$, we have $v=u+2(j-i) \mathbf{1}_T$. Therefore, 
        }
             w_S(v-u) &= w_S(2(j-i) \cdot \mathbf{1}_T)= 2(j-i)\cdot { w_S(\mathbf{1}_T)} = {2(j-i)}\cdot\card{S\cap T} \\ 
            &\leq {2(j-i)}\cdot \frac{d}{7} < \frac{d}{3} \cdot (j-i), 
        \end{align*}
        a contradiction (as the lower bound in RHS of first equation is larger than the upper bound of the second). Hence, the edge $(u,v)$ cannot exist, proving the inducedness property. \Qed{claim:H_S-induced}
   
    \end{proof}

    \paragraph{Concluding the proof.} To conclude, by~\Cref{claim:lower-bound-r} and~\Cref{claim:H_S-induced}, we already established that $G$ is a $(r,t,k)$-cluster packing graph with parameters 
    \[
    	r = \frac{n}{2k^2} \cdot (1-o(1)) \qquad and \qquad t = 2^{\Omega(d)} = n^{\Theta(\frac{1}{\log\log{n}})}; 
    \]
    we only need to ensure that we can pick $d$ to be in $\Theta(\log{n}/\log\log{n})$ as required by our proof. By the construction of the graph and choice of parameter $p$, we have 
    \[
    	n = k \cdot p^d = k \cdot \paren{k \cdot d^{10}}^d = k^{d+1} \cdot d^{10d}. 
    \]
    Since $k = o(\log{n}/\log\log{n})$ in the theorem statement, we can set $d$ to be some fixed choice in $\Theta(\log{n}/\log\log{n})$ and ensure the above equation holds, concluding the proof. 
	\Qed{prop:cpg-large-r} 
	
\end{proof}

	\clearpage
	

\newcommand{\Gclust}{\ensuremath{G^{\textnormal{cl}}}}
\newcommand{\Vclust}{\ensuremath{V^{\textnormal{cl}}}}
\newcommand{\Eclust}{\ensuremath{E^{\textnormal{cl}}}}

\newcommand{\ristar}{\ensuremath{\rv{i}^{\star}}}

\newcommand{\cSI}{\ensuremath{c_{\textnormal{\textsc{SI}}}}}

\newcommand{\cDind}{\ensuremath{\mathcal{D}_{\textnormal{\textsc{IND}}}}}
\newcommand{\cind}{\ensuremath{c_{\textnormal{\textsc{ind}}}}}

\newcommand{\rG}{\ensuremath{\rv{G}}}

\newcommand{\rGinput}[1]{\ensuremath{\rv{G}({#1})}}

\newcommand{\cGfake}{\ensuremath{\mathcal{G}^{\textnormal{fake}}}}
\newcommand{\cGreal}{\ensuremath{\mathcal{G}^{\textnormal{real}}}}

\newcommand{\rProtIND}{\rProt_{\textnormal{IND}}}
\newcommand{\rI}{\ensuremath{\rv{I}}}

\newcommand{\cmsg}{\ensuremath{c_{\textnormal{msg}}}}


\section{Adversarial Streams: A Lower Bound via Cluster Packing Graphs}\label{sec:adversarial} 

In this section, we prove the following theorem that formalizes \Cref{res:adv}.
\begin{theorem}\label{thm:adv-res}
	For any integers $p \geq 3$ and $k \geq 2$, we have:
	\begin{enumerate}
		\item If $p,k = o(\sqrt{\log n})$ and $ k = n^{o(1/p^2)}$, any streaming algorithm on adversarial streams that distinguishes between $k\cdot p$-colorable graphs and $k^p$-colorable ones with probability at least $2/3$ requires $n^{2-o(1)}$ space.
	
	\item 
	For the larger range of $3\leq p, k\leq o(\frac{\log n}{\log \log n})$ and $p \cdot \log k = o(\frac{\log n}{\log \log n})$, the same problem requires $n^{1+\Omega(1/\log\log n)}$ space.
	\end{enumerate}
\end{theorem}

\Cref{res:adv} can be obtained from \Cref{thm:adv-res} by setting $k = 2$. 

We start this section with describing the $p$-player communication model. Then we present the hard distribution and its analysis. Lastly we prove \Cref{thm:adv-res} using the cluster packing graphs from \Cref{sec:stronger-cluster-packing}.
Given the proof consists of various parts and to help with keeping track of the different components of the proof, a schematic organization of the proof is provided in~\Cref{sec:schema-org}. 

\subsection{Background on the Model}\label{subsec:p-player-back}
For any $p\in \IN$, in the $p$-player one-way communication model, there are $p$ players $\player{1}, \ldots, \player{p}$, with inputs $x_i \in \mathcal{X}_i$ for all $i\in[p]$. The goal is to compute a function $f(x_1, \ldots, x_p)$ for some known function $f$ on the domain $\mathcal{X}_1\times \ldots \times \mathcal{X}_p$. 
To do this, the players have a blackboard that everyone can see. For each $i = 1$ to $p$, $\player{i}$ sends a single message $M_i$ to the blackboard as a function of $x_i$ and prior messages $M_1,\ldots,M_{i-1}$, i.e., $\player{1}$ sends the first message, then $\player{2}$ and so on till $\player{p}$. The entire contents of the blackboard is called a \textbf{transcript}.
 The message of $\player{p}$ should also contain the answer to the problem. The players have access to a shared tape of randomness, referred to as \textit{public randomness}, in addition to their own \textit{private randomness}. 

\begin{definition}
	For any protocol $\pi$ in the $p$-player one-way communication model, the communication cost of $\pi$, denoted by $CC_p(\pi)$, is defined as the worst-case length of any message sent by any player $\player{i}$ to the blackboard for $i \in[p]$. 
	\[
	CC_p(\pi) = \max_{(x_1, x_2, \ldots, x_p)} \max_{i\in[p]} \card{M_i(x_i, M_1, M_2, \ldots, M_{i-1})}.
	\]
\end{definition}

The following result, also dating back to~\cite{AlonMS96}, is a direct generalization of~\Cref{prop:stream-cc} to more than two players. 

\begin{proposition}\label{prop:adv-communication-to-stream}
	Any $s$-space streaming algorithm $A$ for computing $f(x_1, \ldots, x_p)$ on adversarial stream $x_1 \circ x_2 \ldots \circ x_p$, namely concatenation of all $x_i$ for $i\in[p]$, implies a communication protocol $\pi$ with $CC_p(\pi)= O(s)$ with the same success probability. 
\end{proposition}
\begin{proof}
	The proof follows similarly as that of \Cref{prop:stream-cc}, as the players can run the algorithm $A$ and send its memory state to the blackboard.
\end{proof}

We now introduce some notation used throughout this section.
\paragraph{Notation.}
In any $(r,t,k)$-cluster packing graph, we use $\cliqnum{i}{j}$ for $i \in [t]$ and $j \in [r]$ to denote the vertices of the $j^\textnormal{th}$ $k$-clique from the $i^\textnormal{th}$ $k$-cluster of size $r$ in the graph. With a slight abuse of notation, when we say edges of $\cliqnum{i}{j}$, we mean all the $\binom{k}{2}$ edges inside the $k$-clique that the vertices in $\cliqnum{i}{j}$ form. 

We use $\cG_p(n, k)$ to denote the hard distribution over graphs with $n$ vertices for $p$-players for $p \geq 2$. The graphs sampled from $\cG_p(n,k)$ will have chromatic number either at most $k \cdot p$ or at least $k^p$. When it is clear from context, we omit the parameters $n, k$. 

For each instance $G \sim \cG_p(n, k)$, there exists a bit $\ans(G) $, termed as a \textbf{special bit}, which is either $0$ or $1$ corresponding to whether  $\chrom(G) \leq k \cdot p$ or $\chrom(G) \geq k^p$, respectively. There also exists a subset of vertices $\spec(G) \subset V$, termed as a \textbf{special set}, of size $k^p$ which forms a clique of size $k^p$ when $\ans(G) = 1$. 

The players $\player{1},\ldots, \player{p}$ receive a partition of the edges of $G \sim \cG_p(n,k)$, but they also receive some auxiliary information about the distribution. This can only help the players with respect to distinguishing between $\chrom(G) \leq k \cdot p$ and $\chrom(G) > k^p$, as the auxiliary information can always be ignored. At every step of the distribution, we will state the edges received, and any other random variables given to the players, which can be treated as auxiliary information.

In a $p$-player protocol for solving $G \sim \cG_p(n,k)$, the message of last player $\player{p}$ has to contain the value of $\ans(G)$, and should be on the blackboard at the end of the protocol. For $i \in [p]$, we use $\Prot_i$ to denote the message sent by $\player{i}$ in any protocol $\prot$, and $\Prot$ to denote the transcript at the end of all the messages sent in the protocol.

In order to prove \Cref{thm:adv-res} we need the following lemma.
The lower bound will use the two kinds of cluster packing graphs introduced in \Cref{sec:stronger-cluster-packing}, so for now, we state the lemma in terms of the parameters for cluster packing graphs, and prove \Cref{thm:adv-res} by using the parameters later in \Cref{subsec:thm-adv-res-proof}. 

\begin{lemma}\label{lem:adv-lb-cluster}
For any $n, p, k \geq 2$, suppose for all $2 \leq a \leq p$ there exists a family of $(r_a, t_a, k)$-cluster packing graphs $\{\Gclust = (\Vclust, \Eclust)\}_{a}$ each on $n_a$ vertices, with parameters obeying:
\begin{equation}\label{eq:params-all}
	n_p = n, \qquad n_{a-1} = r_a/4 \textnormal{ for all $3 \leq a \leq p$,} \qquad \textnormal{ and }\qquad r_2 = k;
\end{equation}
then, any deterministic $p$-player protocol $\prot$ that distinguishes between input graph $G$ with $\chrom(G) \leq k \cdot p$ and $\chrom(G) \geq k^p$ with probability at least $2/3$ must have,
\[
CC(\prot)  = 	\Omega(1) \cdot\underset{2 \leq a \leq p}{\min}\Big \{ \frac{r_a \cdot t_a}{(a+1)^2 \cdot k^{a-1}}\Big\}.
\]
\end{lemma}
Proof of \Cref{lem:adv-lb-cluster} is given in \Cref{subsec:player-elim}.

\subsection{Base Case: Two Player Lower Bound}\label{subsec:adv-basecase}

		In this subsection, we construct the hard distribution for the two-player case, which serves as the base case for extending the construction to $p > 2$ players. This section closely follows~\Cref{sec:two-player}, with minor modifications to ensure it fits the base case requirements.

	Recall the Index communication problem defined in \Cref{sec:two-player-background}, and its hardness in \Cref{prop:INDEX-lowerbound}. Before proceeding, we need a similar result where the probability of success can vary and 
	also hold for \emph{internal information complexity} that we define next. 
	
\subsubsection{Detour: Internal Information Complexity}

Recall the two-player communication model we introduced in~\Cref{sec:two-player}: there are two players, Alice and Bob, with inputs $x\in \mathcal{X}$ and $y \in \mathcal{Y}$, respectively.\footnote{For the case of this definition, (and for other canonical communication complexity problems that we use in this section), the players will be Alice and Bob. In the two-player problem corresponding to \Cref{lem:adv-lb-cluster}, the players will be $\player{1}$ and $\player{2}$. } The goal is to compute
$f(x,y)$ for some known function $f$ on the domain $\mathcal{X} \times \mathcal{Y}$, where Alice sends a single message $M$ to Bob. 
	\begin{Definition}[\!\!\cite{BarakBCR10}]\label{def:ic}
		In any rnprotocol $\prot$ for computing $f$, where the input of Alice and Bob is distributed according to distribution $\cD$, the \textbf{internal information complexity} of protocol $\prot$ with respect to distribution $\cD$ is defined as:
		\[
			\ic(\prot, \cD) = \mi{\rX}{\rv{M} \mid \rY}.
		\]
	\end{Definition}

The following proposition is standard (see, e.g.~\cite{Ablayev96} for the communication version of this result and \cite[Lemma 3.4]{AssadiKL16} for its extension to information complexity). 

\begin{proposition}\label{prop:index-low-prob}
	For any $\eps \in (0,1)$, any protocol $\prot$ that solves the Index problem over the uniform distribution $\cDind$ over $\{0,1\}^{\ell} \times [\ell]$ of inputs with probability of success at least $1/2 + \eps$ has internal information cost $$\ic(\prot, \cDind) \geq \cind \cdot \eps^2 m,$$ for some absolute positive constant $\cind$. 
\end{proposition}

\subsubsection{Back to the Base Case Construction}
	
	We now describe the hard distribution for our coloring problem for two players.

		\begin{Distribution}\label{dist:adv-2player}
		\textbf{The distribution $\cG_2(n,k)$ for $n, k \geq 1$}
		\begin{enumerate}[label=$(\roman*)$]
			\item Pick an $(r_2, t_2, k)$-cluster packing graph $\Gclust = (\Vclust, \Eclust)$ on $n_2 = n$ vertices with parameters $r_2=k$ and $t_2 = \Omega(n^2/k^5)$ that was constructed in \Cref{lem:graph-simple} (which obeys the parameters in statement of \Cref{lem:adv-lb-cluster} also).
			
			\item Sample a cluster index $\istar$ uniformly at random from $[t_2]$, and give it to $\player{2}$. Let the edge set of $\player{2}$ be as follows:
			\[
				E_2 = \{ (u, v)\mid u\in \cliqnum{\istar}{j} \textnormal{ and } v\in \cliqnum{\istar}{j'}, \forall j\ne j' \in [k]\}.
			\] 
			
			\item  Sample string $x \in \{0, 1\}^{t_2}$ (subset of clusters) uniformly at random and give it to $\player{1}$.
The edge set given to $\player{1}$ will be:
			\[
				E_1 = \{\textnormal{edges of clique $\cliqnum{i}{j}$} \mid x_i = 1, \forall j \in [r_2]\}.
			\]
			(Note that edges inside the other cliques are not added to the graph.)
			
			\item Set $\ans(G)=x_{\istar}$ and $\spec(G) = \{u \mid u \in \cliqnum{\istar}{j}, \forall j\in [r_2] \}$.
			
		\end{enumerate}
	\end{Distribution}

	The following observation is the analogue of~\Cref{lem:two-player-reduction} for~\Cref{dist:adv-2player} and we omit its proof. 
		\begin{observation}\label{obs:two-player-clique}
		 For any $G\in \cG_2(n, k)$, 
		 \begin{enumerate}[label=$(\roman*)$]
			\item If $\ans(G)=0$, the graph $G$ has $\chrom(G)\leq 2k$.
			\item If $\ans(G)=1$, then $\spec(G)$ forms a $k^2$-clique and hence $\chrom(G)\geq k^2$.
		 \end{enumerate}
		\end{observation}
		
		We also need the following simple observation regarding the distribution of $\ans(G)$. 
		\begin{observation}\label{obs:2player-ind-prop-1}
			The distribution $\ans(G)$ is uniform over $\{0, 1\}$ and is independent of the input of the player $\player{2}$.
		\end{observation}
		\begin{proof}
			The value of $\ans(G)$ is determined by $x_{\istar}$ which is chosen uniformly at random over $\{0, 1\}$. As the input of $\player{2}$ only depends on $\istar$, no information from $x$ is known to $\player{2}$.
		\end{proof}
		
		We now prove the lower bound for base case of~\Cref{lem:adv-lb-cluster}. The rather peculiar constants in the probability of success and communication bounds are chosen so that they precisely match the bounds of~\Cref{lem:adv-lb-cluster} 
		for $p=2$. 
		
		\begin{claim}\label{clm:two-player-basecase}
		Given any $(r_2, t_2, k)$-cluster packing graph $\Gclust = (\Vclust, \Eclust)$ on $n_2$ vertices with $r_2 = k$, any two-player protocol $\prot$ that distinguishes between input graphs $G$ satisfying $\chi(G) \leq 2k$ and those with $\chi(G) > k^2$, with success probability at least $11/18$, has communication cost of at least $c_2 \cdot r_2 \cdot t_2 / 9k$ for some absolute constant $c_2 > 0$.
		\end{claim}
		\begin{proof}
			We use \Cref{alg:protindex-adv} to construct protocol $\protindex$ given any protocol $\prot$ such that $\protindex$ outputs the correct answer whenever $\prot$ distinguishes between input graphs $G$ with $\chi(G) \leq 2k$ and $\chi(G) > k^2$  correctly.

			It is easy to see that if the input distribution of $\protindex$ is the uniform distribution $\cDind$ over $\{0,1\}^{t_2} \times [t_2]$, the distribution given to $\prot$ in \Cref{alg:protindex-adv} is exactly the distribution $\cG_2(n,k)$. This is because, the string $x \in \{0,1\}^{t_2}$ is chosen uniformly at random and so is index $\istar$. 
			Therefore, by the same argument as in~\Cref{thm:two-player}, $\protindex$ solves the Index problem over $\cDind$ with probability at least $11/18$.
			
			By~\Cref{prop:index-low-prob}, we then have:
			\[
				CC(\protindex) \geq \ic(\protindex, \cDind ) \geq \cind \cdot (7/18)^2 \cdot t_2,
			\]
			where the first inequality follows because information cost is upper bounded by entropy of the message and entropy is upper bounded by the length of the message (see  \itfacts{info-zero})). 
			
			Using the fact that $r_2 = k$ in \Cref{eq:params-all}, we obtain the lower bound 
			$\cind \cdot (7/18)^2 \cdot r_2\cdot t_2/k$, which means there is a positive constant $c_2$ as desired.
		\end{proof}		
		
	
	\subsection{A Hard Input Distribution}\label{subsec:adv-construct}

	In this subsection, we construct our hard distribution over input graphs as a function of the number of players $p$. The construction proceeds by induction on $p$, starting from the base case $p = 2$, which was discussed in~\Cref{subsec:adv-basecase}.

	This subsection will be presented in three parts, corresponding to the three components of the distribution $\cG_p(n,k)$. We remark that the three parts are not disjoint from each other, and will be correlated through some random variables. 
	
	\subsubsection{The Set Intersection Component}
	
	We begin by describing the first part of the distribution $\cG_p(n,k)$, based on hard instances of the set intersection problem.
	Some necessary details about this problem are given below. 

	We start with the definition of the problem. 

	\begin{problem}\label{def:set-intersect}
	The \textit{Set Intersection} problem is a two-player communication problem where Alice and Bob are each given a subset of $[\ell]$, denoted $A$ and $B$ respectively, with the promise that there exists a unique element $\{\estar\}$ such that $A \cap B = \{\estar\}$. The goal is to identify the target element $\estar$ using back-and-forth communication.
	\end{problem}
	
	It is not sufficient for our purpose to work simply with probability of success for solving set intersection. We need a more nuanced notion of making progress on the finding element $\estar$---due to~\cite{AssadiCK19b,AssadiR20}---which is defined below. 
	
	\begin{definition}[\!\!\cite{AssadiCK19b,AssadiR20}]\label{def:solve-epsilon-SI}
		Let $\cD$ be a distribution of inputs $(A, B)$ for the Set Intersection problem, known to Alice and Bob. A protocol $\prot$ internal $\eps$-solves Set Intersection over $\cD$ iff at least one of the following is holds:
		\[
		\Exp_{\rProt, \rA} \tvd{\distribution{\estar \mid \Prot, A}}{\distribution{\estar \mid A}} \geq \eps \qquad \textnormal{ or } \qquad \Exp_{\rProt, \rB} \tvd{\distribution{\estar \mid \Prot, B}}{\distribution{\estar \mid B}} \geq \eps.
		\]
		where $\Prot$ denotes all the messages communicated by both Alice and Bob.
	\end{definition}	
	\noindent
	Note that $\eps$-solving in general is an easier task than actually finding the intersecting the element. 
	
	We  have the following hard distribution for set intersection. 
	
	\begin{Distribution}[The distribution $\cDSI$ for set intersection]\label{dist:SI}
		~
		\begin{enumerate}[label=$(\roman*)$]
			\item Sample a set $A$ uniformly at random from $[\ell]$ such that $\card{A} = \ell/4$. 
			\item Sample an element $\estar $ uniformly at random from $A$. 
			\item Sample set $B$ uniformly at random conditioned on $\card{B} = \ell/4$ and $A \cap B = \{\estar\}$. 
		\end{enumerate}
		
	\end{Distribution}
	
	The work of \cite{AssadiR20} gives an impossibility result for $\eps$-solving set intersection on input distribution $\cDSI$ that will be used later (see \Cref{prop:set-int}). As this result is not relevant to the description of distribution $\cG_p(n,k)$, 
	skip it for now and proceed to describe the set intersection component of $\cG_p(n,k)$. Informally, this component consists of the following:

	\begin{itemize}
		\item A hard instance of set intersection from distribution $\cDSI$, given to $\player{1}$ as Alice and all remaining players collectively as Bob;
		
		\item  Additional sets (a total of $t_p$), also given to $\player{1}$, so that they cannot distinguish which of these sets corresponds to the hard instance described above.
	\end{itemize}
	
	The formal description is given below. All random variables defined in this subsection will also be used in the second component of $\cG_p(n,k)$.

	\begin{Part}[\textbf{First Part of the Distribution $\cG_p(n,k)$ for $n, k \geq 1$ -- Set Intersection}]\label{distpart:set-int}
		~
		\begin{enumerate}[label=$(\roman*)$]
			\item Pick an $(r_p, t_p, k)$-cluster packing graph $\Gclust = (\Vclust, \Eclust)$ on $n_p = n$ vertices from the statement of \Cref{lem:adv-lb-cluster}.\label{part:adv-set-int-cluster}
			
			\item Sample a bit $\ans(G) \in \{0,1\}$ uniformly at random. \label{part:set-int-ans}
			
			
			\item Sample a cluster index $\istar$ uniformly at random from $[t_p]$, and give it to $\player{j}$ for all $j \geq 2$. \label{part:set-int-istar}
			\item Sample sets $S_{\istar}, T \subset [r_p]$ of clique indices uniformly at random such that $\card{S_{\istar}} = \card{T} = r_p/4$ and $\card{S_{\istar} \cap T } = k^{p-1}$.\label{part:set-int-Sstar}
			
			\item Sample sets $S_i \subset [r_p]$ of clique indices for all $i \in [t_p]$ and $i \neq \istar$ uniformly at random and independently such that $\card{S_i} = r_p/4$. 
			
			
			\item For each $i\in [t_p]$ give $S_i$ to $\player{1}$, and give the set $T$ to all players $\player{j}$ for $j \geq 2$.
			
			\item Set $\spec(G) = \{u \mid u \in \cliqnum{\istar}{j}, j \in S_{\istar} \cap T\}$. 
		\end{enumerate}
	\end{Part}


	We will later show, in~\Cref{subsec:adv-specset}, that protocols with low communication cannot gain sufficient information about the intersection of $S_{\istar}$ and $T$.

	\subsubsection{The Index Component}

	We now present the second part of the distribution $\cG_p(n,k)$, whose hardness is derived from the Index problem. This component also specifies the particular edges assigned to $\player{1}$.

		\begin{Part}[\textbf{Second Part of Distribution $\cG_p(n,k)$ for $n, k \geq 1$ -- Index}]\label{distpart:index}
		~
		\noindent
		\begin{enumerate}[label=$(\roman*)$]
		\item Sample a matrix $x \in \{0,1\}^{t_p \times r_p}$ uniformly at random, conditioned on:
		\begin{enumerate}
			\item  $x_{\istar, j} = \ans(G)$ for all $j \in S_{\istar} \cap T$.  (Here, $\ans(G)$, $\istar$, and $S_{\istar}, T$ are from steps \ref{part:set-int-ans}, \ref{part:set-int-istar} and \ref{part:set-int-Sstar} of \Cref{distpart:set-int}, respectively.) \label{part:index-one}
			
			\item For each $i \in [t_p]$, the number of ones and zeroes in row $i$ of matrix $x$ inside the $r_p/4$ columns of set $S_i$ are equal. That is, for all $i \in [t_p]$,  \[
			\card{\set{j \mid j \in S_{i}, x_{i,j} = 1}} = r_p/8.  
			\]
		\end{enumerate}
		
		\item 	The matrix $x$ is given to $\player{1}$. They also get the following edge set:
		\[
		E_1 = \{\textnormal{edges of clique $\cliqnum{i}{j}$} \mid j \in S_i, x_{i, j} = 1, i \in [t_p]\}.
		\]
		\end{enumerate}
		\end{Part} 

	We will show in~\Cref{subsec:adv-specbit} that $\player{1}$ cannot send a lot of information about $\ans(G)$ to $\player{2}$ using a short message, by using the lower bounds for the Index communication problem (\Cref{prop:index-low-prob}).

	\subsubsection{The Hard Instance for the Remaining $(p-1)$ Players}
	
	We now present the final part of the construction of the hard distribution $\cG_p(n_{p},k)$. This step involves sampling a graph from $\cG_{p-1}$ with some correlation with the existing random variables in $\cG_p$, and embedding it into the output graph.

	Before we proceed, we introduce the notion of \textbf{clique joins} within a cluster packing.

	\begin{definition}\label{def:join-cliques}
		Given an $(r, t, k)$-cluster packing graph $G = (V, E)$, the \textbf{join operation} on two cliques $\cliqnum{i}{j}$ and $\cliqnum{i}{j'}$ (for $i \in [t]$, $j, j' \in [r]$, and $j \neq j'$) consists of adding all edges $(u, v)$ for every $u \in \cliqnum{i}{j}$ and $v \in \cliqnum{i}{j'}$.
	\end{definition} 
	
	Informally, this operation connects all vertices between two different specified cliques in the same induced $k$-cluster. 
	We are ready to define the third component. This component fixes all the edges given to players $\player{j}$ for $2 \leq j \leq p$. 
	
	\begin{Part}[\textbf{Last Part of the Distribution $\cG_p(n,k)$ for $n, k \geq 1$ -- Smaller Hard Instance}]\label{distpart:hard-instance}
		~
		\vspace{-15pt}
		\begin{enumerate}[label=$(\roman*)$]	
			\item Sample a mapping $\sigma$ uniformly at random from the vertex set of any graph sampled from $\mathcal{G}_{p-1}(n_{p-1}, k)$ to the set $\{\cliqnum{\istar}{j} \mid j \in T\}$. This is well-defined, since the vertex set of any such graph consists of exactly $n_{p-1} = r_p / 4 = \card{T}$ vertices by~\Cref{eq:params-all}.
			
			\item Sample a graph $G_{p-1} \sim \cG_p(n_{p-1}, k)$, conditioned on two properties: (i) $\ans(G_{p-1}) = \ans(G)$, and (ii) the set $\spec(G_{p-1})$ is mapped by $\sigma$ to $\{\cliqnum{\istar}{j} \mid j \in S_{\istar} \cap T\}$. This is well-defined, since $\spec(G_{p-1})$ has size $k^{p-1}$ and, by construction, the intersection $S_{\istar} \cap T$ also contains exactly $k^{p-1}$ elements. The graph $G_{p-1}$ is sampled independently of all other random variables.
			
			\item Give the mapping $\sigma$ to all players $\player{2}, \player{3}, \ldots, \player{p}$. 
			
			\item Give the inputs corresponding to $G_{p-1}$ to the $(p-1)$ players $\player{2}, \player{3}, \ldots, \player{p}$ in order. That is, the input originally intended for the first player in $G_{p-1}$ is assigned to $\player{2}$, the input for the second player is given to $\player{3}$, and so on.
			
			\item For each edge $(u,v)$ in vertex set of $G_{p-1}$, perform a join operation on the cliques $\sigma(u)$ and $\sigma(v)$, and give the edges from this join operation to the edge set of the player who originally held edge $(u,v)$ of $G_{p-1}$. \label{part:hard-inst-cliqjoin}
		\end{enumerate}
	\end{Part}
	
	This concludes the description of the hard distribution. 
	
	
	\subsubsection{Properties of the Hard Distribution} \label{subsec:adv-properties}
We now prove some properties of the hard distribution. 
		We start with a simple property about the independence of different parts of the input.
		
		\begin{observation}\label{obs:rows-independent}
		For any $j \in [t_p]$ with $j \neq \istar$, the set $S_j$ and row $j$ of the matrix $x$ are independent of all other sets and rows of $x$. Moreover, the set $S_{\istar}$, the set $T$, and row $\istar$ of $x$ are independent of the other rows of $x$ and all sets $S_j$ with $j \neq \istar$. 
		\end{observation}
		\begin{proof}
			In \Cref{distpart:set-int} of the distribution $\cG_p$, the sets $S_{\istar}$ and $T$ are sampled independently of the remaining sets, and each sets $S_j$ for $j \neq \istar$ is sampled independently of all other random variables. In \Cref{distpart:index}, row $j$ of $x$ is sampled based only on the corresponding set $S_j$ for each $j \in [t_p]$, and is therefore independent of all other rows in $x$ and all other sets.
		\end{proof}
		


We also need the following properties about $\spec(G)$ and $S_{\istar} \cap T$. 
		\begin{observation}\label{obs:Sstar-T}
			Distribution of set $S_{\istar} \cap T$ is uniform over all $k^{p-1}$ subsets of $T$.
		\end{observation}
		\begin{proof}
			We can view the construction of the sets $S_{\istar}$ and $T$ as follows: the set $T\subset [r_p]$ is chosen uniformly at random such that $\card{T}=r_p/4$, and the set $S_{\istar}\subset [r_p]$ is also chosen uniformly at random such that $\card{S_{\istar}}=r_p/4$ and $\card{S_{\istar} \cap T } = k^{p-1}$.
			Therefore, the intersection $S_{\istar} \cap T$ is uniformly distributed over all $k^{p-1}$ subsets of $T$, as each such subset occurs with equal probability under the uniform randomly choice of $S_{\istar}$.
		\end{proof}
	
	\begin{observation}\label{obs:spec-uniform}
		For any $p > 1$ and graph $G \sim \cG_{p}(n,k)$, the set $\spec(G)$ is uniform over all elements in its support, i.e., vertices of $k^{p-1}$ many $k$-cliques from a single cluster $i \in [t_p]$.
	\end{observation}
	\begin{proof}
		By construction, we know that $\istar$ and $T\subset [r_p]$ are chosen uniformly at random. We also know that distribution of $S_{\istar} \cap T$ is uniform over all $k^{p-1}$ subsets of $T$ from \Cref{obs:Sstar-T}. As $\spec(G)=\{u \mid u \in \cliqnum{\istar}{j}, j \in S_{\istar} \cap T\}$, we conclude that $\spec(G)$ is uniform over all possible choices it has.
	\end{proof}

%
%
	We prove a key property of our distribution on the chromatic number of its generated graphs
	
	\begin{claim}\label{clm:ans-to-chromatic}
		For any $G \sim \cG_p(n,k)$, 
		\begin{enumerate}[label=$(\roman*)$]
			\item When $\ans(G) = 0$, the graph $G$ has $\chrom(G) \leq  k\cdot p$. 
			\item When $\ans(G) = 1$, $\spec(G)$ forms a $k^{p}$-clique and hence $\chrom(G) \geq k^{p}$.  
		\end{enumerate}
	\end{claim}
	\begin{proof}
		We prove this claim by induction on $p$.
		For the base case where $p=2$, by~\Cref{obs:two-player-clique} we know the claim holds.
		
		Suppose the statement holds for $p-1$. If $\ans(G_{p-1})=0$ then $\chrom(G_{p-1})\leq k\cdot (p-1)$ and when $\ans(G_{p-1}) = 1$, then $\spec(G_{p-1})$ forms a  $(k^{p-1})$-clique. We know that $\ans(G_{p-1}) = \ans(G)$ also by construction.

		If $\ans(G)=1$, by the induction hypothesis, the vertices in $\spec(G_{p-1})$ form a $k^{p-1}$-clique.  Mapping $\sigma$ is chosen so that each vertex of this clique corresponds to one of the sets in $\{\cliqnum{\istar}{j} \mid j \in S_{\istar} \cap T\}$. We know that $\spec(G) = \cup_{j \in S_{\istar} \cap T} \cliqnum{\istar}{j}$, thus all vertices inside $\cliqnum{\istar}{j}$ are connected to all vertices inside $\cliqnum{\istar}{j'}$ for all $j, j' \in S_{\istar} \cap T$ and $j \neq j'$ by the clique join operation in step \ref{part:hard-inst-cliqjoin} of \Cref{distpart:hard-instance}. On the other hand, as $\ans(G)=1$, all the edges inside each clique $\cliqnum{\istar}{j}$ for $j \in S_{\istar} \cap T$ are added by $E_1$ given to $\player{1}$. Therefore, all the vertices inside $\spec(G)$ are connected and form a $k^p$-clique, resulting in $\chi(G)\geq k^p$.

 		If $\ans(G)=0$, by induction we know $\chi(G_{p-1})\leq k\cdot (p-1)$. Properly color $G_{p-1}$ using $k\cdot (p-1)$ colors and use the mapping $\sigma$ to color the corresponding vertices in $\bigcup_{j\in T} C_j^{(i^*)}$ with the same colors. As $\ans(G)=0$, each vertex $\sigma(v)$ is mapped to a clique $\cliqnum{i^*}{j}$ for $j \in T$, and no edges inside these cliques are added by $E_1$. Thus, no monochromatic edge appears in the mapped vertices.  We also know by \Cref{def:cpg} that the cluster packing graph $\Gclust$ used in construction of distribution $\cG_p$ is $k$-colorable. We can use another $k$ colors to properly color the graph $G[V \setminus \bigcup_{j\in T} C_j^{(i^*)}]$. Hence, the graph can be properly colored using $k\cdot (p-1)+k=k\cdot p$ colors, which concludes the proof.
	\end{proof}

	We also need the following simple independence property in distribution $\cG_p$.

	\begin{claim}\label{clm:ind-prop-3}
		The random variables $\ans(G_{p-1})$ and $\spec(G_{p-1})$ are independent of set $T$, index $\istar$ and mapping $\sigma$. 
	\end{claim}
	\begin{proof}
		First, let us argue that $\ans(G_{p-1}) = \ans(G)$ is independent of $T$, $\istar$ and $\sigma$. Initially, the bit $\ans(G)$ is chosen uniformly at random from $\{0,1\}$. The values of $T, \istar$ and $\sigma$ are not correlated with $\ans(G)$ at all in \Cref{distpart:set-int} and \Cref{distpart:hard-instance} respectively. 
		
	Now let us argue that $\spec(G_{p-1})$ is independent of the value of $T$, $\istar$ and $\sigma$ even conditioned on $\ans(G)$. 
		Mapping $\sigma$ is chosen so that $\spec(G_{p-1})$ is mapped to the cliques in $S_{\istar} \cap T$, but by \Cref{obs:Sstar-T}, we know that the set $S_{\istar} \cap T$ is uniform over $T$. 
		Hence, whatever value $\sigma$, $\istar$ and $T$ take, distribution of $\spec(G_{p-1})$ will be unchanged. 
		These statements are true even conditioned on $\ans(G_{p-1})$, which proves the claim.
	\end{proof}

This concludes our subsection about the hard distribution.
	
	\subsection{Player Elimination}\label{subsec:player-elim}
	
	In this subsection, we give a $(p-1)$-player protocol for distribution $\cG_{p-1}(n_{p-1}, k)$ based on any protocol for $\cG_p$, through a player elimination argument. We also establish a lower bound on the communication required to solve  instances from the $p$-player hard distribution.
	The main lemma of this subsection is the following. 

	\begin{lemma}\label{lem:inductive-hardness-Gp}
		For any $p \geq 2$, any deterministic protocol $\prot$ that outputs the value of $\ans(G)$ with probability of success at least, 
		\[
		\frac12 + \frac{p}{6 (p+1)},
		\]
		requires communication at least 
		
		\begin{equation*}
		s_p(n_p) := \min\Big\{\frac{c \cdot r_p \cdot t_p}{k^{(p-1)} \cdot (p+1)^2}, s_{p-1}(n_{p-1})\Big\} \textnormal{ for $p > 2$} \qquad \textnormal{ and }\qquad s_2(n_2) := \frac{c \cdot r_2 \cdot t_2}{9k},
		\end{equation*}
		for some absolute constant $c > 0$.
	\end{lemma}

	\Cref{lem:adv-lb-cluster} follows almost directly from \Cref{lem:inductive-hardness-Gp}.

	\begin{proof}[Proof of \Cref{lem:adv-lb-cluster}]
	We have constructed the hard distributions $\{\cG_a\}_{2 \leq a \leq p}$ based on the cluster packing graphs given in the statement of the lemma. 
	
	Now, let us assume towards a contradiction that there exists a deterministic $p$-player protocol $\prot$ that distinguishes between $\chrom(G) \leq k \cdot p$ and $\chrom(G) > k^p$ for graphs $G \sim \cG_p(n_p, k)$ with probability of success at least $2/3$.
	By~\Cref{clm:ans-to-chromatic}, this protocol can also find the value of $\ans(G)$ as  $\ans(G) = 0$ when $\chrom(G) \leq k \cdot p$ and $\ans(G) = 1$ when $\chrom(G) > k^p$. Thus, the success probability of $\prot$ in correctly computing $\ans(G)$ is at least:
	\[
		\frac23 \geq \frac12 + \frac{p}{6(p+1)}.
	\]
	
	Therefore, by~\Cref{lem:inductive-hardness-Gp}, the communication cost of $\prot$ must be at least 
	\[
		s_p(n_p) = \min\Big\{\frac{c \cdot r_p \cdot t_p}{k^{(p-1)} \cdot (p+1)^2}, s_{p-1}(n_{p-1})\Big\}  = \underset{2 \leq a \leq p}{\min}\Big \{ \frac{c \cdot r_a \cdot t_a}{(a+1)^2 \cdot k^{a-1}}\Big\},
	\]
	where the equality follows from expanding the recursion for $s_{p-1}(n_{p-1})$ down to $p = 2$. 
	\end{proof}

	In our player elimination argument, we crucially rely on the fact that the message sent by $\player{1}$ does not reveal too much information about $\ans(G_{p-1})$ and $\spec(G_{p-1})$. This property is formally stated below
	and is the heart of the whole proof. 

	\begin{lemma}\label{lem:first-msg}
		Message $\rProt_1$ cannot change the distribution of $(\ans(G_{p-1}), \spec(G_{p-1}))$ too much.  Specifically, for any protocol $\prot$ where 
		\[
		CC(\prot) < \cmsg \cdot \frac{r_p \cdot t_p}{(p+1)^2 \cdot  k^{(p-1)}},
		\]
		for some absolute constant $\cmsg>0$, we have
		\[
		\tvd{\cG_p(\ans(G_{p-1}), \spec(G_{p-1}) \mid \rProt_1, \rT,\ristar, \rsigma)}{\cG_p(\ans(G_{p-1}), \spec(G_{p-1}) \mid \rT, \ristar, \rsigma)} \leq \frac{1}{6(p+1)^2}.
		\]
	\end{lemma}
	
	\Cref{lem:first-msg} is proved in~\Cref{subsec:adv-specbit}, and we will use it in this subsection to analyze the $(p-1)$-player protocol that we construct.

	\begin{Algorithm}
	\textbf{Protocol $\prot_{p-1}$ for $p-1$ players for hard distribution $\cG_{p-1}(n_{p-1}, k)$:}

		Input: Graph $G_{p-1} \sim \cG_{p-1}(n_{p-1}, k)$.
		\begin{enumerate}[label=$(\roman*)$]
			\item Sample the following random variables using public randomness: message $\rProt_1$, set $\rT$, index $\ristar$ and mapping $\rsigma$. 
			\item Add edges to the graph $G $ according to \Cref{distpart:hard-instance} of $\cG_p(n,k)$ using $T, \ristar$, $\rsigma$ and the input graph $G_{p-1}$. 
			\item Run protocol $\prot$ from statement of \Cref{lem:inductive-hardness-Gp} starting from $\player{2}$ assuming the first message is $\Prot_1$ from public randomness. 
		\end{enumerate}
	\end{Algorithm}
	
	In protocol $\prot_{p-1}$, when the $(p-1)$ players use $\prot$, the input variables  $\Prot_1, T$, $\istar$ and $\sigma$ are not sampled exactly as they would have been if the input distribution was from $\cG_p$, as there is some correlation between the input of $\player{1}$ in $\cG_p$ and the remaining $(p-1)$ players. This correlation may be large, and it is not feasible to sample the input of $\player{1}$ exactly. However, the players $\player{2}, \player{3}, \ldots, \player{p}$ only require the message $\Prot_1$ sent by $\player{1}$ to continue running $\prot$. \Cref{lem:first-msg} shows that the correlation between $\Prot_1$ and the input of the other players is low, and hence, the joint distribution of $(\rProt_1, \rT, \ristar, \rsigma, \rG_{p-1})$ is close in total variation distance to what it should have been. 

	We use $\cGfake_p$ to denote the distribution of $\rProt, \rT$, $\ristar, \rsigma$ and $\rG_{p-1}$ as generated during the execution of the protocol $\prot_{p-1}$. Let $\cGreal_p$ denote the actual distribution of these random variables if the input is sampled according to the true distribution $\cG_p$, i.e., graph $G \sim \cG_p$ was sampled so that the input to the $(p-1)$ players sampled from $\cG_{p-1}(n_{p-1},k)$ is the graph used to construct $G$ in \Cref{distpart:hard-instance}.

	\begin{claim}\label{clm:dist-close-fake}
		Let $\cmsg$ be the constant from~\Cref{lem:first-msg}. Then, if protocol $\prot$ has communication cost
		\[
		CC(\prot) \leq \cmsg \cdot \frac{r_p \cdot t_p}{(p+1)^2 \cdot k^{p-1}},
		\]
		the distributions $\cGreal_p$ and $\cGfake_p$ are close:
		\[
		\tvd{\cGreal_p(\rProt_1, \rT, \ristar, \rsigma, \rG_{p-1})}{\cGfake_p( \rProt_1, \rT, \ristar, \rsigma, \rG_{p-1})} \leq \frac1{6(p+1)^2}. 	
		\]
	\end{claim}
	
	\begin{proof}
		We prove \Cref{clm:dist-close-fake} using the data processing inequality for total variation distance, from \Cref{fact:tvd-data-processing} and the weak chain rule from \Cref{fact:tvd-chain-rule}. 
		\begin{align*}
			&\tvd{\cGreal_p(\ \rProt_1, \rT, \ristar, \rsigma, \rG_{p-1})}{\cGfake_p(\rProt_1, \rT, \ristar, \rsigma, \rG_{p-1})} \\
			&\leq \tvd{\cGreal_p( \rProt_1, \rT, \ristar, \rsigma)}{\cGfake_p(\rProt_1, \rT, \ristar, \rsigma)} \\
			&\hspace{4mm} + \Exp_{ (\rProt_1, \rT, \ristar, \rsigma) \sim \cGreal_p} \tvd{\cGreal_p(\rG_{p-1} \mid  \rProt_1, \rT, \ristar, \rsigma)}{\cGfake_p(\rG_{p-1} \mid \rProt_1, \rT, \ristar, \rsigma)} \tag{by \Cref{fact:tvd-chain-rule}} \\
			&= 0 + \Exp_{ (\rProt_1, \rT, \ristar, \rsigma) \sim \cG_p} \tvd{\cGreal_p(\rG_{p-1} \mid  \rProt_1, \rT, \ristar, \rsigma)}{\cGfake_p(\rG_{p-1} \mid \rProt_1, \rT, \ristar, \rsigma)} \tag{by definition of $\cGfake_p$} \\
			&= \Exp_{ (\rProt_1, \rT, \ristar, \rsigma) \sim \cG_p} \tvd{\cGreal_p(\rG_{p-1} \mid  \rProt_1, \rT, \ristar, \rsigma)}{\cGfake_p(\rG_{p-1})},
		\end{align*}
		as in $\cGfake_p$, distribution of $\rG_{p-1}$ is chosen independently from all other random variables. 
		We continue as, 
		\begin{align*}
			& \Exp_{ (\rProt_1, \rT, \ristar, \rsigma) \sim \cG_p} \tvd{\cGreal_p(\rG_{p-1} \mid  \rProt_1, \rT, \ristar, \rsigma)}{\cGfake_p(\rG_{p-1})}\\
			&\leq \Exp_{ (\rProt_1, \rT, \ristar, \rsigma) \sim \cG_p} \tvd{\cGreal_p(\ans(G_{p-1}), \spec(G_{p-1}) \mid  \rProt_1, \rT, \ristar, \rsigma)}{\cGfake_p(\ans(G_{p-1}), \spec(G_{p-1}))},  
		\end{align*}
		where we have used \Cref{fact:tvd-data-processing}, as $\rG_{p-1}$ is sampled only based on $\ans(G_{p-1}), \spec(G_{p-1})$ when conditioned on $\rsigma$ in \Cref{distpart:hard-instance} of distribution $\cGreal_p$. 
		Next, we have, 
		\begin{align*}
			& \Exp_{ (\rProt_1, \rT, \ristar, \rsigma) \sim \cG_p} \tvd{\cGreal_p(\ans(G_{p-1}), \spec(G_{p-1}) \mid  \rProt_1, \rT, \ristar, \rsigma)}{\cGfake_p(\ans(G_{p-1}), \spec(G_{p-1}))} \\
			& =  \Exp_{ (\rProt_1, \rT, \ristar, \rsigma) \sim \cG_p} \tvd{\cG_p(\ans(G_{p-1}), \spec(G_{p-1}) \mid  \rProt_1, \rT, \ristar, \rsigma)}{\cG_p(\ans(G_{p-1}), \spec(G_{p-1}))},
		\end{align*}
		as the distribution of these random variables in $\cGreal_p$ and $\cGfake_p$ are exactly the same as in the original distribution. 
		The proof now follows directly from \Cref{lem:first-msg} as protocol $\prot$ has low communication complexity from our premise. 
	\end{proof}
	
	We are ready to prove the main player elimination lemma now.
	
	\begin{proof}[Proof of \Cref{lem:inductive-hardness-Gp}]
		We pick constant $c$ as the minimum of $c_2 $ from \Cref{clm:two-player-basecase} and $\cmsg$ from \Cref{lem:first-msg}.
		We prove the lemma by induction on $p$. \Cref{clm:two-player-basecase} serves as the base case when $p = 2$, as our choice of constant $c \leq c_2$. We assume that the statement is true for $p-1$ as the induction hypothesis. 
		
		Assume towards a contradiction that there exists a deterministic protocol $\prot$ where $\player{p}$ outputs the value of $\ans(G)$ with probability at least $\frac12 + \frac{p}{6(p+1)}$ with communication less than
		\[
			 \frac{c \cdot r_p \cdot t_p}{k^{(p-1)} \cdot (p+1)^2}. 
		\]
		\Cref{clm:dist-close-fake} holds with the upper bound on the communication of $\prot$, and as $c \leq \cmsg$. 
		We analyze the probability of success for protocol $\prot_{p-1}$ on some graph $G_{p-1} \sim \cG_{p-1}(n_{p-1}, k)$. 
		\begin{align}
			&\Pr[\prot_{p-1} \textnormal{ succeeds on $G_{p-1}$}] \notag\\
			&\geq \Pr[\prot_{p-1} \textnormal{ succeeds on $\cGreal$}] - \tvd{\cGreal(\rProt)}{\cGfake(\rProt)}, \label{eq:interim-pe-1}
		\end{align}
		where we have used \Cref{fact:tvd-small} on the event that the protocol succeeds on the two distributions $\cGreal$ and $\cGfake $ of the transcript. 
		
		We know by our assumption that, 
		\begin{equation}\label{eq:interim-pe-2}
		 \Pr[\prot_{p-1} \textnormal{ succeeds on $\cGreal$}]  \geq \frac12 + \frac{p}{6(p+1)},
		\end{equation}
		as the joint distribution of $G_{p-1}, \Prot_1, T, \istar$, and $ \sigma$ that $\prot_{p-1}$ runs $\prot$ on would be the same as if the input was sampled from  $\cG_p(n_p,k)$. As $\ans(G) = \ans(G_{p-1})$, $\prot_{p-1}$ outputs the correct answer whenever $\prot$ is correct. 
		Hence,
		\begin{align*}
				&\Pr[\prot_{p-1} \textnormal{ succeeds on $G_{p-1}$}] \\ 
				&\geq \frac12 + \frac{p}{6(p+1)} -  \tvd{\cGreal(\rProt)}{\cGfake(\rProt)} \tag{by \Cref{eq:interim-pe-1} and \Cref{eq:interim-pe-2}} \\
				&\geq  \frac12 + \frac{p}{6(p+1)} -  \tvd{\cGreal(\rProt_1, \rT, \ristar, \rsigma, \rG_{p-1})}{\cGfake(\rProt_1, \rT, \ristar, \rsigma, \rG_{p-1})} \tag{as $\prot$ is deterministic and $\Prot$ is a function of $\rProt_1, \rT, \ristar, \rsigma$ and $\rG_{p-1}$, and by \Cref{fact:tvd-data-processing}} \\
				&\geq    \frac12 +  \frac{p}{6(p+1)} - \frac1{6(p+1)^2} \tag{by \Cref{clm:dist-close-fake}}\\
				&\geq \frac12 + \frac{p-1}{6p}. \tag{for any $p \geq 0$}
		\end{align*}

	By the easy direction of Yao's minimax principle (namely, an averaging argument), protocol $\prot_{p-1}$ can be made deterministic by choosing an appropriate random string so that we get a deterministic protocol with the same probability of success. The communication cost of $\prot_{p-1}$ is at most that of $\prot$, as the players only send messages by running $\prot$. 
	
	By the induction hypothesis, we get that, $\prot_{p-1}$ (and thus $\prot$) has communciation cost at least $s_{p-1}(n_{p-1})$. 
	To conclude the proof, either the premise of \Cref{clm:dist-close-fake} does not hold, which implies
	\[
		s_p(n_p) >  \frac{c \cdot r_p \cdot t_p}{k^{(p-1)} \cdot (p+1)^2},
	\]
	or we have shown that if \Cref{clm:dist-close-fake} holds, then,
	\[
		s_p(n_{p}) > s_{p-1}(n_{p-1}). 
	\]
	In either case, we have that, for any $p > 2$, 
	\[
	 s_p(n_p) = \min\Big\{\frac{c \cdot r_p \cdot t_p}{k^{(p-1)} \cdot (p+1)^2}, s_{p-1}(n_{p-1})\Big\} ,
	\]
	completing the proof.
	\end{proof}

	It remains to prove \Cref{lem:first-msg} which is the focus of the remaining parts of this section.	

	\subsection{Distribution of the Special Set}\label{subsec:adv-specset}
	
	The first step towards proving \Cref{lem:first-msg} is showing that the distribution of $\spec(G)$ does not change by a large factor even conditioned on the message of the first player. 	
	We perform this step through a reduction from the set intersection problem, and tools used in \cite{AssadiR20}: 
	
	\begin{proposition}[\!\!{\cite[Theorem 1]{AssadiR20}}]\label{prop:set-int}
		In $\cDSI$ (\Cref{dist:SI}) for set intersection over $[\ell]$, 
		\begin{enumerate}[label=$(\roman*)$]
			\item For any $A$ or $B$ sampled from $\cDSI$ both $\distribution{\estar \mid A}$ and $\distribution{\estar \mid B}$ are uniform distributions on sets $A$ and $B$, each of size $\ell/4$ respectively. 
			\item For any $\eps \in (0,1)$, any protocol $\prot$ that internal $\eps$-solves the set intersection problem over the distribution $\cDSI$ has internal information cost $\ic(\prot, \cDSI) \geq \cSI\cdot \eps^2 \cdot \ell$, for some absolute positive constant $\cSI$. 
		\end{enumerate}
	\end{proposition}
	
	For our purpose, we prove the following claim which states that the distribution of $\spec(G_{p-1})$ does not change by too much in protocols with small communication.

	\begin{claim}\label{clm:set-int-hard-basic}
		For any $\eps \in (0,1)$, in any protocol $\prot$ with  $$CC(\prot) \leq \eps^2 \cdot r_p \cdot t_p \cdot \cSI \cdot (1/k^{p-1}),$$ 
		we have, 
		\[
		\tvd{\cG_p(\spec(G) \mid \rProt_1,  \rT, \ristar)}{\cG_p(\spec(G) \mid  \rT, \ristar)} \leq \eps.
		\]
	\end{claim}
	
	We give a reduction from $\eps$-solving set intersection and use \Cref{prop:set-int} to prove \Cref{clm:set-int-hard-basic}. Notice that there is no conditioning on $\rsigma$ in the statement of \Cref{clm:set-int-hard-basic}. We will add this conditioning later in \Cref{clm:spec-set-add-sigma}. 
	
	Assume towards a contradiction that there exists a protocol $\prot$ where the first player sends a message $\Prot_1$ that changes the distribution of $\spec(G)$ by at least $\eps$. We will use $\prot$ to internally $\eps$-solve the set intersection problem on a different universe of size $\ell = r_p / k^{p-1}$. 

	\begin{Algorithm}
		\textbf{Protocol $\protsetint$ for set intersection on $[\ell]$ given $\prot$ for $\cG_p(n,k)$:} 
		
		Input: Alice has set $A$ and Bob has set $B$ from $\cDSI$, each of size $\ell/4$ so that $A \cap B = \{e\}$. 
		\begin{enumerate}[label=$(\roman*)$]
			\item Alice  samples an index $\istar \in [t_p]$ and a random partition $\tau$ of $[r_p]$ into $\ell = r_p/k^{p-1}$ sets $C_i$'s of size $k^{p-1}$ each  using private randomness, namely 
			\[
			\tau = (C_1, C_2, \ldots, C_{\ell})
			\] 
			\item In addition, Alice does the following: 
			\begin{enumerate}
				\item Fix set $S_{\istar} = \cup_{a \in A} C_a$. 
				\item Sample a string $y \in \{0,1\}^{\ell}$ uniformly at random with private randomness such that
				\[
				\card{\set{j \mid y_j = 1, j \in A}} = \card{\set{j \mid y_j = 0, j \in A}},
				\]
				that is, the number of ones and zeroes inside the coordinates of set $A$ are equal. 
				\item Fix row $\istar$ of matrix $x$ as, for all $j \in [r_p]$, $x_{\istar, j} = y_q$ with $j \in C_q$, $q \in [\ell]$. 
				\item For each $i \neq \istar$, and $i \in [t_p]$, sample a set $A_i \subset [\ell]$ of $\ell/4$ elements uniformly at random and independently of each other. Fix $S_i$ and row $i$ of matrix $x$ following earlier three steps, exactly how row $\istar$ and set $S_{\istar}$ were sampled based on set $A$, but using set $A_i$ now in place of set $A$.  
			\end{enumerate}
			\item Alice sends the value of $\tau$ and $\istar$ to Bob. Alice also sends to Bob the message that $\player{1}$ would have sent in protocol $\prot$ from \Cref{clm:set-int-hard-basic}. 
				\item Bob sets value of set $T$ to be $\cup_{b \in B} C_b$. 
		\end{enumerate}
	\end{Algorithm}

	\begin{claim}\label{clm:dist-set-int-same}
		In protocol $\protsetint$, the distribution of the sets $S_1, S_2, \ldots, S_{t_p}$, index $\istar$, matrix $x$ and set $T$ are exactly the same as in $\cG_p$. 
	\end{claim}
	\begin{proof}
		First, we argue about the distribution of $\istar,S_{\istar}$ and $T$. 
		Distribution of $\istar$ is the same in both cases by construction. 
		As the partitioning $\tau$ is chosen uniformly at random, and the set $A$ is also random from $\cDSI$ to begin with, set $S_{\istar}$ is a uniformly random set of $\ell/4 \cdot k^{p-1} = r_p/4$ elements. This is true for $T$ as well. Moreover, as sets $A$ and $B$ have exactly one element $e$ in common, the sets $S_{\istar}$ and $T$ have only the elements corresponding to one group $C_e$ in the intersection, and size of $C_e$ is $k^{p-1}$ by construction. 
		
		Now, we argue about row $\istar$ of matrix $x$. We know that this obeys two conditions:
		\begin{itemize}
			\item All the values corresponding to $x_{\istar, j}$ for $j \in S_{\istar} \cap T$ are the same. This is true in protocol $\protsetint$ because all these values are equal to $y_e$ where $e \in [\ell]$ is the single element in $A \cap B$. 
			
			\item The number of ones and zeroes in row $\istar$ in the co-ordinates contained in set $S_{\istar}$ should be equal to each other. This is true because string $y$ is sampled so that in the co-ordinates present in set $A$, the number of ones and zeroes are equal, and these numbers get multiplied by a $k^{p-1}$ factor corresponding to the group sizes. 
		\end{itemize}
		
		The marginal distribution of row $i$ of matrix $x$ and set $S_i$ are also the same as in $\cG_p$ by the same reasoning for $i \neq \istar$. 
		It remains to argue that the joint distribution of all rows of $x$ and all sets $S_1, \ldots, S_{t_p}, T$ are the same.  By \Cref{obs:rows-independent}, we know that row $j$ of $x$ and set $S_j$ are independent of other rows and sets. We also know $S_{\istar}, T$ and row $\istar$ are independent of other rows and sets. Since we have argued that the marginal distributions are the same, we also have that the joint distribution of these random variables is the same because of this independence. 
	\end{proof}
	
	The next claim shows that the information revealed about the input in protocol $\protsetint$ is a factor $1/t_p$ smaller than the communication of $\prot$. 
	The proof of this claim is standard. 
	
	\begin{claim}\label{clm:set-int-dirsum}
	Protocol $\protsetint$ has low internal information complexity:
		\[
		\ic(\protsetint, \cDSI) \leq \frac1{t_p} \cdot \en{\rProt_1}.
		\]
	\end{claim}
	\newcommand{\rProtSI}{\rProt_{\textnormal{SI}}}
	\newcommand{\rtau}{\ensuremath{\bm{\tau}}}
	\newcommand{\re}{\ensuremath{\rv{e}}}
	\begin{proof}
		The internal information complexity of $\protsetint$ can be written as, 
		\begin{align*}
			\ic(\protsetint, \cDSI) = \mi{\rA}{\rProtSI \mid \rB, \rR},
		\end{align*}
		where $\rProtSI$ denotes the random variable corresponding to the message sent by Alice, and $\rR$ corresponds to the public randomness in $\protsetint$.  (See \Cref{def:ic}.)
		We have, 
		\begin{align*}
			\mi{\rA}{\rProtSI \mid \rB, \rR} &= \mi{\rA}{\rProt_1, \ristar, \rtau \mid \rB, \rR} \tag{expanding $\rProtSI$}\\
			&= \mi{\rA}{\rProt_1, \ristar, \rtau \mid \rB} \tag{as $\prot$ is deterministic, and no pubic randomness is used in $\protsetint$} \\
			& =\mi{\rA}{\ristar, \rtau \mid \rB} + \mi{\rA}{\rProt_1 \mid \rB, \ristar, \rtau} \tag{by \itfacts{chain-rule}} \\
			&= \mi{\rA}{\rProt_1 \mid \rB, \ristar, \rtau}  \tag{as $\ristar, \rtau$ is independent of the joint distribution of $\rA, \rB$, by \itfacts{info-zero}} \\
			&\leq \mi{\rA}{\rProt_1 \mid \ristar, \rtau} \tag{as $\rProt_1 \perp \rB \mid \rA, \ristar, \rtau$ and by \Cref{prop:info-decrease}} \\
			&= \mi{\rS_{\istar}}{\rProt_1 \mid \ristar, \rtau} \tag{as $\rS_{\istar}$ if fixed by $\rA$ and vice-versa when conditioned on $\rtau$} \\
			&=\frac1{t_p} \cdot \sum_{i = 1}^{t_p} \mi{\rS_i}{\rProt_1 \mid  \rtau, \istar =i} \tag{as $\ristar$ is uniform over $[t_p]$ and by definition of conditional mutual information} \\
			&= \frac1{t_p} \cdot \sum_{i = 1}^{t_p} \mi{\rS_i}{\rProt_1 \mid  \rtau} \tag{as value of $\ristar$ is independent of the joint distribution of $(\rS_1, \rS_2, \ldots, \rS_{t_p}, \rProt_1, \rtau)$} \\
			&\leq \frac1{t_p} \cdot \sum_{i = 1}^{t_p} \mi{\rS_i}{\rProt_1 \mid \rS_{< i},  \rtau} \tag{as $\rS_i \perp \rS_{<i} \mid \rtau$ and by \Cref{prop:info-increase}} \\
			&=  \frac1{t_p} \cdot \mi{\rS_1, \rS_2, \ldots, \rS_{t_p}}{\rProt_1 \mid  \rtau} \tag{by \itfacts{chain-rule}} \\
			&\leq \frac1{t_p} \cdot \en{\rProt_1}, \tag{by \itfacts{info-zero}, and as conditioning reduces entropy}
		\end{align*}
		proving the claim.
	\end{proof}
	
	Assuming $\Pi_1$ is a small message, ~\Cref{clm:set-int-dirsum} shows that the information cost of $\protsetint$ will be smaller by a factor of $t_p$. Therefore, it cannot internally $\eps$-solve the problem (by~\Cref{prop:set-int}). We will use this to prove \Cref{clm:set-int-hard-basic}.
	With a slight abuse of notation we use $\cDSI$ to also denote the distribution of all random variables involved in $\protsetint$. 
	
	
	\begin{proof}[Proof of \Cref{clm:set-int-hard-basic}]
		Let us assume towards a contradiction that there exists a protocol $\prot$ with communication less than $\eps^2 \cdot (r_p/k^{p-1}) \cdot t_p \cdot \cSI$, which changes the distribution of $\spec(G)$ by at least $\eps$ in total variation distance. 
		That is, 
		\begin{align*}
			\tvd{\cG_p(\spec(G) \mid \rProt_1 , \rT, \ristar)}{\cG_p(\spec(G) \mid \rT, \ristar)} &> \eps.
			\end{align*}
			This, in turn, implies that,
		\begin{align}
			\tvd{\cDSI(C_{\re} \mid \rProt_1 , \rT, \ristar)}{\cDSI(C_{\re} \mid \rT, \ristar)} &> \eps,\label{eq:inter-1}
		\end{align}
	because by \Cref{clm:dist-set-int-same}, distribution of $\spec(G)$ is the same as the group which element $\re$ of the intersection is mapped to. 
	Next, we use \Cref{fact:tvd-over-conditioning} to say,
		\begin{align*}
			\tvd{\cDSI(C_{\re} \mid \rProt_1 , \rT, \rtau, \ristar)}{\cDSI(C_{\re} \mid \rT, \rtau, \ristar)} 
			\geq \tvd{\cDSI(C_{\re} \mid \rProt_1 , \rT, \ristar)}{\cDSI(C_{\re} \mid \rT, \ristar)}.
		\end{align*}
		We also have, 
		\begin{align*}
			\cDSI(C_{\re} \mid \rProt_1, \rT, \ristar, \rtau) 
			= \cDSI(\re \mid \rB, \rProt_1,  \ristar, \rtau),  \qquad \textnormal{ and } \qquad 
		\cDSI(C_{\re} \mid \rT, \rtau, \ristar) = \cDSI(\re \mid \rB, \ristar, \rtau)
		\end{align*}
		as when conditioned on $\rtau$, $C_{\re}$ is fixed by $\re$ and vice-versa. The set $\rT$ is fixed by $\rB$ and vice-versa also. For the next step, recall that the transcript $\rProtSI$ also contains the random variables $\ristar, \rtau$. Hence, 
		\[
		\cDSI(\re \mid \rB, \rProt_1,  \ristar, \rtau) = \cDSI(\re \mid \rB, \rProtSI). 
		\] 
		Lastly, we know that the random variables $\ristar$ and $\rtau$ are independent of the inputs $\rA, \rB$ and $\re$ in protocol $\protsetint$. Thus, 
		\[
		\cDSI(\re \mid \rB, \ristar, \rtau) = \cDSI(\re \mid \rB). 
		\]
		We can conclude the proof now. 
		\begin{align*}
			\tvd{\cDSI(\re \mid \rB, \rProtSI)}{\cDSI(\re \mid \rB)} 
			&= \tvd{\cDSI(C_{\re} \mid \rProt_1 , \rT, \ristar, \rtau)}{\cDSI(C_{\re} \mid \rT, \ristar, \rtau)} \tag{by our equations above} \\
			&\geq \eps, \tag{by \Cref{eq:inter-1}}
		\end{align*}
		therefore by \Cref{prop:set-int}, we know that $\ic(\protsetint, \cDSI)$ is at least $\eps^2 \cdot \cSI \cdot \ell = \eps^2 \cdot \cSI \cdot r_p/k^{p-1}$. On the other hand, \Cref{clm:set-int-dirsum} says that, 
		$\ic(\protsetint, \cDSI)  < (1/t_p) \cdot \en{\rProt_1} $, which implies,
		\[
		\en{\rProt_1} > t_p \cdot \ic(\protsetint, \cDSI) > \eps^2 \cdot t_p \cdot r_p \cdot \cSI \cdot (1/k^{p-1}),
		\]
		finishing the proof.
	\end{proof}
	
	We have shown that the set $T$, index $\istar$ and message $\Prot_1$ collectively do not change the distribution of $\spec(G) $ by a lot. The mapping $\sigma$ also has some information about $\spec(G)$. In the next claim, we will show that adding the information contained in $\sigma$ cannot change the distribution of $\spec(G)$ by too much either. The difference between the following claim and~\Cref{clm:set-int-hard-basic} is only in the extra conditioning on $\sigma$. 
	
	\begin{claim}\label{clm:spec-set-add-sigma}
		In any protocol $\prot$ for $\cG_p$, we have, 
		\begin{align*}
			&	\tvd{\cG_p(\spec(G) \mid \rProt_1, \rT, \ristar, \rsigma)}{\cG_p(\spec(G) \mid \rT, \ristar, \rsigma)}\\
			&\hspace{20pt} \leq 2 \cdot \tvd{\cG_p(\spec(G) \mid \rProt_1, \rT, \ristar)}{\cG_p(\spec(G) \mid \rT, \ristar)}. 
		\end{align*}
	\end{claim}
	
	\newcommand{\feas}{\ensuremath{\textnormal{\textsc{feas}}}}
	
	Before we can prove \Cref{clm:spec-set-add-sigma}, we need some more definitions and simple observations. 
	We use $W$ to index over choices of $\spec(G)$ in the following part. We also fix a specific choice of message $\Prot$, set $T$ and index $\istar$.
	
	For any set $W$, and mapping $\sigma$, we call $(W, \sigma)$ to be a \textbf{feasible pair} if mapping $\sigma$ takes some set which lies in the support of $\spec(G_{p-1}) \sim \cG_{p-1}$ to the set $W \subset T$. Informally, $(W, \sigma)$ forms a possible pair of values for random variables $(\spec(G), \rsigma)$. We use $\feas(\sigma)$ to denote the set of $W$ which form a feasible pair with $\sigma$.
	
	\begin{observation}\label{obs:feasible-pair}
		For any choice of $\Prot$, $T$ and $\istar$, we have the following:
		\begin{enumerate}[label=$(\roman*)$]
			\item The distribution of $\rsigma$ is independent of $\rProt_1$ conditioned on $\spec(G), T$ and $\istar$. 
			\item The value of $\Pr[\sigma \mid \spec(G) = W, T, \istar]$ is zero for all $\sigma $ which is not feasible with $W$, and is uniform over all $\sigma$ feasible with $W$. 
			\item For every $W \subset T$ of size $k^{p-1}$, the number of mappings $\sigma$ feasible with $W$ is independent of $W$, and is a fixed number.
			\item Conditioned on $T$, $\sigma$ and $\istar$, the value of $\spec(G)$ is uniform over all $W$ feasible with $\sigma$.  
		\end{enumerate}
	\end{observation}
	\begin{proof}
		For part $(i)$, we know that $\rsigma$ is chosen only so that $\spec(G_{p-1})$ is mapped to $S_{\istar} \cap T = W$. Once $W = \spec(G)$ is fixed, $\rsigma$ is independent of the message of the first player, as they do not have access to this mapping. 
		
		For part $(ii)$, it is evident that when $\sigma$ is not feasible, it cannot be chosen as the mapping. Otherwise, by symmetry of the distribution $\cG_{p-1}$, all choices are equally likely. 
		
		For part $(iii)$, this set of mappings feasible with $\sigma$ is all the mappings that take some set in the support of $\spec(G_{p-1})$ to $W$. For any choice of $W$, this number is unchanging. 
		
		Lastly, for part $(iv)$, as distribution of $\spec(G_{p-1})$ is uniform over its support by \Cref{obs:spec-uniform}, we know that $W$ is uniform over its support of choices feasible with $\sigma$ as well. 
	\end{proof}
	
	\begin{proof}[Proof of \Cref{clm:spec-set-add-sigma}]
		Let us fix an arbitrary choice of $\Prot_1$, $T$ and $\istar$. 
		\begin{align*}
			&\Exp_{\sigma \mid \Prot_1, T, \istar}	\tvd{\cG_p(\spec(G) \mid \Prot_1, T, \istar,\sigma)}{\cG_p(\spec(G) \mid T, \istar, \sigma)} \\
			&= \sum_{\sigma} \Pr[\sigma \mid \Prot_1, T, \istar] \cdot \paren{	\tvd{\cG_p(\spec(G) \mid \Prot_1, T, \istar,\sigma)}{\cG_p(\spec(G) \mid T, \istar, \sigma)}} \\
			&=\frac12 \cdot \sum_{\sigma} \Pr[\sigma \mid \Prot_1, T, \istar] \cdot \sum_{W \in \feas(\sigma)} \card{\Pr[W \mid \Prot_1, T, \istar, \sigma] - \Pr[W \mid T, \istar, \sigma ]} \tag{by definition of total variation distance} \\
			&= \frac12 \cdot \sum_{\sigma} \Pr[\sigma \mid \Prot_1, T, \istar] \cdot \sum_{W \in \feas(\sigma)} \card{\Pr[W \mid \Prot_1, T, \istar] \cdot \frac{\Pr[\sigma \mid \Prot_1, T, \istar, W]}{\Pr[\sigma \mid \Prot_1, T, \istar]} - \Pr[W \mid T, \istar, \sigma ]}  \\
			&= \frac12 \cdot \sum_{\sigma} \sum_{W \in \feas(\sigma)} \card{\Pr[W \mid \Prot_1, T, \istar] \cdot \Pr[\sigma \mid T, \istar, W]- \Pr[W \mid T, \istar, \sigma ] \cdot \Pr[\sigma \mid \Prot_1, T, \istar] }. \tag{by part $(i)$ of \Cref{obs:feasible-pair}}
		\end{align*}
		Let $\Pr[\sigma \mid T, \istar, W] = x $ for some $x \in [0,1]$ when $W \in \feas(\sigma)$. By part $(ii)$ and $(iii)$ of \Cref{obs:feasible-pair}, such a value of $x$ exists. 
		\begin{align*}
			\Pr[\sigma \mid \Prot_1, T, \istar] = \sum_{W \in \feas(\sigma)} \Pr[W \mid \Prot_1, T, \istar] \cdot \Pr[\sigma \mid W, \Prot_1, T, \istar] = \sum_{W \in \feas(\sigma)} \Pr[W \mid \Prot_1, T, \istar] \cdot x. 
		\end{align*}
		We use $\beta_{\sigma}$ to denote the following:
		\[
		\beta_{\sigma}:= \sum_{W \in \feas(\sigma)} \card{\Pr[W \mid \Prot_1, T, \istar] - \Pr[W \mid T, \istar]}.
		\]
		We can continue the proof as, 
		\begin{align*}
			&\Exp_{\sigma \mid \Prot_1, T, \istar}	\tvd{\cG_p(\spec(G) \mid \Prot_1, T, \istar,\sigma)}{\cG_p(\spec(G) \mid T, \istar, \sigma)} \\
			&= \frac{x}2 \cdot  \sum_{\sigma} \sum_{W \in \feas(\sigma)} \card{\Pr[W \mid \Prot_1, T, \istar]- \Pr[W \mid T, \istar, \sigma ] \cdot \paren{ \sum_{W \in \feas(\sigma)} \Pr[W \mid \Prot_1, T, \istar]}} \\
			&\leq  \frac{x}2 \cdot  \sum_{\sigma} \sum_{W \in \feas(\sigma)} \Bigg(\card{\Pr[W \mid \Prot_1, T, \istar]- \Pr[W \mid T, \istar]}\\
			&\hspace{6mm} + \card{\Pr[W \mid T, \istar, \sigma ] \cdot \paren{ \sum_{W \in \feas(\sigma)} \Pr[W \mid \Prot_1, T, \istar]} - \Pr[W \mid T, \istar]}\Bigg) \tag{by triangle inequality}\\
			&\leq \frac{x}2 \cdot \sum_{\sigma} \beta_{\sigma} +  \sum_{W \in \feas(\sigma)}\card{\frac1{\card{\feas(\sigma)}} \cdot \paren{ \sum_{W \in \feas(\sigma)} \Pr[W \mid \Prot_1, T, \istar]} - \Pr[W \mid T, \istar]} \tag{by part $(iv)$ of \Cref{obs:feasible-pair}, and definition of $\beta_\sigma$} \\
			&\leq \frac{x}2 \cdot \sum_{\sigma} \beta_{\sigma}	+ \frac1{\card{\feas(\sigma)}} \cdot\sum_{W \in \feas(\sigma)}\card{ \paren{ \sum_{W \in \feas(\sigma)} \Pr[W \mid \Prot_1, T, \istar]} - \Pr[W \mid T, \istar] \cdot \card{\feas(\sigma)}}\\
			&\leq \frac{x}2 \cdot \sum_{\sigma} \beta_{\sigma} + \frac1{\card{\feas(\sigma)}} \cdot\sum_{W \in \feas(\sigma)} \beta_{\sigma} \tag{by triangle inequality and definition of $\beta_{\sigma}$} \\
			&= \frac{x}2 \cdot \sum_{\sigma} 2 \beta_{\sigma} \\
			&=   \sum_{\sigma} \sum_{W \in \feas(\sigma)} \card{\Pr[W \mid T, \Prot_1, \istar] - \Pr[W \mid T, \istar]} \cdot \Pr[\sigma \mid T, \istar, W] \tag{by definition of $\beta_{\sigma}$ and $x$} \\
			&= \sum_{W} \card{\Pr[W \mid T, \Prot_1, \istar] - \Pr[W \mid T, \istar]} \sum_{\sigma: W \in \feas(\sigma)}  \Pr[\sigma \mid T, \istar, W] \\
			&=  \sum_{W} \card{\Pr[W \mid T, \Prot_1, \istar] - \Pr[W \mid T, \istar]} \cdot 1 \tag{$\rsigma \mid T, \istar, W$ is uniform over $\feas(\sigma)$, by part $(ii)$ and $(iii)$ of \Cref{obs:feasible-pair}} \\
			&\leq	2 \cdot \tvd{\cG_p(\spec(G) \mid \Prot_1, T, \istar)}{\cG_p(\spec(G) \mid T, \istar)}. \tag{by definition of total variation distance}
		\end{align*}
		This completes the proof as choice of $\Prot_1, \istar$ and $T$ was arbitrary. 
	\end{proof}
	
	Finally, to conclude this subsection, we show that conditioned on $\rsigma$, the distribution of $\spec(G_{p-1})$ also does not change too much. So far, we have proved that distribution of $\spec(G)$ does not change too much. This claim is fairly direct from the previous statements.
	
	\begin{claim}\label{clm:set-int-final}
		For any $\eps \in (0,1)$, in any protocol $\prot$ with  $$CC(\prot) < \eps^2 \cdot (1/4k^{p-1}) \cdot r_p \cdot t_p \cdot \cSI,$$ we have, 
	\[
	\tvd{\cG_p(\spec(G_{p-1}) \mid \rProt_1,  \rT, \ristar, \rsigma)}{\cG_p(\spec(G_{p-1}) \mid  \rT, \ristar, \rsigma)} \leq \eps.
	\]
	\end{claim}
	
	\begin{proof}
		It is easy to see that $\spec(G)$ is fixed by $\spec(G_{p-1})$ and vice-versa when conditioned on mapping $\rsigma$, set $\rT$ and index $\ristar$.
		The set $\spec(G_{p-1})$, along with $\sigma$, points to $S_{\istar} \cap T$, and $\spec(G)$ is made of all the vertices in all the cliques in $S_{\istar} \cap T$ in induced $k$-cluster $\istar$. 		Hence, we have:
		\begin{align*}
		&\tvd{\cG_p(\spec(G_{p-1}) \mid \rProt_1,  \rT, \ristar, \rsigma)}{\cG_p(\spec(G_{p-1}) \mid  \rT, \ristar, \rsigma)}  \\
		&\hspace{20pt}= \tvd{\cG_p(\spec(G) \mid \rProt_1,  \rT, \ristar, \rsigma)}{\cG_p(\spec(G) \mid  \rT, \ristar, \rsigma)}.  
		\end{align*}
		
	Next, by \Cref{clm:spec-set-add-sigma}, we have, 
		\begin{align*}
	&	\tvd{\cG_p(\spec(G) \mid \rProt_1,  \rT, \ristar, \rsigma)}{\cG_p(\spec(G) \mid  \rT, \ristar, \rsigma)} \\
	&\hspace{20pt}\leq 2 \cdot \tvd{\cG_p(\spec(G) \mid \rProt_1,  \rT, \ristar)}{\cG_p(\spec(G) \mid  \rT, \ristar)}. 
		\end{align*}
		By picking $\eps/2$ as the parameter in \Cref{clm:set-int-hard-basic} and combining the equations above, we get, 
		\begin{align*}
		\tvd{\cG_p(\spec(G_{p-1}) \mid \rProt_1,  \rT, \ristar, \rsigma)}{\cG_p(\spec(G_{p-1}) \mid  \rT, \ristar, \rsigma)} \leq \eps,
		\end{align*}
		in any protocol with communication at most $(\eps/2)^2 \cdot r_p \cdot t_p \cdot \cSI$, completing the proof. 
	\end{proof}
	\subsection{Distribution of the Special Bit}\label{subsec:adv-specbit}
	
	In this subsection, we prove~\Cref{lem:first-msg}. We have already established in~\Cref{clm:set-int-final} that the distribution of $\spec(G_{p-1})$ does not change significantly because of message $\rProt_1$. We now show that the distribution of $\ans(G_{p-1})$ also remains nearly unchanged under this conditioning. Combining these two results will complete the proof of~\Cref{lem:first-msg}.

	Our impossibility result for changing the distribution of $\ans(G_{p-1})$ is stated below. 
	
	\begin{claim}\label{clm:index-hard}
		In any protocol $\prot$ for $\cG_p(n,k)$, with  $$CC(\prot) < \eps^2 \cdot r_p \cdot t_p \cdot \cind/  8 \cdot k^{(p-1)},$$ we have, 
		\[
		\tvd{	\cG_p(\ans(G_{p-1}) \mid \rProt_1, \ristar, \rT, \rS_1, \rS_2, \ldots, \rS_{t_p}, \rsigma)}{	\cG_p(\ans(G_{p-1}) \mid \rT, \ristar, \rS_1, \rS_2, \ldots, \rS_{t_p}, \rsigma)} \leq \eps,
		\]
		where $\cind$ is the constant from \Cref{prop:index-low-prob}.
	\end{claim}
	
	We use the protocol $\prot$ to solve index communication problem on universe size $\ell = r_p/8k^{p-1}$. 
	We use $\rR$ to denote the random variable corresponding to the bits of public randomness in the following protocol, and $\rProtIND$ to denote the random variable corresponding to the message from Alice to Bob. Let $\rI$ denote the random variable indicating the index Bob has in the Index problem.
	\begin{Algorithm}
		\textbf{Protocol $\protindex$ for Index on $[t_p \cdot r_p/4k^{p-1}]$ given $\prot$ for $\cG_p(n,k)$:} 
		
		Input: Alice has string $y \in \{0,1\}^{\ell}$ and Bob has $i \in [\ell]$  from the distribution $\cDind$.

		\begin{enumerate}[label=$(\roman*)$]
			\item Alice sample an index $\istar \in [t_p]$ and set $S_{\istar} \subset [r_p]$ of size $r_p/4$ using private randomness. Alice also samples $\ell$ disjoint subsets of $S_{\istar}$ of size $k^{p-1}$ each using private randomness, namely $\tau_{\istar} = (C_{\istar, 1}, C_{\istar, 2}, \ldots, C_{\istar, \ell})$.\label{part:alice-sampling-partitions}
			\item Alice samples row $\istar$ of matrix $x$ using private randomness, subject to the conditions:
			\begin{enumerate}
				\item \label{item:row-istar-fix-by-Y} For each $a \in [\ell]$ and $b \in C_a$, set $x_{\istar, b} = y_a$. 
				\item The number of ones and zeroes in the columns indexed by $S_{\istar}$ are equal:
				\[
				\card{\set{j \mid j \in S_{\istar}, x_{\istar, j} = 0}} = \card{\set{j \mid j \in S_{\istar}, x_{\istar, j} = 1}}.
				\]
			\end{enumerate} 
			\item For each $j \in [t_p]$ where $j \neq \istar$:
			\begin{enumerate}
				\item Alice privately samples a set $S_j \subset [r_p]$ of $r_p/4$ elements uniformly at random and independently of each other. 
			In addition, Alice privately samples $\ell$ disjoint subsets of $S_j$, of size $k^{p-1}$ each, namely $\tau_j = (C_{j,1}, \ldots, C_{j, \ell})$. 
			\item Alice privately samples a string $y^{(j)} \in \{0,1\}^{\ell}$ uniformly at random, and samples the row $j$ of matrix $x$ privately, exactly how row $\istar$ of $x$ was sampled in step~\Cref{part:alice-sampling-partitions} using input $y$, set $S_{\istar}$ and sets $\tau_{\istar} = (C_{\istar, 1}, \ldots, C_{\istar, \ell})$, but instead with $y^{(j)}$, set $S_j$ and sets $\tau_j = (C_{j,1}, \ldots, C_{j, \ell})$ respectively. 
			\end{enumerate}
			 \item Alice sends a single message to Bob containing the values of $\istar$, and $S_j, \tau_j$ for each $j \in [t_p]$ along with the message that $\player{1}$ would have sent in protocol $\prot$ from \Cref{clm:index-hard}. 
			\item Bob samples $T$ of size $r_p/4$ using private randomness such that $S_{\istar} \cap T = C_{\istar, i}$. 
			\item Bob outputs $1$ as the answer if the probability that $\rY_{\rI} = 1$ conditioned on $ \rProtIND, \rI$ and $\rT$ is greater than $ 1/2$. He outputs $0$ otherwise. 
		\end{enumerate}
	\end{Algorithm}
	
We can analyze the distribution that $\protindex$ runs $\prot$ on.
	\begin{claim}\label{clm:dist-same-index}
		In protocol $\prot_{\textnormal{index}}$, the distribution of the sets $S_1, S_2, \ldots, S_{t_p}$, the index $\istar$, the matrix $x$, and the set $T$ is identical to their distribution in $\cG_p$.
	\end{claim}
	\begin{proof}
		First, we argue about $\istar$, $T$, and $S_\istar$. Distribution of $\istar$ and $S_\istar$ are the same in both cases by construction. For all $a\in[\ell]$ subset $C_{\istar, a}$ of size $k^{p-1}$ is uniformly sampled from $S_\istar$. As Bob samples $T$ such that $T \cap S_\istar = C_{\istar, i}$, and as set $C_{\istar, i}$ is uniformly random from $S_{\istar}$ of size $k^{p-1}$, set $T$ sampled by Bob has the same distribution as in $\cG_p$.
		
		Now, we argue about the row $\istar$ of matrix $x$. We have the following conditions:
		\begin{itemize}
			\item \textbf{Balance in $S_{\istar}$:} The number of ones and zeros in the $r_p/4$ columns indexed by $S_{\istar}$ must be equal in row $\istar$. This is explicitly stated in protocol $\protindex$. To see that such a row can be sampled, note that at most $\ell \cdot k^{p-1} = r_p/8$ co-ordinates in $S_{\istar}$ are fixed due to step~\ref{item:row-istar-fix-by-Y} in $\protindex$. Regardless of how biased these fixed values are, the remaining $r_p/8$ co-ordinates in $S_{\istar}$ remain free and can be set appropriately to balance the number of zeroes and ones across the total $r_p/4$ columns of set $S_{\istar}$.
			
			\item \textbf{Consistency in $S_{\istar} \cap T$:} For all $j\in S_\istar\cap T$, we want $x_{\istar, j}$ to be equal. In protocol $\protindex$, we have $x_{\istar, j} = y_i$ for any $j \in C_{\istar, i}$, and by construction $S_{\istar} \cap T = C_{\istar, i}$. Therefore, this condition is also satisfied. 
		\end{itemize}
		
		For all $j \in [t_p]$ where $j\neq \istar$ the marginal distribution of $S_j$ and row $j$ of matrix $x$ are also the same as in $\cG_p$ by the same reasoning. 
		The joint distribution of these random variables are also exactly as in $\cG_p$ by the independence properties stated in \Cref{obs:rows-independent}.
	\end{proof}
	
	We can also see that the value $\ans(G)$ corresponds to the answer of the Index problem also.
	\begin{observation}\label{obs:ans-index-same}
		In protocol $\protindex$, the value of $\ans(G)$ for the graph $G$ on which $\prot$ is executed is identical to the answer of the Index problem $y_i$.
	\end{observation}
	\begin{proof}
		We know that in distribution $\cG_p$ the single bit that appears in the row $\istar$ of $x$ on indices $j\in S_\istar\cap T$ is $\ans(G)$. By construction, we know that in protocol $\protindex$ we have $x_{\istar, j}=y_i$ for all $j\in C_{\istar, i}= S_\istar\cap T$. Thus, $\ans(G)=y_i$.
	\end{proof}
	
	Finally, we can bound the information complexity of the protocol $\protindex$. 
	\begin{claim}\label{claim:low-info-index}
		Protocol $\protindex$ has low internal information complexity:
		\[
		\ic(\protindex, \cDind) \leq \frac1{t_p} \cdot \en{\rProt_1}.
		\]
	\end{claim}

	\newcommand{\rGamma}{\ensuremath{\rv{\Gamma}}}
	\begin{proof}
	 The internal information complexity of $\protindex$ is 
		\begin{align*}
			\ic(\protindex, \cDind) = \mi{\rY}{\rProtIND \mid \rI, \rR}.
		\end{align*}
		Let $\rtau_j$ be the random variable indicating the subsets $(C_{j, 1}, \ldots, C_{j, \ell})$ for $j \in [t_p]$.  Let $\rGamma$ denote the joint random variable corresponding to $\rS_1, \rtau_1, \rS_2, \rtau_2, \ldots, \rS_{t_p}, \rtau_{t_p}$. We have,
		\begin{align*}
			\mi{\rY}{\rProtIND \mid \rI, \rR} 
			&= \mi{\rY}{\rProt_1, \ristar,\rS_1, \rtau_1, \rS_2, \rtau_2, \ldots, \rS_{t_p}, \rtau_{t_p} \mid \rI,
			\rR} \tag{expanding $\rProtIND$}\\
			&= \mi{\rY}{\rProt_1, \ristar,\rGamma\mid \rI, \rR} \tag{by defintion of $\rGamma$}\\
			&= \mi{\rY}{\rProt_1, \ristar, \rGamma \mid \rI} \tag{as $\protindex$ uses no public randomness, and $\prot$ is deterministic} \\
			&= \mi{\rY}{\rGamma, \ristar \mid \rI} +  \mi{\rY}{\rProt_1 \mid \rGamma, \ristar, \rI} \tag{by \itfacts{chain-rule}} \\
			& = 0 +  \mi{\rY}{\rProt_1 \mid \rI, \ristar,\rGamma} \tag{as $\ristar, \rGamma$ is independent of the joint distribution of $\rY, \rI$, and by \itfacts{info-zero}} \\
			&\leq \mi{\rY}{\rProt_1 \mid \ristar,\rGamma} 
			\tag{as $\rProt_1 \perp \rI \mid \rY, \ristar, \rGamma$ and by \Cref{prop:info-decrease}} \\
			& \leq \mi{\rX_\ristar}{\rProt_1 \mid \ristar, \rGamma} \tag{as $\rY$ is fixed by $\rX_\ristar$ when conditioned on $\rS_{\istar}, \rtau_{\istar}$ inside $\rGamma$, and by \itfacts{data-processing}}\\
			& = \frac1{t_p} \cdot \sum_{j=1}^{t_p}\mi{\rX_j}{\rProt_1 \mid \ristar=j, \rGamma} \tag{as $\istar$ is uniform over $[t_p]$ and by definition of conditional mutual information}\\
			&=\frac1{t_p} \cdot \sum_{j=1}^{t_p}\mi{\rX_j}{\rProt_1 \mid \rGamma}
			\tag{as value of $\istar$ is independent of the joint distribution of $(\rGamma, \rX, \rProt_1)$}\\
			& \leq \frac1{t_p}  \cdot \sum_{j=1}^{t_p} \mi{\rX_j}{\rProt_1 \mid \rX_{<j},  \rGamma} \tag{as $\rX_j \perp \rX_{<j} \mid  \rGamma$ and by \Cref{prop:info-increase}}\\
			&=\frac1{t_p} \cdot \mi{\rX_1, \rX_2, \ldots, \rX_{t_p}}{\rProt_1 \mid  \rGamma} \tag{by \itfacts{chain-rule}} \\
			&\leq \frac1{t_p} \cdot \en{\rProt_1}, \tag{by \itfacts{info-zero}, and as conditioning reduces entropy}
		\end{align*}
		concluding the proof.
	\end{proof}
		
	\begin{proof}[Proof of \Cref{clm:index-hard}]
		Suppose for contradiction that there exists a protocol $\pi$ with communication less than $\eps^2\cdot r_p\cdot t_p/8 \cdot k^{(p-1)}$, which changes the distribution of $\ans(G_{p-1})$ by at least $\eps$ in total variation distance. That is, 
		\begin{equation}
			\tvd{\cG_p(\ans(G_{p-1})\mid \rProt_1, \rT, \ristar,  \rS_1, \ldots, \rS_{t_p}, \rsigma)}{\cG_p(\ans(G_{p-1})\mid \rT, \ristar, \rS_1, \ldots, \rS_{t_p}, \rsigma)} >\eps.\label{eq:interim-index-2}
		\end{equation}
		We will show that when $\protindex$ is run using $\prot$, there is a large probability of success. With a slight abuse of notation, we use $\cDind$ to also denote the distribution of the graph $G $ that $\prot$ is run on, the transcripts and public randomness in $\protindex$. 
	\begin{align}
			&\hspace{-15pt}\cDind(\rY_\rI \mid \rProtIND, \rR, \rI, \rT) \notag\\
			&=\cDind(\ans(G) \mid \rProtIND, \rR, \rI, \rT) \tag{by \Cref{obs:ans-index-same}, $\rY_\rI = \ans(G)$}\\
			&= \cDind(\ans(G) \mid \rProt_1, \ristar, \rS_1, \rtau_1, \ldots, \rS_{t_p},\rtau_{t_p},  \rI, \rT) \tag{expanding $\rProtIND$, and as $\prot$ is deterministic}\\
				&= \cDind(\ans(G) \mid \rProt_1, \ristar, \rS_1, \rtau_1, \ldots, \rS_{t_p},\rtau_{t_p},  \rI, \rT, \rsigma),\label{eq:interim-index-1}
		\end{align}
		where we have used that distribution of $\ans(G)$ is independent of $\rsigma$ when conditioned on $\rProt_1, \rS_{1},\rtau_1$, $\ldots$, $\rS_{t_p}, \rtau_{t_p}$, $ \ristar, \rT$ and $\rI$. The only correlation between $\ans(G)$, matrix $x$, $\rProt_1$ along with all other random variables in the conditioning, and mapping $\rsigma$ is through $\spec(G)$, but this set is fixed by the random variables $\ristar,\rS_{\ristar}$ and $\rT$. The value of $\ans(G)$, matrix $\rX$, $\rProt_1$ and all the random variables in the conditioning barring $\ristar, \rT, \rS_{\ristar}$ are independent of the actual value $\rsigma $ takes. 
	\begin{align}
		&\hspace{-20pt}\cDind(\rY_\rI \mid  \rI, \rT) \notag\\
	 &=\cDind(\ans(G) \mid \rI, \rT) \tag{by \Cref{obs:ans-index-same}, $\rY_\rI = \ans(G)$}\\
	 &= \cDind(\ans(G) \mid \rI, \rT, \ristar, \rGamma) \tag{as $\ristar, \rGamma \perp \ans(G) \mid \rT, \rI$} \\
		&= \cDind(\ans(G) \mid  \ristar, \rS_1, \rtau_1, \ldots, \rS_{t_p},\rtau_{t_p},  \rI, \rT) \tag{expanding $\rGamma$}\\
		&= \cDind(\ans(G) \mid  \ristar, \rS_1, \rtau_1, \ldots, \rS_{t_p},\rtau_{t_p},  \rI, \rT, \rsigma),\label{eq:interim-index-3}
	\end{align}
	where in the last line, we have used the same reasoning as above, to say that $\ans(G)$ is independent of the value of $\rsigma$ when conditioned on $\rS_{\ristar}$, $\ristar$ and $\rT$. 
		\begin{align}
	&\tvd{\cDind(\rY_\rI \mid \rProtIND, \rR, \rI, \rT) }{\cD(\rY_\rI \mid \rR, \rI, \rT) } \notag\\
	&= \tvd{ \cDind(\ans(G) \mid \rProt_1, \ristar, \rS_1, \rtau_1, \ldots, \rS_{t_p},\rtau_{t_p},  \rI, \rT, \rsigma)}{\cDind(\ans(G) \mid  \ristar, \rS_1, \rtau_1, \ldots, \rS_{t_p},\rtau_{t_p},  \rI, \rT, \rsigma)} \tag{by \Cref{eq:interim-index-1} and \Cref{eq:interim-index-3}} \\
	&= \tvd{ \cG_p(\ans(G) \mid \rProt_1, \ristar, \rS_1, \ldots, \rS_{t_p}, \rT, \rsigma)}{\cG_p(\ans(G) \mid  \ristar, \rS_1, \ldots, \rS_{t_p}, \rT, \rsigma)} \tag{by \Cref{clm:dist-same-index}} \\
	&= \tvd{ \cG_p(\ans(G_{p-1}) \mid \rProt_1, \ristar, \rS_1, \ldots, \rS_{t_p}, \rT, \rsigma)}{\cG_p(\ans(G_{p-1}) \mid  \ristar, \rS_1, \ldots, \rS_{t_p}, \rT, \rsigma)} \tag{as $\ans(G) = \ans(G_{p-1})$ in $\cG_p$} \\
	&> \eps, \label{eq:interim-index-4}
	\end{align}
	where, in the last line, we have used the lower bound on total variation distance in \Cref{eq:interim-index-2}.
	
	The probability of success that Bob gains in $\protindex$ via the maximum likelihood estimator is, 
	\begin{align*}
	&\Exp_{(\Prot_{\textnormal{IND}}, i, T) \sim\cDind(\rProtIND, \rI, \rT)}	\max\{\Pr[\rY_i = 1 \mid (\Prot_{\textnormal{IND}}, i, T )], \Pr[\rY_i = 0 \mid (\Prot_{\textnormal{IND}}, i, T )]\} \\
	&= \frac12 + \tvd{\cDind(\rY_\rI \mid \rProtIND,  \rI, \rT)}{\cD(\rY_\rI \mid \rI, \rT)} \tag{as $\rY_\rI \mid  \rI, \rT$ is uniform over $\{0,1\}$} \\
	&\geq \frac12 + \eps. \tag{by \Cref{eq:interim-index-4}}
	\end{align*}

	Therefore, by \Cref{prop:index-low-prob}, we know that 
		\[
		\ic(\protindex, \cDind) \geq \cind \cdot \eps^2 \cdot \ell = \cind \cdot \eps^2 \cdot r_p/8 \cdot k^{p-1}.
		\]
		
		On the other hand, \Cref{claim:low-info-index} says that,
		\[
		\ic(\protindex, \cDind)  < \frac{1}{t_p} \cdot \en{\rProt_1},
		\]
		which implies,
		\[
		\en{\rProt_1} > t_p \cdot \ic(\protsetint, \cDind) > \cSI \cdot \eps^2 \cdot t_p \cdot r_p \cdot (1/8 \cdot k^{p-1}),
		\]
		finishing the proof.
	\end{proof}
	
	
	We are ready to prove \Cref{lem:first-msg}. 
	
	\begin{proof}[Proof of \Cref{lem:first-msg}]
		We want to bound in total variation distance how much $\rProt_1$ changes the distribution of $(\ans(G_{p-1}, \spec(G_{p-1}))$. 
		Fix $\eps := 1/12 (p+1)^2$. Let us start with a protocol $\prot$ that has low communication cost. 
		\begin{align*}
			&	\tvd{\cG_p(\ans(G_{p-1}), \spec(G_{p-1}) \mid \rProt_1, \rT,\ristar, \rsigma)}{\cG_p(\ans(G_{p-1}), \spec(G_{p-1}) \mid \rT, \ristar, \rsigma)} \\
			&\leq \tvd{\cG_p( \spec(G_{p-1}) \mid \rProt_1, \rT,\ristar, \rsigma)}{\cG_p(\spec(G_{p-1}) \mid \rT, \ristar, \rsigma)} \\
			&\hspace{3mm} + \tvd{\cG_p(\ans(G_{p-1})  \mid \spec(G_{p-1}), \rProt_1, \rT,\ristar, \rsigma)}{\cG_p(\ans(G_{p-1}) \mid  \spec(G_{p-1}), \rT, \ristar, \rsigma)} \tag{by \Cref{fact:tvd-chain-rule}} \\
			&\leq \eps + \tvd{\cG_p(\ans(G_{p-1})  \mid \spec(G_{p-1}), \rProt_1, \rT,\ristar, \rsigma)}{\cG_p(\ans(G_{p-1}) \mid  \spec(G_{p-1}), \rT, \ristar, \rsigma)}  \tag{by \Cref{clm:set-int-final}, as long as communication is less than $\eps^2 \cdot r_p \cdot t_p \cdot \cSI/4 k^{p-1}$} \\
			&\leq \eps + \tvd{\cG_p(\ans(G_{p-1})  \mid \rS_1, \rS_2, \ldots, \rS_{t_p}, \rProt_1, \rT,\ristar, \rsigma)}{\cG_p(\ans(G_{p-1}) \mid  \rS_1, \rS_2, \ldots, \rS_{t_p}, \rT, \ristar, \rsigma)}  \tag{by \Cref{fact:tvd-over-conditioning}, as $\spec(G_{p-1})$ is fixed by $\rS_{\istar}, \rT$} \\
			&\leq 2 \eps. \tag{by \Cref{clm:index-hard}, as long as communication is less than $\eps^2 \cdot r_p \cdot t_p \cdot \cind/8 \cdot k^{p-1}$} \\
			&\leq 1/6(p+1)^2. \tag{by choice of $\eps = 1/12(p+1)^2$}
		\end{align*}
	
	The above equations hold if,
	\begin{align*}
		CC(\prot) < \min\{\eps^2 \cdot r_p \cdot t_p \cdot \cind/8 \cdot k^{p-1}, \eps^2 \cdot r_p \cdot t_p \cdot \cSI/4 k^{p-1}\} = r_p \cdot t_p  \cdot \cmsg/ (p+1)^2 \cdot k^{p-1}, 
	\end{align*}
	for choice of $\cmsg = \min\{\cind/1152, \cSI/576\}$. This completes the proof. 
	\end{proof}

	\subsection{Concluding the Proof (\Cref{thm:adv-res})}\label{subsec:thm-adv-res-proof}
	In this subsection, we will use the parameters for cluster-packing graphs in \Cref{sec:stronger-cluster-packing} and plug them into \Cref{lem:adv-lb-cluster} to prove the lower bounds for semi-streaming algorithms. 
	
	\begin{proof}[Proof of \Cref{thm:adv-res}]
		For the first part of \Cref{thm:adv-res}, we use the construction in \Cref{prop:cpg-dense}. We know for all $k, n_a \geq 1$ such that $k=o(\sqrt{\log n_a})$ for all $2 < a \leq p$, there exists an $(r_a, t_a, k)$-cluster packing graph $G$ such that 
		\[
		r_a = \frac{n_a}{2^{\eta_r \cdot (\sqrt{\log n_a\cdot\log k})}} \quad \textnormal{and}  \quad
		t_a = \frac{n_a}{2^{\eta_t \cdot (\sqrt{\log n_a\cdot\log k})}},
		\]
		for some absolute constants $\eta_r, \eta_t$. 
		Starting from large enough $n_p=n$, we get constructions for all $3 \leq a < p$ with  $n_{a-1}=r_a/4$. Thus, by~\Cref{lem:adv-lb-cluster} we know that any deterministic $p$-player protocol $\pi$ that distinguishes between input graph $G$ with $\chrom(G) \leq k\cdot p$ and $\chrom(G)\geq k^p$ with probability at least  $2/3$ has communication cost at least 
		\begin{align}
			\Omega(1)\cdot \min_{2 \leq a\leq p} \{\frac{r_a\cdot t_a}{(a+1)^2\cdot k^{a-1}}\}
			&  = 	\Omega(1)\cdot \min_{2 \leq a\leq p}  \{\frac{(n_a)^2}{(a+1)^2\cdot k^{a-1} \cdot 2^{(\eta_r + \eta_t) \cdot (\log n_a \cdot \log k)^{1/2}}}\} \notag\\
		&	\geq \frac{(n_2)^2}{(p+1)^2 \cdot k^{p-1} \cdot 2^{(\eta_r + \eta_t) \cdot (\log n_p\cdot\log k)^{1/2}}}. \label{eq:interim-3}
		\end{align}
		This is because for each $2\leq a<p$, we have, 
		\[
			n_{a} = \frac{r_{a+1}}{4} = \frac{n_{a+1}}{2^{2 + \eta_r \cdot (\log n_{a+1} \cdot \log k)^{1/2}}} < n_{a+1}.
		\]
		To get \Cref{eq:interim-3}, we have taken the lowest numerator (corresponding to $a = 2$) and the largest denominator (corresponding to $a = p$). 
	
	We also have, 
	\begin{align}
		n_2 =\frac{r_3}{4} = \frac{n_3}{2^{2 + \eta_r \cdot (\log n_3 \cdot \log k)^{1/2}}} \geq \frac{n_p}{2^{p \cdot (2+\eta_r \cdot (\log n_{p} \cdot \log k)^{1/2})}},\label{eq:interim-4}
	\end{align}
	by using the bounds on parameters in \Cref{eq:params-all}.
	Putting \Cref{eq:interim-3} and \Cref{eq:interim-4} together, we get the lower bound on $p$-player deterministic communication protocol $\prot$ to be at least:
	\begin{align*}
CC(\prot) \geq \frac{n_p^2}{2^{2 \log (p+1) + (p-1) \log k + 2(p+1) \cdot (2 + (\eta_r + \eta_t) \cdot (\log n_p \cdot \log k)^{1/2})}} = \frac{n_p^2}{2^{\Theta(p \cdot \sqrt{\log k} \cdot (\sqrt{\log k} + \sqrt{\log n_p}))}}.
	\end{align*}
	
	Parameters $p, k$ obey $p \cdot \sqrt{\log k} = o(\sqrt{\log n_p})$ from statement of \Cref{thm:adv-res}, which gives the required lower bound that $CC(\prot) = n^{2-o(1)}$. 
	
	We just need to ensure that the value of $k$ satisfies the condition in \Cref{prop:cpg-dense} for all $2 \leq a \leq p$. That is, we  need $k = o(\sqrt{\log n_2})$, and this is true because $\log n_2  = \log n_p - o(\log n_p)$ because of from \Cref{eq:interim-4}.


		For the second part of \Cref{thm:adv-res}, we use the cluster packing graphs from ~\Cref{prop:cpg-large-r}.  We know for $n, k\geq 1$ such that $k=o(\frac{\log n}{\log \log n})$, there is a $(r, t, k)$-cluster packing graph $G$ with parameters 
		\[
		r=\frac{n}{2k^2}\cdot (1-o(1)) \quad \textnormal{and} \quad t = n^{\Omega(1/\log\log{n})}. 
		\]
		To get the properties of~\Cref{lem:adv-lb-cluster}, let $n_p=n$, $r_p=r$, and $t_p=t$.
		Then for all $3\leq a \leq p$ we set
		$n_{a-1} = \frac{r_a}{4}$, and $r_{a-1}, t_{a-1}$ based on $n_{a-1}$. 
		Thus, by~\Cref{lem:adv-lb-cluster} and starting from large enough $n$ we get the communication complexity lower bounded by
		\begin{align}
			\Omega(1)\cdot \min_{2 \leq a\leq p} \{\frac{r_a\cdot t_a}{(a+1)^2\cdot k^{a-1}}\}
			\geq \frac{n_2^{1+{\Omega(1/\log \log n_2)}}}{2k^{p+1} \cdot (p+1)^2}.\label{eq:interim-6}
		\end{align}
		This is because we have:
		\[
	n_{a-1} = \frac{r_a}{4} = \frac{n_a}{2k^2 \cdot 4} \cdot (1-o(1)) < n_a, 		
		\]
		for all $2 < a \leq p$, and the density $r_a \cdot t_a = n_a^{1+\Omega(1/\log \log n_a)}$ is increasing with $n_a$. 
		We also have:
		\begin{equation}\label{eq:interim-7}
			n_2 = \frac{n_3}{8k^2} \cdot (1-o(1)) \geq \frac{n_3}{9k^2} \geq \frac{n_p}{(9k^2)^p}.
		\end{equation}
		Hence, the lower bound in \Cref{eq:interim-6} becomes:
		\[
			CC(\prot) \geq \frac{n_p^{1+\Omega(1/\log (\log n_p - p \log k))}}{2^{\Theta(p \log k )}} = n_p^{1+\Omega(1/\log \log n_p)},
		\]
		for parameter $p \cdot \log k = o(\frac{\log n_p}{\log \log n_p})$, as in \Cref{thm:adv-res}. We just have to check if $k$ obeys the parameters required in \Cref{prop:cpg-large-r} for all $2 \leq a \leq p$. This is true because from \Cref{eq:interim-7}, we get $\log n_2 = \Omega(\log n_p - p \log k) = \Omega(\log n_p)$ based on the condition on parameters in the second part of \Cref{thm:adv-res}. 
		
		Lastly, by \Cref{prop:adv-communication-to-stream} and Yao's minimax principle, we can extend the lower bound to any streaming algorithm (either deterministic or randomized).
	\end{proof}

	\clearpage
	

\section{Random Order Streams: Beyond Adversarial Order Lower Bounds}\label{sec:random}
	
In this section, we show that by considering random order streams instead of adversarial ones, we can fully bypass the strong impossibility results established in~\Cref{sec:adversarial}. 
Specifically, we prove the following theorem that formalizes the first part of~\Cref{res:rand}. 

\begin{theorem}\label{thm:rand-upper}
	There is a single-pass streaming algorithm that for any integers $q,t \geq 2$, and any given graph $G$ in a random order stream, 
	with exponentially high probability (over the randomness of the stream), distinguishes between the cases $\chi(G) \leq q$ and $\chi(G) > q^t$ in $\Ot(n^{1+1/t})$ space. 
\end{theorem}

By allowing $t \rightarrow \infty$, we obtain $n^{1+o(1)}$-space streaming algorithms in random order streams for distinguishing $q$-colorable graphs vs ones that are not even $\poly(q)$-colorable. 
This is in sharp contrast with our lower bounds in~\Cref{sec:adversarial} for adversarial order streams. 


As stated earlier in the introduction, the first non-trivial performance of this algorithm is for distinguishing $q$- from $q^2$-colorable graphs (which it achieves in $\Ot(n^{1.5})$ space); our next theorem, formalizing the second part of~\Cref{res:rand}, 
shows the necessity of this $q^2$ ``barrier''. 

\begin{theorem}\label{thm:rand-lower}
	For any integer $q \geq 3$ such that  $q=o(\log^{1/4}{n})$, any single-pass streaming algorithm on random order streams for distinguishing between $q$-colorable graphs
	and $(q^2/4)$-colorable graphs with probability of success at least $1-2^{-q^4}/4$ requires $\Omega({n^2}/{(q^5 \cdot 4^{q^4})})$ space. 
\end{theorem}

\Cref{thm:rand-lower} is strongest for constant values of $q \geq 5$ (for $q \leq 4$, the theorem holds vacuously as in~\Cref{thm:two-player}), although for the entire range of $q=o(\log^{1/4}{n})$, it still rules out $n^{2-o(1)}$ space algorithms that have error probability $n^{-o(1)}$. 

\subsubsection*{Detour: A Multi-Pass Streaming Algorithm} 
Before moving on, we take a quick detour to note that~\Cref{thm:rand-upper} also implies a \emph{multi-pass} streaming
algorithm on \emph{adversarial} streams, namely, an algorithm that reads the stream a few more times before outputting its answer. 

\begin{proposition}\label{prop:multi-pass}
	For any integers $q , t \geq 2$, there is a randomized $t$-pass $\Ot(n^{1+1/t})$ space streaming algorithm given any graph $G=(V,E)$ presented in an adversarial order stream, 
	with exponentially high probability, distinguishes between the cases $\chi(G) \leq q$ and $\chi(G) > q^t$. 
\end{proposition}

While our focus in this paper is solely on single-pass algorithms, we mention~\Cref{prop:multi-pass} in passing as it immediately follows from~\Cref{thm:rand-upper}
and since existence of a multi-pass algorithm explains why our lower bound approach in~\Cref{sec:adversarial} needed to be ``inherently'' single-pass. 
We present the proof of this result in~\Cref{rem:multi-pass}.



\subsection{A Random Order Algorithm via Repeated Sparsification (\Cref{thm:rand-upper})}\label{sec:random-alg}

We now prove~\Cref{thm:rand-upper}. We present the algorithm in a model-independent way and for the following more general problem: finding a $\chi(G)^t$-coloring of a graph $G$ for integer $t \geq 2$; we will then show how this algorithm can be easily implemented in random order streams and can be used to solve the distinguishing problem also and conclude the proof of~\Cref{thm:rand-upper}. 

\subsubsection{A Model-Independent Coloring Algorithm}
We start by presenting the model-independent algorithm which has the main ideas needed to prove~\Cref{thm:rand-upper}. Starting with any graph $G=(V,E)$ and any $t \geq 2$, the algorithm works in $t$ \emph{iterations}: in each iteration, 
it samples $\Ot(n^{1+1/t})$ edges from the graph, finds a proper coloring of these edges using at most $\chi(G)$ colors, and removes any edge that is \emph{not} monochromatic under this coloring and goes to next iteration; at the end, 
the algorithm takes a \emph{product} of these colorings (as in~\Cref{footnote:product}) to output a proper coloring of $G$. The key to the analysis of this algorithm is its \emph{iterative sparsification} nature: in each iteration, 
it reduces the number of remaining edges in the graph by an $n^{1/t}$ factor. 

We now formalize the algorithm and present its analysis. 

\newcommand{\Mi}[1]{\ensuremath{M_{#1}}}
\newcommand{\Hi}[1]{\ensuremath{H_{#1}}}
\newcommand{\Ci}[1]{\ensuremath{C_{#1}}}

\begin{Algorithm}[An algorithm for $\chi(G)^t$-coloring of a given graph $G=(V,E)$ and $t \geq 2$]\label{alg:sparsification}
~\vspace{-4pt}
\begin{enumerate}
    \item Set $\Mi{1}= E$ as the set of all ``currently monochromatic'' edges. 
    \item For $i:=1$ to $t$ \emph{iterations}:
    \begin{enumerate}
        \item Sample $n^{1+{1/t}}\cdot\log n$ edges from $\Mi{i}$ uniformly to obtain a subgraph $\Hi{i}$ of $\Mi{i}$.
        \item Let $\Ci{i}$ be a proper coloring of $\Hi{i}$ with $\chi(\Hi{i})$ colors and let $\Mi{i+1} \subseteq \Mi{i}$ be the set of edges in $\Mi{i}$ that are monochromatic 
        under $\Ci{i}$. 
    \end{enumerate}
    \item Return $C := \Ci{1} \times \ldots \times \Ci{t}$ as the final coloring: the color of any $v \in V$ is a tuple defined as $C(v) := (\Ci{1}(v),\ldots,\Ci{t}(v))$; we can then arbitrarily map these
    tuples to integers. 
\end{enumerate}
\end{Algorithm}

\begin{lemma}\label{lem:coloring-alg1}
	\Cref{alg:sparsification} returns a proper $\chi(G)^t$-coloring of $G$ with exponentially high probability. 
\end{lemma}
\begin{proof}
Firstly, in each iteration $i \in [t]$, we have $\chi(\Hi{i}) \leq \chi(G)$ and thus the algorithm returns a $\chi(G)^t$-coloring in the worst case. Moreover, under the coloring returned by~\Cref{alg:sparsification} no edge in $G \setminus \Mi{t+1}$ will be monochromatic: consider the iteration $i \in [t]$ wherein the edge $e=(u,v)$ is removed from $\Mi{i}$; under the coloring $\Ci{i}$ this 
edge is not monochromatic and so in the final coloring, the tuple of colors for $u$ and $v$ will be different (at least) in index $i \in [t]$. We thus only need to prove that $\Mi{t+1}$ will be empty
which implies the correctness of the algorithm immediately. 

The following claim shows that size of $\Mi{i}$ drops by a factor of $n^{1/t}$ in each iteration. 
\begin{claim}\label{clm:bound-on-monochromatic}
	With exponentially high probability, for every $i \in [t]$, 
\[
	\card{\Mi{i+1}} \leq n^{-1/t} \cdot \card{\Mi{i}}. 
\]
\end{claim}
\begin{proof}
	Fix any function $f: V \rightarrow [n]$ as a `potential' coloring of the sampled graph $\Hi{i}$. Let $M(f)$ denote the set of monochromatic edges in $\Mi{i}$ under the coloring $f$. 
	Suppose that 
\begin{align}\label{eq:multipass-contradiction}
    \card{M(f)} > n^{-1/t} \cdot  \card{\Mi{i}}; 
\end{align}
we show that with exponentially high probability $f$ cannot be a proper coloring of $\Hi{i}$ and then do a union bound over all choices of $f$ to argue that the coloring
$\Ci{i}$ returned in this iteration cannot satisfy~\Cref{eq:multipass-contradiction}, which implies the claim. Over the randomness of $\Hi{i}$ from $\Mi{i}$, we have, 
\begin{align*}
    \Pr\paren{\text{$f$ is a proper coloring of $\Hi{i}$}}  &=\Pr\paren{\text{no edge from $M(f)$ is sampled}}\\
    &\leq \prod_{e\in M(f)}(1-\Pr(e\text{ sampled})) \tag{by the negative correlation of sampling without replacement}\\
    &\leq \exp\paren{-\sum_{e\in M(f)}\Pr(e\text{ sampled})} \tag{as $1-x\leq e^{-x}$ for all $x\in[0, 1]$}\\
    &= \exp\paren{-\frac{\card{M(f)}}{\card{\Mi{i}}} \cdot n^{1+1/t} \cdot \log{n}} \\
    &\leq \exp\paren{-\frac{n^{-1/t} \cdot \card{\Mi{i}}}{\card{\Mi{i}}} \cdot n^{1+1/t} \cdot \log{n}}\tag{by~\Cref{eq:multipass-contradiction}} \\
    &= \exp\paren{-n\log{n}} \leq n^{-1.44n}. 
\end{align*}
A union bound over $n^n$ possible choices of the coloring function $f$ ensures that, with probability at least $1-n^{-0.44n}$, any proper coloring $f$ of $\Hi{i}$ cannot 
satisfy~\Cref{eq:multipass-contradiction}, hence the returned coloring $\Ci{i}$ and $\Mi{i+1}$ satisfy the requirement of the claim, concluding the proof. 
\Qed{clm:bound-on-monochromatic}

\end{proof}

We can now finalize the proof of~\Cref{lem:coloring-alg1}. By applying~\Cref{clm:bound-on-monochromatic} to the first $t-1$ iterations repeatedly, we have that with exponentially high probability, 
\[
	\card{\Mi{t}} \leq \paren{n^{-1/t}}^{t-1} \cdot \card{\Mi{1}} = n^{-(1-1/t)} \cdot \card{E} \leq n^{-(1-1/t)} \cdot n^2 = n^{1+1/t}.
\]
As size of $\Mi{t}$ is smaller than the sampling budget of the algorithm in the last iteration, we have that $\Hi{t} = \Mi{t}$ and thus a proper coloring of $\Hi{t}$ ensures that no monochromatic edges remain and 
so $\Mi{t+1} = \emptyset$ as desired. This concludes the proof. \Qed{lem:coloring-alg1}

\end{proof}

\subsubsection{Random Stream Implementation}

We now show how to implement~\Cref{alg:sparsification} in random order streams and prove~\Cref{thm:rand-upper}. The implementation is as follows: store the first $\Ot(n^{1+1/t})$ 
edges of the stream, find a $\chi(G)$-coloring of these sampled edges, then read the edges of the stream and store any monochromatic edge under the coloring until 
we have $\Ot(n^{1+1/t})$ edges; find a $\chi(G)$-coloring of these edges, consider the product of the colorings, and then continue as before. The algorithm is formally as follows. 

\begin{Algorithm}[An implementation of~\Cref{alg:sparsification} in random order streams.]\label{alg:random-order}
~\vspace{-4pt}
	\begin{enumerate}
		\item Let $C$ be a coloring that initially colors all vertices the same. 
		\item For $i=1$ to $t$ \emph{iterations}: 
		\begin{enumerate}
			\item Read the edges of the stream and store any edge that is monochromatic under $C$ in a subgraph $\Hi{i}$ until $n^{1+1/t}\log{n}$ edges
			are stored (or the stream finishes, in which case we terminate the for-loop);  
			\item Compute a $\chi(\Hi{i})$-coloring $\Ci{i}$ of $\Hi{i}$ and update $C \leftarrow C \times \Ci{i}$.  
		\end{enumerate}
		\item Return $C$ as the final coloring of $G$. 
	\end{enumerate}
\end{Algorithm}

\begin{lemma}\label{lem:alg-random-order}
	\Cref{alg:random-order} returns a proper $\chi(G)^t$-coloring of $G$ with exponentially high probability using $O(n^{1+1/t}\log^2{n})$ space. 
\end{lemma}
\begin{proof}

For the purpose of analysis, for any $i \geq 1$, define $\Mi{i}$ as the monochromatic edges under the existing coloring $C$ in iteration $i$ of the for-loop in~\Cref{alg:random-order}.  
Then, the randomness of the stream ensures that $\Hi{i}$ is chosen uniformly from $\Mi{i}$ (note that none of the edges in $H_1,\ldots,H_{i-1}$ stored before iteration $i$ or skipped due to not being monochromatic in $C$ can be part of $\Mi{i}$). 
This is exactly as in~\Cref{alg:sparsification}; the rest of~\Cref{alg:random-order} is verbatim as~\Cref{alg:sparsification}, hence the entire algorithm is a faithful simulation of~\Cref{alg:sparsification} (technically speaking,~\Cref{alg:random-order} 
may not finish reading the entire stream if the $t$ iterations of the algorithm finishes before all edges of the stream arrive; however, this happens with exponentially small probability, since as shown in the analysis of~\Cref{alg:sparsification}, the probability that $
\Mi{t+1} \neq \emptyset$ is exponentially small). The correctness of the algorithm thus follows from~\Cref{lem:coloring-alg1}. 

As for the space, the algorithm only needs to store $n^{1+1/t}\log{n}$ edges of the subgraph $\Hi{i}$ for the current iteration $i \in [t]$, which only needs $O(n^{1+1/t}\log^2{n})$ space, and a coloring $C$ which can be stored in $O(n\log{n})$ space (since we never need
to use more than $n$ colors and so each color can be stored in $O(\log{n})$ bits). This concludes the proof. 
\end{proof}

\Cref{thm:rand-upper} follows immediately from~\Cref{lem:alg-random-order} using one additional step: we run~\Cref{alg:random-order} and if it outputs a coloring with at most $q^t$-colors, we output $\chi(G) \leq q$, otherwise we output $\chi(G) > q^t$. 
By~\Cref{lem:alg-random-order}, if $G$ is $q$-colorable, this algorithm also outputs $\chi(G) \leq q$ with exponentially high probability; on the other hand, also by the same lemma, the coloring 
returned by the algorithm, with exponentially high probability, is proper; hence, if it finds a $q^t$-coloring, we cannot have $\chi(G) > q^t$, implying the correctness of the algorithm for the distinguishing task as well.  

\begin{remark}\label{rem:multi-pass}
	We can also directly implement~\Cref{alg:sparsification} in $t$ passes of adversarial streams and prove~\Cref{prop:multi-pass}. Basically, implement each iteration of the algorithm by making a separate pass
	over the entire stream and sample $n^{1+1/t}\log{n}$ edges from the monochromatic ones under the current coloring using reservoir sampling. The correctness and space complexity of this algorithm is verbatim as in~\Cref{lem:alg-random-order} 
	and it also needs $t$ passes at most, proving~\Cref{prop:multi-pass}. 
\end{remark}

\subsection{A Lower Bound via Robust Two-Party Communication (\Cref{thm:rand-lower})}\label{sec:random-lower}

We now switch to proving~\Cref{thm:rand-lower}. The proof, similar to our warm-up in~\Cref{thm:two-player}, is based on two-player communication complexity; however, in order to be able to deduce a lower bound
for random order streaming algorithms, we need to consider the \emph{robust} communication model of~\cite{ChakrabartiCM08}. We introduce this model first and then present our proof. 

\subsubsection{Robust Two-Player Communication Complexity}\label{sec:robust}

Consider a twist in the two-player communication model in~\Cref{sec:two-player-background}, wherein instead of partitioning the input between the players adversarially, we use a random partition. 

Specifically for any graph problem $\mathbb{P}$, we have an adversarially chosen graph $G=(V,E)$; we partition its edges randomly into $E_A$ and $E_B$ given to Alice and Bob, respectively, by sending each edge to one of the sets chosen uniformly and independently. As before, Alice needs to send a single message to Bob, and Bob should output the answer to $\mathbb{P}(G)$. The following proposition, due to~\cite{ChakrabartiCM08}, 
extends~\Cref{prop:stream-cc} to this new model. 

\begin{proposition}[\!\!\cite{ChakrabartiCM08}]\label{prop:robust-stream-cc}
	Any $s$-space random order streaming algorithm $A$, for a problem $\mathbb{P}$ on graphs $G$, implies a communication
	protocol $\pi$ in the robust communication complexity model with $CC(\pi) = O(s)$ for $\mathbb{P}$ with the same success probability. 
\end{proposition}
\begin{proof}
	Given a random partition $E_A$ and $E_B$ of $G$ in the robust communication model, Alice and Bob can form the stream $\sigma = \sigma_A \circ \sigma_B$ where $\sigma_A$ and $\sigma_B$ are random permutations of 
	$E_A$ and $E_B$, respectively. They can then run $A$ on $\sigma$ exactly as in~\Cref{prop:stream-cc} with $O(s)$ communication and since $\sigma$ is a random permutation of $E$---given the random partitioning of $E$ into $E_A$ and $E_B$, and 
	the individual random ordering of $\sigma_A$ and $\sigma_B$---$A$ solves $\mathbb{P}$ with the same probability of success as before, leading to the desired communication protocol. 
\end{proof}

\Cref{prop:robust-stream-cc} allows us to prove random order streaming lower bounds by proving robust communication lower bounds instead. 

\paragraph{Robust communication complexity of the Index problem.} We also need the following result due to~\cite{ChakrabartiCM08} for the Index problem---defined in~\Cref{sec:two-player-background}---in the robust communication model. 
Let $a,b,m \geq 1$ be known integers. In $Index^{a,b}_m$ problem, we have a string $x \in \set{0,1}^m$ and $i \in [m]$ where each $x_j \in \set{0,1}$ for $j \in [m]$ is copied $a$ times and index $i \in [m]$ is copied $b$ times. 
We partition each copy of $x_j$ or $i$ independently and uniformly between Alice and Bob. Alice needs to send a single message to Bob and Bob outputs $x_i$. 

\begin{proposition}[\!\!\cite{ChakrabartiCM08}]\label{prop:robust-index}
	For any $a,b,m \geq 1$, any one-way communication protocol $\prot$ for $Index^{a,b}_m$ in the robust communication model with error probability at most $\delta = 2^{-(a+b)}/4$ has
	\[
		CC(\pi) = \Omega(m) \cdot (1-H_2(\delta)) = \Omega(m \cdot 4^{-(a+b)})~bits. 
	\]
\end{proposition}
We note that strictly speaking~\cite{ChakrabartiCM08} proves~\Cref{prop:robust-index} only for constant values of $a,b$ (in their Theorem 5.3) but the above bound for all choices of $a,b \geq 1$ follows immediately from their proofs as 
we show in~\Cref{app:missing}.

\subsubsection{A Robust Two-Player Communication Lower Bound for Coloring}\label{sec:robust-coloring} 

We can now use~\Cref{prop:robust-index} in a similar reduction as the one in our proof of~\Cref{thm:two-player} to prove the following result 
in the robust communication model. 

\begin{lemma}\label{lem:two-player-robust}
	For any integer $q\geq 3$ such that $q = o(\log^{1/4}(n))$, any two-player one-way communication protocol in the robust communication model that distinguishes $q$-colorable from $(q^2/4)$-colorable graphs with probability of success at least
	$1-2^{-q^4}/4$ has communication cost $\Omega(n^2/(q^5 \cdot 4^{q^4}))$ bits. 
\end{lemma}

\Cref{thm:rand-lower} follows directly from~\Cref{lem:two-player-robust} and~\Cref{prop:robust-stream-cc}. We now prove~\Cref{lem:two-player-robust}.

\begin{proof}[Proof of~\Cref{lem:two-player-robust}]
	Let $\protcolor$ be a protocol in the robust communication model that distinguishes $q$-colorable from $(q^2/4)$-colorable graphs with probability of success at least $1-2^{-q^4}/4$.

	Let graph $G^*$ be a graph obtained by~\Cref{lem:graph-simple} for parameter $k:=q/2$, and let $t=\Omega(n^2/q^5)$ be the number of  specified subgraphs in $G^*$. Let these subgraphs themselves be $H_1, H_2, \ldots, H_t$. 
	 Now consider the $Index^{a,b}_m$ problem of $x\in\{0, 1\}^t$ and $i\in [t]$ with parameters 
	\begin{equation}\label{eq:a-b-values}
		a = k \cdot \binom{k}{2} \qquad \textnormal{and} \qquad b = k^2 \cdot \binom{k}{2}.
	\end{equation}
	We design a protocol $\protindex$ for this problem using $\protcolor$.

\newcommand{\xcopy}[2]{\ensuremath{x_{#1}^{#2}}}
\newcommand{\icopy}[1]{\ensuremath{i^{#1}}}
	
	We use the same reduction as in \Cref{alg:protindex-adv} to construct the graph $G = (V, E)$ on which protocol $\protcolor$ is run based on input $x$ and $i$. 
	It is useful to explicitly state this reduction for our purposes as the partitioning of the edges is different. 
	The edges $E$ will be:
	\begin{align*}
		E_x &:= \set{\textnormal{ edges of $H_{\ell}$} \mid \ell \in [t], x_{\ell} = 1}, \\
		E_i &:= \set{(u,v) \mid u, v\textnormal{ are vertices in two different cliques in }H_i},\\
		E &:= E_x \cup E_i. 
	\end{align*}
	


	\begin{Algorithm}\label{alg:prot-robust-index}
	Protocol $\protindex$ for $Index^{a,b}_t$ of $x \in \set{0,1}^t $ and $ i \in [t]$ given the protocol $\protcolor$: 
		~\vspace{-0.5cm}
		\begin{enumerate}

			\item For all $j\in [t]$, partition the edges of $H_j$ between Alice and Bob using public randomness, ensuring that each player receives a number of edges in $H_j$ equal to the number of copies of $x_j$ they hold.
			If $x_j=1$, each player adds the chosen edges to the graph and if $x_j = 0$, the players keep $H_j$ an entirely empty subgraph. 	
		
			\item For index $i$, partition the edges of $E_i$ between Alice and Bob using public randomness, ensuring that each player receives a number of edges in $E_i$ equal to the number of copies of $i$ they hold.
			
			\item Let $E_B$ be the edge set that Bob holds, and $E_A$ be the edge set of Alice.
			
			\item Alice runs protocol $\protcolor$ on $G_A=(V, E_A)$ and sends the message $M(G_A)$ to Bob. Bob continues to run protocol $\protcolor$ on edges $E_B$.  
			\item Bob outputs $x_i=1$ if $\protcolor$ declares $\chi(G)\geq q^2/4$ and otherwise outputs $0$.
		\end{enumerate}
	\end{Algorithm}
	
	We argue that each edge from the set $E$ is present in exactly one of $E_A, E_B$ with equal probability and independently of other edges. This will show that the partition of edges of $E$ is uniformly random between the two players Alice and Bob, 
	as required in the robust communication model.  

	First we argue about edges in $E_x$. For each $j\in [t]$, where $x_j=1$, if Alice holds all copies of $x_j$ then all of $H_j$ goes to Alice's edges $E_A$, and if Bob holds all copies, then all of $H_j$ goes to Bob's edges $E_B$. 
	And, if Alice holds $0 < \ell < a$ copies of $x_j$ and Bob holds $a-\ell$ many of them, for each $e\in H_j$ either Alice gets this edge or Bob does (using the random partitioning with public randomness). 
	In this case, for each edge $e \in H_j$: 
	\[
		\Pr\paren{\text{Alice receives $e \in H_j$ in $E_A$} \mid \text{$\ell$ copies of $x_j$ are sent to Alice}} = \frac{{a-1 \choose \ell-1}}{{a \choose \ell}},
	\]
	using the randomness of the partitioning via public randomness. Moreover, given the random partitioning of copies of $x_j$ in $Index^{a,b}_t$, the probability that 
	Alice receives $\ell$ copies of $x_j$ is $\binom{a}{\ell}/2^a$.  Thus, in total the probability that Alice gets an edge $e\in H_j$ is 
	\begin{align*}
		\Pr\paren{\text{Alice gets $e \in H_j$}}
		= \sum_{\ell=1}^{a} \frac{{a \choose \ell}}{2^a}\cdot\frac{{a-1 \choose \ell-1}}{{a \choose \ell}}
		= \frac{1}{2^a} \sum_{\ell=1}^{a} {a-1 \choose \ell-1}
		=\frac{1}{2^a}\cdot 2^{a-1}=\frac{1}{2}.
	\end{align*}
	The same holds for Bob.	This same argument is true for $E_i$. Thus, all edges in $E$ have been partitioned randomly in this reduction. The independence also follows by the independence across
	the copies of entries in $(x,i)$ and the public randomness of players. 

	
	The correctness of the reduction is already established in \Cref{lem:two-player-robust}. We get that the error probability of $\protindex$ is exactly the same as the error probability of $\protcolor$. Thus, the probability of error of $\protindex$ is at most 
	$\delta = 2^{-q^4}/4 \leq 2^{-(2k)^4}/4  \leq 2^{-(a+b)}/4$ based on values of $a, b$ in \Cref{eq:a-b-values} .

	Thus, using~\Cref{prop:robust-index}, we get 
	\begin{align*}
		CC(\protcolor) \geq CC(\protindex)= \Omega(t)\cdot (4^{-a+b}) = \Omega(n^2/(q^5 \cdot 4^{(2k)^4})) = \Omega(n^2/(q^5 \cdot 4^{q^4})),
	\end{align*}
 concluding the proof.
\end{proof}

	\clearpage
	

\section{Dynamic Streams: Nearly Optimal Bounds}\label{sec:dynamic}

We now switch to the ``algorithmically hardest'' version of the model, namely, dynamic streams that allow both insertion and deletion of edges. A dynamic graph stream, introduced first by~\cite{AhnGM12} is a 
stream  $\langle a_1,a_2,\ldots,a_t \rangle$ defines a multi-graph $G=(V,E)$ on $n$ vertices $V := [n]$. For any $i \in [t]$, each $a_i$ is a triple $a_i := (u_i,v_i,\Delta_i)$ where $u_i,v_i \in V$ and $\Delta_i \in \set{-1,+1}$. 
The multiplicity of an edge $(u,v)$ is defined to be: 
\[
	E(u,v) := \sum_{a_i: (u_i,v_i) = (u,v)} \hspace{-15pt} \Delta_i.
\]
The multiplicity of every edge is required to be always non-negative during the stream. 

The following two theorems formalize~\Cref{res:dynamic} 
by establishing (nearly) optimal upper and lower bounds (up to logarithmic factors) on the space complexity of the coloring problem in this model. 


\begin{theorem}\label{thm:dynamic-upper}
	There is a randomized single-pass streaming algorithm that for integers $q,t \in \IN$, given any multi-graph $G=(V,E)$ with $\poly(n)$ edge multiplicities in a dynamic stream, 
	with high probability, distinguishes between $\chi(G) \leq q$ vs $\chi(G) \geq q \cdot t$ in $\Ot(n^2/t^2)$ space.  
\end{theorem}

\begin{theorem}\label{thm:dynamic-lower}
	Any randomized single-pass streaming algorithm that for $3 \leq q \leq \poly\!\log{(n)}$ and any $t \in \IN$, given any multi-graph $G=(V,E)$ with $O(n)$ edge multiplicities in a dynamic stream, 
	with high constant probability distinguishes between $\chi(G) \leq q$ vs $\chi(G) \geq q \cdot t$ requires $\Omgt(n^2/t^2)$ space. 
\end{theorem}

\Cref{thm:dynamic-upper} is interesting for all choices of $q > 2$ but for $t$ larger than some $\poly\!\log{\!(n)}$ as for smaller values of $t$, one can simply store the entire (multi-)graph in $\Ot(n^2)$ space and solve the problem exactly. 
Our lower bound in~\Cref{thm:dynamic-lower} also shows for all choices of $q > 2$ bounded by $\poly\!\log{\!(n)}$ (the regime we focus in our paper), the algorithm in~\Cref{thm:dynamic-upper} is optimal up to $\poly\!\log{\!(n)}$ factors in the space\footnote{For $q=2$, the bipartiteness testing algorithm of~\cite{AhnGM12} distinguishes between $2$-colorable vs $>2$-colorable graphs in $\Ot(n)$ space; hence, the requirement of $q > 2$ in the lower bound is crucial.}. 

We also note that our lower bound in~\Cref{thm:dynamic-lower} rules out many standard techniques like sampling and sketching algorithms that are used quite frequently in insertion-only streams also; it thus suggests 
that to obtain better algorithms in insertion-only (adversarial order) streams than the one implied by~\Cref{thm:dynamic-upper}, one needs to come up with techniques that crucially use the insertion-only aspects of the stream. 

 Finally, we have the following important remark regarding the definition of the model and in particular the nature of our lower bound in~\Cref{thm:dynamic-lower}. 
 
 \begin{remark}\label{rem:sim-dynamic}
 	The definition of dynamic streams does \underline{not} limit the length of the stream nor the edge frequencies during the stream (only at the end). Almost all algorithms in the literature follow this definition 
	and the only exception we are aware of belongs to~\cite{KallaugherP20} for a rather non-standard problem. This definition allows one to use the characterization of~\cite{AiHLW16} (see also~\cite{LiNW14}) 
	in dynamic streams to prove lower bounds on space complexity of algorithms via a certain simultaneous communication model (see~\Cref{sec:sim}); this characterization is used 
	in many previous lower bounds (see, e.g.~\cite{Konrad15,WeinsteinW15,AssadiKLY16,AssadiKL17}) and we use it in this paper also. It remains an interesting 
	question to prove~\Cref{thm:dynamic-lower} without relying on the characterization of~\cite{AiHLW16} and thus possibly for shorter streams; see, e.g.~\cite{DarkK20} for more discussion on this topic. 
 \end{remark}

\subsection{A Dynamic Stream Algorithm via Random Subgraphs (\Cref{thm:dynamic-upper})}\label{sec:dynamic-algo}

Our dynamic streaming algorithm follows almost immediately from a purely combinatorial sampling lemma for graph coloring that we prove first. We then
show how to obtain the streaming algorithm. 

\subsubsection{Chromatic Number of Random Induced Subgraphs}\label{sec:vertex-induced}

We prove that random induced subgraphs of a graph with ``large'' chromatic number, cannot have a ``very small'' chromatic number with a large probability. This result is reminiscent of similar 
bounds for maximum matching size in~\cite{AssadiKL17} or minimum set cover in~\cite{AssadiKL16} although the actual bounds and  the proof techniques are not necessarily related. 

\begin{lemma}\label{lem:vertex-sample}
    Let $G$ be an $n$-vertex graph and $H$ be an induced subgraph of $G$ obtained by sampling each vertex of $G$ independently with probability $p \in (0, 1)$ such that $(2\ln{n}/p)$ is an integer. Then,
    \begin{align*}
    \Pr\paren{\chi(H) < {\frac{p}{2\ln{n}}} \cdot \chi(G) - 1}\leq  \frac{1}{2}.
    \end{align*} 
\end{lemma}
\begin{proof}
	Define 
	\[
		t := {\frac{p}{2\ln{n}}} \cdot \chi(G).
	\]
    	To prove this lemma we need the following claim, which shows that the inclusion of any one fixed vertex in the subgraph $H$ cannot increase
	its chromatic number by much (stochastically). 
     \begin{claim}\label{claim:prob-removingVertex}
        For a sampled subgraph $H$ as in~\Cref{lem:vertex-sample} and any vertex $v\in V$, we have
        \[
            \Pr\Paren{\chi(H)< t \mid v\in H} \geq \Pr\Paren{\chi(H)< t-1}.
        \]
     \end{claim}
     \begin{proof}
    Consider $H \setminus \set{v}$, the graph obtained by removing $v$ from $H$ (if it exists). Then
    \begin{align*}
        \Pr\Paren{\chi(H) < t  \mid v \in H} \geq \Pr\Paren{\chi(H \setminus \set{v}) < t-1 \mid v \in H} = \Pr\Paren{\chi(H \setminus \set{v}) < t-1};
    \end{align*}
    here, the inequality holds because conditioned on $v \in H$ whenever $\chi(H \setminus \set{v}) < t-1$, we necessarily have $\chi(H) < t $ also since we can just color $v$ in $H$ using a different color than all of $H \setminus \set{v}$; 
    and, the equality holds because distribution of $H \setminus \set{v}$ is a random subset of $G \setminus \set{v}$ (by picking each vertex independently with probability $p$) regardless of whether or not $v \in H$. 
    Finally, as $\chi(H \setminus v) \leq \chi(H)$, 
    \[
        \Pr\Paren{\chi(H \setminus \set{v}) < t-1} \geq \Pr\Paren{ \chi(H) < t-1},
    \]
    proving the claim. \Qed{claim:prob-removingVertex} 
    
    \end{proof}
    We now continue the proof of~\Cref{lem:vertex-sample}. 
    Suppose by contradiction that the lemma is false and thus we have, 
    \begin{align}\label{eq:prob-sample-contradiction}
    \Pr\Paren{\chi(H) < t-1} > \frac{1}{2}.
    \end{align}      
    Define $\mu$ as the distribution of subgraph $H$ \emph{conditioned} on the event that $\chi(H) < t$ (we emphasize that this is $t$ and not $t-1$ unlike in~\Cref{eq:prob-sample-contradiction}). 
    Sample 
    \[
    	k := {\frac{2\ln{n}}{p}}
    \]
    independent subgraphs $H_1,\ldots,H_k$ from $\mu$ (by the lemma assumption, $k$ is an integer). We prove that there exist $k$ such sampled subgraphs that every vertex $v\in V$ has been sampled at least once among them and use this to properly color
    $G$ with less than $\chi(G)$ colors, a contradiction. 
        
    Consider any vertex $v \in V$ and any $i\in[k]$; we have,
    \begin{align*}
        \Pr_{\mu}\Paren{v\in H_i} &= \Pr\Paren{v \in H \mid \chi(H) < t} \tag{in RHS and throughout this part, $H$ is a subgraph sampled as in the lemma statement} \\
        &=\frac{\Pr \Paren{v\in H} \cdot \Pr\Paren{\chi(H)< t \mid v \in H}}{\Pr\Paren{\chi(H)< t}} \tag{by Bayes' rule}\\
        &\geq\frac{\Pr \Paren{v\in H} \cdot \Pr\Paren{\chi(H)< t-1}}{\Pr\Paren{\chi(H)< t}} \tag{by~\Cref{claim:prob-removingVertex} to drop the conditioning on $v \in H$}\\
        &> \frac{p}{2} \tag{by~\Cref{eq:prob-sample-contradiction} in the nominator and upper bounding denominator by one}. 
    \end{align*}
    So in sampling of $H_1,\ldots,H_k \sim \mu^k$, the probability that $v\in V$ is not sampled in any subgraph is
    \begin{align*}
    \Pr_{\mu^k}\Paren{v \notin \cup_{i=1}^t H_i} &= \prod_{i=1}^{k} \Pr_{\mu}\Paren{v \notin H_i} \tag{by the independence of samples $H_1,\ldots,H_t$} \\
    &< \paren{1-\frac{p}{2}}^{k} \tag{by the bound established in the previous equation} \\
    &\leq \exp\paren{-\frac{p}{2} \cdot \frac{2\ln{n}}{p}} \tag{by the choice of $k$ and since $1-x \leq e^{-x}$ for $x \in [0,1]$} \\
    &= \frac{1}{n}.
    \end{align*}
    Thus, by union bound over all vertices, we get that the probability of having a vertex that is not sampled is less than $1$ (given the strict inequality above). Thus, there exists a set of samples $H_1, \ldots, H_k$ such that every vertex shows up in at least one of 
    these samples. Moreover, we have that $\chi(H_i) < t$ for all $i \in [k]$. If we color each of these subgraphs using a distinct palette of colors (by assigning each vertex to any one of the subgraphs it appears in arbitrarily), 
    we obtain a proper coloring of $G$ with 
    \[
    	\sum_{i=1}^{k} \chi(H_i) < k \cdot t = \paren{\frac{2\ln{n}}{p}} \cdot \paren{\frac{p}{2\ln{n}} \cdot \chi(G)} = \chi(G), 
    \]
    namely, a proper coloring of $G$ with less than $\chi(G)$ colors, a contradiction with the definition of the chromatic number. Thus,~\Cref{eq:prob-sample-contradiction} cannot hold
    and we  conclude the proof of \Cref{lem:vertex-sample}.  
\end{proof}

\subsubsection{The Algorithm} 

Equipped with~\Cref{lem:vertex-sample}, we can design the following simple algorithm for~\Cref{thm:dynamic-upper}. The algorithm involves sampling $O(\log{n})$ random vertex induced subgraphs for $p \sim (\ln{n})/t$ (for the given
parameter $t$ in~\Cref{thm:dynamic-upper}), computing the chromatic number of each subgraph and returning ``large'' chromatic number case if chromatic number of any subgraph is larger than $q$, and returning ``small'' otherwise. 
Formally, the algorithm is as follows. 

\begin{Algorithm}[A dynamic streaming algorithm given $G=(V,E)$, $q \geq 2$, and $t \geq 4\log{n}$]\label{alg:dynamic}
	~\vspace{-4pt}
	\begin{enumerate}
		\item Let $p := (4\ln{n})/t$ and for $i=1$ to $k := 2\log{n}$ times \underline{in parallel}: 
		\begin{enumerate}
			\item Sample $V_i$ by picking each vertex independently with probability $p$ from $V$;
			\item For every pair of vertices $u,v \in V_i$ maintain a counter that counts the number of insertion and deletion of the edge $(u,v)$ mod $n^{\log{n}}$. 
			\item At the end of the stream, let $H_i$ be a subgraph of $G$ on $V_i$ by including any edge $(u,v)$ whose maintained counter is positive. 
			\item If $\chi(H_i) > q$, return $\chi(G)$ is `large' and terminate; 
		\end{enumerate}
		\item If none of the parallel trials have terminated, return $\chi(G)$ is `small'. 
	\end{enumerate}
\end{Algorithm}

\begin{lemma}\label{lem:dynamic-correct}
	\Cref{alg:dynamic} distinguishes between graphs with $\chi(G) \leq q$ and $\chi(G) \geq q \cdot t$ and uses $O((n/t)^2 \cdot \log^5{n})$ space, with high probability.  
\end{lemma}
\begin{proof}
	Since \emph{final} multiplicities of edges in $G$ are promised to be $\poly(n)$ and we are maintaining a counter mod $n^{\log{n}}$ in the algorithm, each $H_i$ will be 
	the induced subgraph of $G$ on $V_i$. Thus, if $\chi(G) \leq q$, the algorithm deterministically outputs `small' as none of its subgraphs can have $\chi(H_i) > q$. 
	On the other hand, if $\chi(G) \geq q \cdot t$, then, by~\Cref{lem:vertex-sample}, for any $i \in [k]$, we have, 
	\[
		1/2 \geq \Pr\paren{\chi(H_i) < \frac{p}{2\ln{n}} \cdot \chi(G) -1} = \Pr\paren{\chi(H_i) < \frac{4\ln{n}}{2\ln{n} \cdot t} \cdot q \cdot t -1} \geq \Pr\paren{\chi(H_i) \leq q}. 
	\]
	Hence, the probability that none of the trials output `large' is at most $(1/2)^k = 1/n^2$, meaning that with high probability, in this case, the algorithm outputs `large'. This proves the correctness. 
	
	We now bound the space complexity of the algorithm. The algorithm requires storing each $V_i$ which, with high probability, can be stored in $O(n/t \cdot \log^2{n})$ bits (by doing the sampling implicitly first and then storing the name of sampled vertices explicitly which with high probability are at most $O(n/t \cdot \log{n})$ many). It also requires storing each $H_i$ which involves maintaining $\leq \card{V_i}^2$ counters with $O(\log^2{n})$ bits each, for a total of $O(n^2/t^2 \cdot \log^4{n})$ bits for each $H_i$. Given the algorithm stores $O(\log{n})$ choices of $V_i$ and $H_i$, 
	we obtain the final bound on the space as well. 
\end{proof}

\Cref{thm:dynamic-upper} follows immediately from~\Cref{lem:dynamic-correct} (note that the requirement of $t \geq 4\log{n}$ in~\Cref{alg:dynamic} is inconsequential when proving~\Cref{thm:dynamic-upper} since for smaller
values of $t$ we can just store the entire graph in $\Ot(n^2)$ space which in this case is still $\Ot(n^2/t^2)$ and solve the problem). 

\subsection{A Dynamic Stream Lower Bound via Simultaneous Model (\Cref{thm:dynamic-lower})}

The proof of~\Cref{thm:dynamic-lower}, similar to our lower bound in~\Cref{sec:adversarial}, is based on multi-player communication complexity. But, since we are now proving a 
lower bound for dynamic streams and not insertion-only ones, and owing to the characterization of~\cite{AiHLW16}, we can focus on a more ``lower bound friendly'' communication model: the \emph{simultaneous} communication model. 
In the following, we first define this model and mention the result of~\cite{AiHLW16} that relates it to dynamic streaming lower bounds, and then present a communication lower bound in this model that allows us to prove~\Cref{thm:dynamic-lower}. 

\subsubsection{The Simultaneous Communication Model}\label{sec:sim} 

The simultaneous communication model is defined as follows. There are $p$ players, $P_1, \ldots, P_p$, each receiving input $x_i \in \mathcal{X}_i$. There is also an additional referee who receives no input. 
The goal of players and the referee is to compute some known function  $f$ on the domain $ \mathcal{X}_1\times\ldots\times \mathcal{X}_p$. To do this, each player $P_i$, \emph{simultaneously} with others, 
sends a message $M_i(x_i)$ to the referee, and the referee outputs $f(x_1,\ldots,x_p)$. The players follow some protocol $\pi$, and they have access to \textit{public randomness} in addition to their own \textit{private randomness}.

 \begin{definition}
    For any protocol $\pi$ in the simultaneous model, the \textnormal{\textbf{communication cost}} of $\pi$, denoted by $CC_{sim}(\pi)$, is defined as the worst-case length of all messages sent to the referee, i.e., 
    \[
    CC_{sim}(\pi) = \max_{(x_1, \ldots, x_p)} \sum_{i=1}^p \card{M(x_i)}.
    \]
\end{definition}

The following result, due to~\cite{AiHLW16} and follow up on~\cite{LiNW14}, shows lower bounds on communication cost of protocols in the simultaneous model also imply
dynamic streaming lower bounds. We note that unlike previous connections in~\Cref{prop:stream-cc} and~\Cref{prop:robust-stream-cc}, the proof of the following result is highly non-trivial. 

\begin{proposition}[\!\!\cite{AiHLW16}]\label{prop:dynamic-streaming-to-cc}
    Suppose that a single randomized streaming algorithm $A$ solves a problem $\mathbb{P}$ on any dynamic stream $\sigma$ with probability at least $1-\delta$, 
    and that the space complexity of $A$, denoted by $s(A)$, depends only on the dimension and final values of the frequency vector
    of the stream $\sigma$ (so, in particular not the length of the stream nor how large the intermediate values of the frequency vector are). For any integer $p \geq 1$ , let $f$ be a function 
    on the domain $\mathcal{X}_1\times\ldots\times \mathcal{X}_p$ such that $f(x_1,\ldots,x_p) = \mathbb{P}(x_1 \circ \ldots \circ x_k)$. There exists a $p$-player protocol $\pi$ in the 
    simultaneous communication model that solves $f$ with error probability at most $2\delta$ and has communication cost 
    \[
    	CC_{sim}(\pi) = O(p \cdot s(A)). 
    \]
   \end{proposition}

\newcommand{\DD}{\mathcal{D}}

\newcommand{\nbase}{\ensuremath{n_{base}}}

\newcommand{\ustar}{u^*}
\newcommand{\vstar}{v^*}
\newcommand{\jstar}{j^*}

\newcommand{\rTheta}{\rv{\Theta}}
\newcommand{\rJ}{\rv{J}}

\newcommand{\rPi}{\rv{\Pi}}


\subsubsection{Our Lower Bound in the Simultaneous Model}\label{sec:sim-lower}

We prove the following result on the simultaneous communication complexity of graph coloring. 

\begin{lemma}\label{lem:sim}
	Fix any integer $4 \leq k \leq n/2$ and $p := \binom{k}{2}$. Any $p$-player simultaneous communication protocol $\pi$ that distinguishes $3$-colorable from $k$-colorable graphs with 
	probability of success at least $2/3$ has communication cost $CC_{sim}(\pi) = \Omega(n^2)$ bits. 
\end{lemma}

We prove this lemma in the rest of this subsection and then at the end, show how this lemma together with~\Cref{prop:dynamic-streaming-to-cc} implies~\Cref{thm:dynamic-lower}. 

To prove~\Cref{lem:sim}, we introduce a hard distribution over graph instances such that, in one case, the graph is $3$-colorable, and in the other case, it contains a $k$-clique. Roughly speaking, in this distribution, each player ``sees'' a random bipartite graph on $\sim n-k$ vertices; these random bipartite graphs are put together in the final graph in a way that 
for each player, one \emph{special} pair of vertices in their bipartite graph belongs to a fixed set of $k$ vertices, whereas remaining vertices are part of a ``global'' bipartite graph; the special pair of vertices in the input of 
players are also forming a $k$-clique. In one case, all edges corresponding to these special pairs are present in input of players and thus forming a $k$-clique in the graph, whereas in the other case, 
none of those edges are present, and hence the final graph is $3$-colorable (by coloring the bipartite global part with two colors and the extra remaining $k$ vertices that are not connected to each other with a third color). 
Since from each players perspective, the special pair looks identical to any other pair of vertices, they will not be able to reveal much information about whether an edge exist between them or not, without communicating a large message. See~\Cref{fig:sketching} for an illustration of this structure.

We now formalize this, starting with the hard input distribution $\DD$. For any string $x_i $ for $i \in [p]$, we use $x_{i,j}$ to denote the bit at position $j$ in $x_i$. 

\begin{Distribution}[A hard input distribution $\DD$ for $p=\binom{k}{2}$ players for a given $4 \leq k \leq n/2$]\label{dist:sim}

~

\smallskip

\noindent
\textbf{Parameters.} 
For a sufficiently large integer $\nbase$ governing the number of vertices in the graph, define the parameters: 
\[
	t := ({\nbase})^2 \qquad and \qquad n := k + 2 \cdot (\nbase-1). 
\]
Create an $n$-vertex graph $G=(V,E)$ distributed between $p$ players as follows. 
\begin{enumerate}
    \item Pick $\jstar \in [t]$ and $\theta \in \{0, 1\}$ independently and uniformly at random.
    \item Define a matrix $x \in \{0,1\}^{p \times t}$ as $x_{i,\jstar} = \theta$ for all $i \in [p]$, and $x_{i, j} $ is chosen uniformly at random from $\{0,1\}$ and independently for all $i \in [p]$ and $j \neq \jstar$ ,$j \in [t]$.  
    \item For every $i \in [p]$, independently, 
    \begin{enumerate}
        \item Define a bipartite graph $G_i$ with each bipartition labeled by $[\nbase]$ (i.e., edges are pairs in $[\nbase] \times [\nbase]$).  
        \item 
        The row $i$ of matrix $x$, i.e. vector $x_i \in \{0,1\}^{t}$ represents the incidence vector of edges in $G_i$, i.e., the $j$-th edge in $G_i$ exists iff $x_{i,j} = 1$.\footnote{Throughout, we use the lexicographically-first ordering
        of $[\nbase] \times [\nbase]$ to map $[t]$ to edges of $G_i$.} 
    \end{enumerate}
            
     \item We use $(\ustar,\vstar)$ to denote the vertices of $j^*$-th pair in $[\nbase] \times [\nbase]$. Let $u_1,\ldots,u_{\nbase-1}$ be the vertices in $[\nbase] \setminus \ustar$ and 
        $v_1,\ldots,v_{\nbase-1}$ be the vertices in $[\nbase] \setminus \vstar$. 
        
    \item Pick a uniformly random permutation $\sigma:\set{1, \ldots, n} \rightarrow \set{1, \ldots, n}$. 
    \item Define the input to player $P_i$ for $i \in [p]$, 
    by starting with graph $G_i$ and then \emph{relabeling} its vertices as follows (to map them to $V = [n]$):
    \begin{enumerate}
    	\item For every $u_j$ for $j \in [\nbase-1]$, let ID of $u_j$ of $G_i$ be $\sigma(j)$ in $G$.
	\item For every $v_j$ for $j \in [\nbase-1]$, let ID of $v_j$ of $G_i$ be $\sigma(\nbase-1+j)$ in $G$. 
	\item For $\ustar,\vstar$ of $G_i$, let ID of $(\ustar,\vstar)$ correspond to the $i$-th pair of vertices in 
	\[
	\{\sigma(2(\nbase-1)+1)~,~\ldots~,~\sigma(2(\nbase-1)+k)\} \footnote{Again, we use the lexicographic mapping to map pairs of distinct elements from this set to $[\binom{k}2]$.}.
	\] 
    \end{enumerate}
    \item The final graph $G$ is the union of all edges of the graphs $G_i$ introduced for player $i \in [p]$. 
\end{enumerate}
\end{Distribution}

\begin{figure}[h!]
    \centering
    \includegraphics[width=0.6\textwidth]{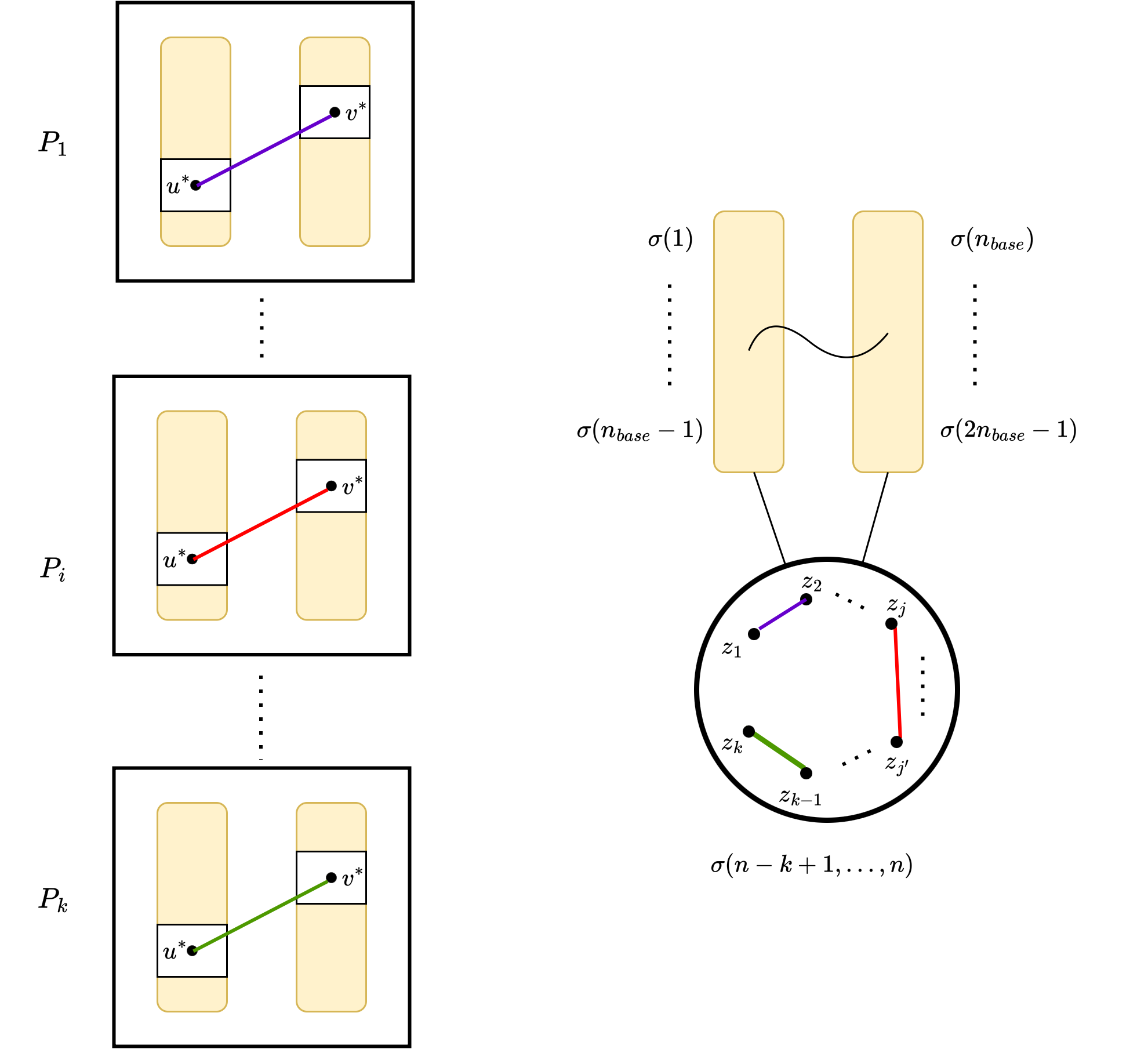}
    \caption{An illustration of the construction of the input graph of each player and the resulting graph combining their input. The shaded yellow region forms $	V_{bipartite} $, and the one edge that each player $P_i$ gets from $V_{clique}$ for $i \in [k]$ is shown.}
    \label{fig:sketching}
\end{figure}


We start by showing that the value of $\theta \in \set{0,1}$ governs the chromatic number of $G$. 

\begin{lemma}\label{lem:colorable-distribution}
	Any graph $G \sim \DD$ has $\chi(G) \leq 3$ when $\theta=0$ and $\chi(G) \geq k$ when $\theta=1$. 
\end{lemma}

\begin{proof}
	Fix $G \sim \DD$ and define the following partition of vertices in $G$ based on their labels: 
	\begin{align*}
		V_{bipartite} &:= \set{\sigma(1),\ldots,\sigma(2 \cdot (\nbase-1))} \\ V_{clique} &:= \set{\sigma(2(\nbase-1)+1),\ldots,\sigma(2(\nbase-1)+k)}. 
	\end{align*}
	We have that $G[V_{bipartite}]$ is a bipartite graph regardless of the choice of $\theta$ by the relabeling step of~\Cref{dist:sim}. Moreover, the pairs of vertices $(\ustar,\vstar)$ across all players all belong to $V_{clique}$ and 
	cover all pairs of vertices in $G[V_{clique}]$. 
	
	When $\theta=0$, there are no edges inside $G[V_{clique}]$ and so we can color $G$ by two coloring $V_{bipartite}$ and using a single new color for all of $V_{clique}$. 
	
	When $\theta=1$, $G[V_{clique}]$ becomes a $k$-clique and thus $\chi(G) \geq k$ in this case. 
\end{proof}

\Cref{lem:colorable-distribution} implies that any protocol $\pi$ that can distinguish between $3$-colorable vs $k$-colorable graphs can also determine the value of $\theta$ in graphs sampled
from $\DD$. We now use this to prove our lower bound. To continue, we need some notation.

\paragraph{Notation.}  Notice that the input to each player $i \in [p]$ is uniquely determined by the string $x_i$ and the relabeling of vertices of $G_i$ from $([\nbase],[\nbase])$ to $[n]$. 
We use $L_i$ to denote the relabeling function on vertices of $G_i$ for player $i \in [p]$ and note that $L_i$ is determined by $\sigma$ and $\jstar$. 

We use $\rTheta$, $\rJ$, $\rsigma$ to denote the random variables for $\theta,\jstar$ and $\sigma$, respectively. Similarly, for $i \in [p]$, we use $\rX_i$ and $\rL_i$ to denote the random variables for $x_i$ and $L_i$. 

For the rest of the proof, fix any $(1/3)$-error $p$-player simultaneous protocol $\pi$ and use $\rPi := (\rPi_1,\ldots,\rPi_p)$ denote the random variable for the messages sent by the players to the referee in $\pi$. 

We start with an easy claim that formalizes the connection between solving the coloring problem
and revealing non-trivial information about the value of $\theta$.

\begin{claim}\label{claim:sketching-omega(1)}
	Any $1/3$-error $p$-player simultaneous protocol $\pi$ for distinguishing between $3$- vs $k$-colorable graphs satisfies
	\[
	\mi{\rTheta}{\rPi \mid \rsigma,\rJ} = \Omega(1).
	\] 
\end{claim}
\begin{proof}
    By~\Cref{lem:colorable-distribution} we know that $\pi$ can be used to determine the value of $\theta$ with probability at least $2/3$. As the referee in the protocol $\pi$ only uses the message $\Pi$ to compute the final answer, 
    and by Fano's inequality (\Cref{fact:Fanos-inequality}), we have $\en{\rTheta\mid \rPi}\leq H_2(1/3)$. We thus have, 
    \begin{align*}
        H_2(1/3) &\geq \en{\rTheta\mid \rPi}\\
        &\geq \en{\rTheta\mid \rPi, \rsigma,\rJ}\tag{as conditioning can only decrease entropy by \itfacts{cond-reduce}}\\
        &=\en{\rTheta \mid \rsigma,\rJ} -\mi{\rTheta}{\rPi \mid \rsigma,\rJ} \tag{by the definition of mutual information} \\
        &= 1-\mi{\rTheta}{\rPi \mid  \rsigma,\rJ}.
    \end{align*}

    The last equality holds because $\theta$ is chosen uniformly at random from $\set{0, 1}$ independent of $\sigma,\jstar$ and thus the first entropy term is one (by~\itfacts{uniform}). 
    Since, $H_2(1/3) < 1$, we can conclude the proof. 
\end{proof}

We now have our main lemma that establishes a lower bound on the communication cost of $\pi$. 

\begin{lemma}\label{lem:sketching-omega(t)}
	Any protocol $\pi$ satisfying~\Cref{claim:sketching-omega(1)} has communication cost 
	\[
		CC_{sim}(\pi) = \Omega(t). 
	\]
\end{lemma}
\begin{proof}
    For all $j\in [t]$, let $\rY_j:=(\rX_{1,j}, \ldots, \rX_{p,j})$, where for each $i\in [p]$, the random variable $\rX_{i,j}$ denotes the $j$-th entry $x_{i,j}$ in the input row vector of player $P_i$. We are going to upper bound the LHS of~\Cref{claim:sketching-omega(1)} 
    in the following. By the definition of conditional mutual information, 
    \[
        \mi{\rTheta}{\rPi \mid \rsigma ,\rJ} = \Exp_{j\in [t]} \bracket{\mi{\rTheta}{\rPi \mid \rsigma , \rJ=j)}} = \frac{1}{t} \sum_{j=1}^t \mi{\rY_j}{\rPi \mid \rsigma , \rJ=j},
    \]
    where the last equality holds since $\rJ$ is uniform over $[t]$ and when $\rJ=j$, for all $i \in [p]$, each $x_{i,j}$ is set to be $\theta$. This implies $\rTheta$ is fixed by $\rY_{j}$ and vice-versa conditioned on $\rJ=j$. Continuing the equations, we have,
    \begin{align}
       \mi{\rTheta}{\rPi \mid \rsigma ,\rJ} &=\frac{1}{t} \sum_{j=1}^t \mi{\rY_j}{\rPi_1,\ldots,\rPi_p \mid \rsigma , \rJ=j} \tag{by the equation above and $\rPi = (\rPi_1,\ldots,\rPi_p)$} \\
        &=\frac{1}{t} \sum_{j=1}^t \sum_{i=1}^p \mi{\rY_j}{\rPi_i \mid \rPi_{1},\ldots,\rPi_{i-1}, \rsigma , \rJ=j} \tag{by the chain rule of mutual information (\itfacts{chain-rule})}\\
        &= \frac{1}{t} \sum_{j=1}^t \sum_{i=1}^p \mi{\rX_{i,j}}{\rPi_i \mid \rPi_{1},\ldots,\rPi_{i-1}, \rsigma , \rJ=j} \tag{as $\rX_{i,j}$ uniquely determines $\rY_j$ and vice versa conditioned on $\rJ=j$} \\ 
        &\leq \frac{1}{t} \sum_{j=1}^t \sum_{i=1}^p \mi{\rX_{i,j}}{\rPi_i \mid \rsigma , \rJ=j}, \label{eq:ref-later}
    \end{align}
    where we prove the inequality of~\Cref{eq:ref-later} in the following. We claim that 
    \[
    	\rPi_i \perp \rPi_{1},\ldots,\rPi_{i-1} \mid \rX_{i,j}, \rsigma,\rJ=j,
    \]
    which allows us to apply~\Cref{prop:info-decrease} to drop the conditioning on $\rPi_1,\ldots,\rPi_{i-1}$ and increase the mutual information term above, hence proving~\Cref{eq:ref-later}. The stated independence holds
    because conditioned on $\rX_{i,j},\rsigma,\rJ=j$, the randomness in the input of each player $i \in [p]$ is that of $x_{i,j'}$ for $j' \neq \jstar=j$, which is chosen independently across the players; hence, 
    the messages of the players that are functions of these randomness (conditioned on the above terms) are also independent. 
    
    Recall the relabeling function $L_i$ of each player $i \in [p]$, which, as stated earlier, is a function of $\sigma$ and $\jstar$. We define $L_{-i}$ as $\sigma$ minus $L_i$, namely, labels of all vertices that player $P_i$ has 
    no edges incident on them. Thus, $\sigma$ is uniquely defined by $(L_i,L_{-i})$ conditioned on $\jstar$. For every $i \in [p]$ and $\rJ=j$ in the RHS of~\Cref{eq:ref-later},
    \[
    	\mi{\rX_{i,j}}{\rPi_i \mid \rsigma , \rJ=j} = \mi{\rX_{i,j}}{\rPi_i \mid \rL_i,\rL_{-i}, \rJ=j} \leq \mi{\rX_{i,j}}{\rPi_i \mid \rL_i,\rJ=j}, 
    \]
    where the inequality holds by applying~\Cref{prop:info-decrease} in the RHS using 
    \[
    	\rPi_i \perp \rL_{-i} \mid \rX_{i,j} , \rL_i, \rJ=j;
    \]
    this independence itself holds exactly as before since conditioned on $( \rX_{i,j} , \rL_i, \rJ=j)$, the input of player $P_i$ and thus $\rPi_i$ is independent of all other players' inputs and in particular $\rL_{-i}$. 
   
    We have another crucial independence property at this point; we claim that 
    \[
    	(\rJ=j) \perp  \rX_{i,j},\rPi_i,\rL_i;
    \]
   intuitively, this is because any single player, given their entire input, cannot determine the value of $\rJ$ at all; more specifically, since $ \rX_{i,j},\rPi_i,\rL_i$ are all functions of the input of player $P_i$ whereas $\rJ=j$ is independent 
   of the input of the player, we have the above independence. But then this implies we can drop the conditioning on this event in the mutual information term above to get,  
   \[
   		\mi{ \rX_{i,j}}{\rPi_i \mid \rsigma , \rJ=j} \leq \mi{ \rX_{i,j}}{\rPi_i \mid \rL_i,\rJ=j} = \mi{ \rX_{i,j}}{\rPi_i \mid \rL_i}.
   \]    
    Using these in~\Cref{eq:ref-later}, we can conclude the proof as follows,  
    \begin{align*}
    	     \mi{\rTheta}{\rPi \mid \rsigma ,\rJ} &\leq \frac1t \sum_{j=1}^t \sum_{i=1}^{p} \mi{ \rX_{i,j}}{\rPi_i \mid \rL_i} \\
	     &\leq \frac1t \sum_{i=1}^{p} \sum_{j=1}^t \mi{ \rX_{i,j}}{\rPi_i \mid \rX_{i,1},\ldots,\rX_{i,j-1},\rL_i},
    \intertext{because $ \rX_{i,j}$ is independent of $\rX_{i,j-1}$ conditioned on $\rL_i$ since coordinates of $x_i \in \set{0,1}^t$ are chosen independently, and thus we can apply~\Cref{prop:info-increase}. Continuing, we have,}
    	    &= \frac1t \sum_{i=1}^{p} \mi{\rX_i}{\rPi_i \mid \rL_i} \tag{by chain rule of mutual information (\itfacts{chain-rule})} \\
	    &\leq \frac{1}{t} \sum_{i=1}^{p} \card{\rPi_i} \tag{since $\mi{\rX_i}{\rPi_i \mid \rL_i} \leq \en{\rPi_i \mid \rL_i} \leq \en{\rPi_i} \leq \card{\rPi_i}$ by \itfacts{uniform} and \itfacts{cond-reduce}} \\
	    &= \frac{1}{t} \cdot CC_{sim}(\pi). 
    \end{align*}
    The lower bound of~\Cref{claim:sketching-omega(1)} on the LHS of the above, implies the lemma. 
\end{proof}

The proof of~\Cref{lem:sim} follows immediately from~\Cref{lem:sketching-omega(t)} using the fact that $t = \Theta(n^2)$. 

Finally, we can also prove~\Cref{thm:dynamic-lower} easily using~\Cref{lem:sim} and~\Cref{prop:dynamic-streaming-to-cc}. 

\begin{proof}[Proof of~\Cref{thm:dynamic-lower}]
Given $q \geq 3$ and $t \in \IN$, define $k := q \cdot t$ and $p := \binom{k}{2}$. By~\Cref{lem:sim}, any $(1/3)$-error $p$-party protocol $\pi$ for distinguishing between $3$-colorable, and hence also $q$-colorable, vs $k$-colorable graphs 
has 
\[
	CC_{sim}(\pi) = \Omega(n^2). 
\]
Combined with~\Cref{prop:dynamic-streaming-to-cc}, this implies that any $(1/9)$-error dynamic streaming algorithm for distinguishing between $q$- vs $(k=) q \cdot t$-colorable graphs requires 
\[
	\Omega(\frac{n^2}{p}) = \Omega(\frac{n^2}{k^2}) = \Omega(\frac{n^2}{q^2 \cdot t^2}) 
\]
bits of space. \Cref{thm:dynamic-lower} now follows immediately from this for any $q = \poly\log{(n)}$. 
\end{proof}

	\clearpage


\bibliographystyle{alpha}
\bibliography{general}
	
	\clearpage
	
\appendix
	

\newcommand{\embed}{\ensuremath{\textnormal{\textsc{embed}}}}

\newcommand{\upc}{\ensuremath{\textnormal{\textsf{UPC}}}}

\newcommand{\startvert}{\ensuremath{\textnormal{\textsc{start}}}}
\newcommand{\finalvert}{\ensuremath{\textnormal{\textsc{final}}}}

\newcommand{\Gdup}{\ensuremath{G_{\textsc{dup}}}}
\newcommand{\cP}{\ensuremath{\mathcal{P}}}
\newcommand{\cH}{\ensuremath{\mathcal{H}}}
\section{More on Cluster Packing Graphs}\label{app:cpg} 

In this appendix, we provide the missing proofs of~\Cref{prop:generalized-simple-graph} and \Cref{prop:cpg-dense}. We also study further aspects of cluster packing graphs as a combinatorial object of its own independent interest in hope of shedding 
more light into them. 

\subsection{A Straightforward Construction for Larger Clusters (\Cref{prop:generalized-simple-graph})}\label{app:generalized-simple-graph}

We briefly sketch the proof of~\Cref{prop:generalized-simple-graph}, restated below. This is an easy generalization of~\Cref{lem:graph-simple} introduced and used in~\Cref{sec:two-player}. 

\begin{proposition*}[\Cref{prop:generalized-simple-graph}]
	For all integers $r,k,n \geq 1$ such that $r \cdot k \leq \sqrt{n}$, there exists a $(r,t,k)$-cluster packing graph $G$ with $t = \Omega(n^2/(r^2 \cdot k^3))$ clusters. 
\end{proposition*}

\begin{proof}[Proof Sketch]
	This construction closely follows the one in \Cref{lem:graph-simple}, except that we use groups of size $r$ instead of $k$ now. 
	
	We create a graph $G = (V, E)$ on $n$ vertices. The edge set $E$ will be specified later. 
The vertex set $V$ is partitioned into $k$ layers of $n/k$ vertices each, namely $\Li{1}, \ldots, \Li{k}$. Each layer $\Li{i}$ is further partitioned into $n/(k \cdot r)$ groups of size $r$ each, denoted by $\Li{j}_1, \ldots, \Li{i}_r$. Inside each group $\Li{i}_j$, the $r$ vertices are denoted by $\vij{i}{j}_1, \ldots, \vij{i}{j}_r$. 

For each
\[
	b \in [\frac{n}{2k \cdot r}], \qquad p \in [\frac{n}{2k^2 \cdot r}], \qquad s \in [r],
\]
we define a geometric line here as:
\[
	P_{b, p, s} = \paren{\vij{1}{b}_s, \vij{2}{b+p}_s, \ldots, \vij{k}{b+(k-1) \cdot p}_s}. 
\]
By our choice of parameters, 
\[
	b + (k-1) \cdot p \leq \frac{n}{2k \cdot r} + (k-1) \cdot \frac{n}{2k^2 \cdot r} \leq \frac{n}{k \cdot r},
\]
hence all the points on all the lines are valid points in $G$. 

We will now define the edges set $E$ of $G$. For each line $P_{b,p,s}$, we add all the edges between the $k$ many vertices in the line to form a $k$-clique. 
For $b \in [n/(2kr)], p \in [n/(2k^2\cdot r)]$, we define subgraph $H_{b,p}$ as a $k$-cluster of size $r$ by making all $k$ many vertices in each line $P_{b,p,s}$ as a $k$-clique for all $s \in [r]$. 
All the lines are vertex-disjoint, and thus subgraph $H_{b,p}$ is a $k$-cluster of size $r$. 

The total number of these subgraphs is $b \cdot p = \frac{n^2}{4k^3 \cdot r^2}$, as required. The inducedness property follows similarly as in the proof of \Cref{lem:graph-simple}.
\end{proof}

\subsection{A Dense Construction with  Almost Linear Size Clusters (\Cref{prop:cpg-dense})}\label{app:cpg-dense}

We now provide a proof of~\Cref{prop:cpg-dense}, restated below. We note that this proof is implicit already in~\cite{AssadiKNS24} and is only provided here for completeness. 

\begin{proposition*}[\Cref{prop:cpg-dense}]
	For all integers $r,k,n \geq 1$ such that $k = o(\sqrt{\log{n}})$, there exists a $(r,t,k)$-cluster packing graph $G$ with parameters 
	\[
		r = \frac{n}{2^{\Theta(\sqrt{\log{n} \cdot \log{k}})}} \quad and \quad t = \frac{n}{2^{\Theta(\sqrt{\log{n} \cdot \log{k}})}}. 
	\]
\end{proposition*}

We need some definitions from Section 4 of \cite{AssadiKNS24}. Note that these definitions are only pertinent to this subsection and not the rest of the paper.

	For any $n, k \geq 1$, a graph $G = (V, E)$ is called a \textbf{layered graph}  if there exists an equipartition of $V$ into $k$ \emph{independent sets} $(V_1, V_2, \ldots, V_k)$ of size $n/k$ each, referred to as \textbf{layers} of $G$. In addition if all the edges of $G$ are between consecutive layers, $G$ is called a \textbf{strictly-layered} graph.		
	For any $n, k \geq 1$, a path $P$ in an $(n,k)$-strictly-layered graph $G$ is called a \textbf{layered path} if it has edges between consecutive layers and exactly one vertex per layer.
	For any layered path $P$, we define $\startvert(P)$ as the vertex of $P$ in the layer $V_1$ of $G$ and $\finalvert(P)$ as the vertex of $P$ in the layer $V_k$.

A set of paths $\cP$ in a \emph{strictly}-layered graph $G$ is called a \textbf{unique path collection (UPC)} iff $(i)$ $\cP$ is a vertex-disjoint set of layered paths, 
	and, $(ii)$ for any pair of vertices $s \in \startvert(\cP) := \set{\startvert(P) \mid P \in \cP}$ and $t \in \finalvert(\cP) := \set{\finalvert(P) \mid P \in \cP}$, 
	the only layered path between $s$ and $t$ in the entire graph $G$ is a path in $\cP$ (if one exists, $s$ and $t$ are the ends of the same $P \in \cP$).

\begin{Definition}\label{def:unique-paths}
	For any $n,k, p,q \geq 1$, a $(p,q,k)$-\textbf{disjoint-unique-paths (DUP)} graph is a $(n,k)$-strictly-layered graph $G$ whose edges can be partitioned into $q$ 
	UPCs $\cP_1,\ldots,\cP_q$, each consisting of exactly $p$ layered paths $P_{i,1},\ldots,P_{i,p} \in \cP_i$ for $i \in [q]$. 
\end{Definition}

\cite{AssadiKNS24} proves the existence of DUP graphs. Although they do not state the parameters exactly the way we do, these parameters are implicit in the proof of Proposition 4.5 of \cite{AssadiKNS24}.

\begin{proposition}[c.f. Proposition 4.5 of \cite{AssadiKNS24}]\label{prop:dup}
	For any sufficiently large $n \in \IN$ and any given $k \geq 1$ satisfying $k \leq 2^{(\log{n})^{1/4}}$, there exists a $(p,q,k)$-DUP graph on $n$ vertices with the following parameters for two absolute constants $\eta_p,\eta_q > 0$,
	\[
	p = \frac{n}{\exp\paren{\eta_p \cdot (\ln{n})^{1/2} \cdot (\ln{k})^{1/2}}} \qquad \text{and} \qquad q = \frac{n}{\exp\paren{\eta_q \cdot (\ln{n})^{1/2} \cdot (\ln{k})^{1/2}}}.
	\]
\end{proposition}

We also need the notion of embedding products and properties about them from \cite{AssadiKNS24}.
\begin{Definition}\label{def:embed}
	Let $\Gdup$ be a $(p,q,k)$-DUP graph on $n_1$ vertices and layers $U_1,\ldots,U_{k}$ for some $p,q,k \geq 1$. 
	Let $\cH := \set{H_{i,j} \mid i \in [q] , j \in [p]}$ be a family of $(n_2,k)$-layered graphs where all the graphs in $\cH$ are on the same layers $W_1,\ldots,W_{k}$ but may have different edges. 
	
	We define the \textbf{embedding of $\cH$ into} $\Gdup$, denoted by $G:=\embed(\cH \rightarrow \Gdup)$, as the following $((n_1 \cdot n_2)/k , k)$-layered graph: 
	\begin{itemize}[leftmargin=10pt]
		\item \textbf{Vertices:} We have layers $V_1,\ldots,V_{k}$ where, for every $\ell \in [k]$, $V_\ell := U_\ell \times W_\ell$; 
		\item \textbf{Edges:} For any layered path $P_{i,j} = (u_1, u_2, ..., u_{k})$ of UPC $\cP_i$ in $\Gdup$ with $i \in [q]$ and $j \in [p]$, and any edge $(x,y) \in H_{i,j}$ between layers $W_{\ell_x}$ and $W_{\ell_y}$, 
		we add an edge $e$ between $(u_{\ell_x},x)$ and $(u_{\ell_y},y)$ to $G$. We say that the edge $e$ is added \emph{w.r.t.}\ the path $P_{i,j}$. 
	\end{itemize}
\end{Definition}

\begin{lemma}[c.f. Lemma 4.7 of \cite{AssadiKNS24}]\label{lem:embed-induced}
	Let $\Gdup$ be a $(p,q,k)$-DUP graph for some $p,q,k \geq 1$, $\cH := \set{H_{i,j} \mid i \in [q] , j \in [p]}$ be a family of layered graphs, and $G := \embed(\cH \rightarrow \Gdup)$. 
	
	For any $i \in [q]$, the \textbf{\emph{induced subgraph}} of $G$ on vertices corresponding to $\cP_i$, i.e., $$\set{(v,*) \mid v \in P_{i,j}~\text{for some $P_{i,j} \in \cP_i$}},$$ is a \emph{\textbf{vertex-disjoint union}} of graphs $H_{i,j}$ for $j \in [p]$ in $\cH$. 
\end{lemma}

This gives us everything we need to prove \Cref{prop:cpg-dense}.
\begin{proof}[Proof of \Cref{prop:cpg-dense}]
	We start with a DUP-graph $\Gdup$ with the parameters given in \Cref{prop:dup}. 
	We use a $p \cdot q$ family $\cH$ of $(k, k)$-layered graphs $H_{i,j}$ for $i \in[q], j \in [p]$, all of which are $k$-cliques. 
	
	We will show that $G = \embed(\cH \rightarrow \Gdup)$ is a $(p, q, k)$-cluster packing graph. Let $H_i$ be the subgraph induced by all the edges added to $G$ via the UPC $\cP_i = (P_{i,1}, \ldots, P_{i,p})$, for each $i \in [q]$. 
	
	Each edge added to $G$ is added with respect to some path $P_{i,j}$, and thus is a part of subgraph $H_i$. The edges in these subgraphs $H_i$ for $i \in [q]$, therefore, form a partition of the edge set of $G$. 
	
	Each subgraph $H_i$ consists of vertex-disjoint union of $H_{i,1}, \ldots, H_{i,p}$ by \Cref{lem:embed-induced}, and each $H_{i,j}$ for $j \in [p]$ is a $k$-clique by construction. Hence, each subgraph is a $k$-cluster of size $p$.  \Cref{lem:embed-induced} also proves that the subgraph $H_i$ does not contain edges from any other subgraph $H_{i'}$, as it states that the subgraph induced by the edges of $H_i$ contains no other edges apart from those in $H_i$. 
	Thus, edges of graph $G$ can be partitioned into $q$-many induced $k$-clusters of size $p$. 
	
	The parameters $p = r$ and $q = t$ are the same as in the statement of the proposition, which completes the proof. 
\end{proof}

\subsection{$k$-Colorability of $(r,t,k)$-Cluster Packing Graphs}\label{app:cpg-k-color}

In our definition of $(r,t,k)$-cluster packing graphs (\Cref{def:cpg}), we explicitly required them to be $k$-colorable (primarily to avoid repeating this extra condition everywhere in the paper). We now show 
that this addition to the definition is essentially without loss.

\begin{lemma}\label{lem:coloring-k-packing}
For any $n, r, t, k \geq 1$, given a graph $G = (V, E)$ on $n$-vertices which is an $(r,t,k)$-cluster packing graph that is not necessarily $k$-colorable, there exists a graph $G' = (V', E')$ on $n'$-vertices that is an $(r', t', k)$-cluster packing graph and \emph{is} $k$-colorable with parameters:
\[
	n' = n \cdot k, \qquad r' = r \cdot k \qquad t' = t.
\]
\end{lemma}

\Cref{lem:coloring-k-packing} shows that parameters achievable for not-necessarily-$k$-colorable cluster packing graphs are almost the same as the one for $k$-colorable ones (when focusing on ``small'' values of $k$).

Let $E_1, E_2, \ldots E_t$ denote the partition of the edges in $E$ into $t$ sets, each of which corresponds to an induced $k$-cluster of size $r$. We use $\vij{i}{j, \ell}$ to denote the $\ell^\textnormal{th}$ vertex in the $j^{\textnormal{th}}$ $k$-clique of edge partition $E_i$, for $i \in [t], j \in [r]$ and $\ell \in [k]$.

We define $k$ different permutations of the set $[k]$ as follows. For $i \in [k]$ and any $x \in [k]$,
\[
\tau_i (x) = ((x + (i-2)) \mod k )+ 1.
\]
Notice that these permutations have the following property: for any $x \in [k]$, $\tau_i(x) \neq \tau_j(x)$ for any $i , j \in [k]$ with $i \neq j$. Refer to \Cref{fig:tau-perm} for an illustration of these permutations and the property, and to \Cref{fig:k-color-construct} for an illustration of the larger graph we create in this subsection. 
\newcommand{\cliq}[1]{\ensuremath{C^{(#1)}}}
\newcommand{\ellstar}{\ensuremath{\ell^{\star}}}
\newcommand{\jstarb}{\ensuremath{j^{\star}}}

\begin{figure}[h!]
	\centering
	\includegraphics[scale=0.25]{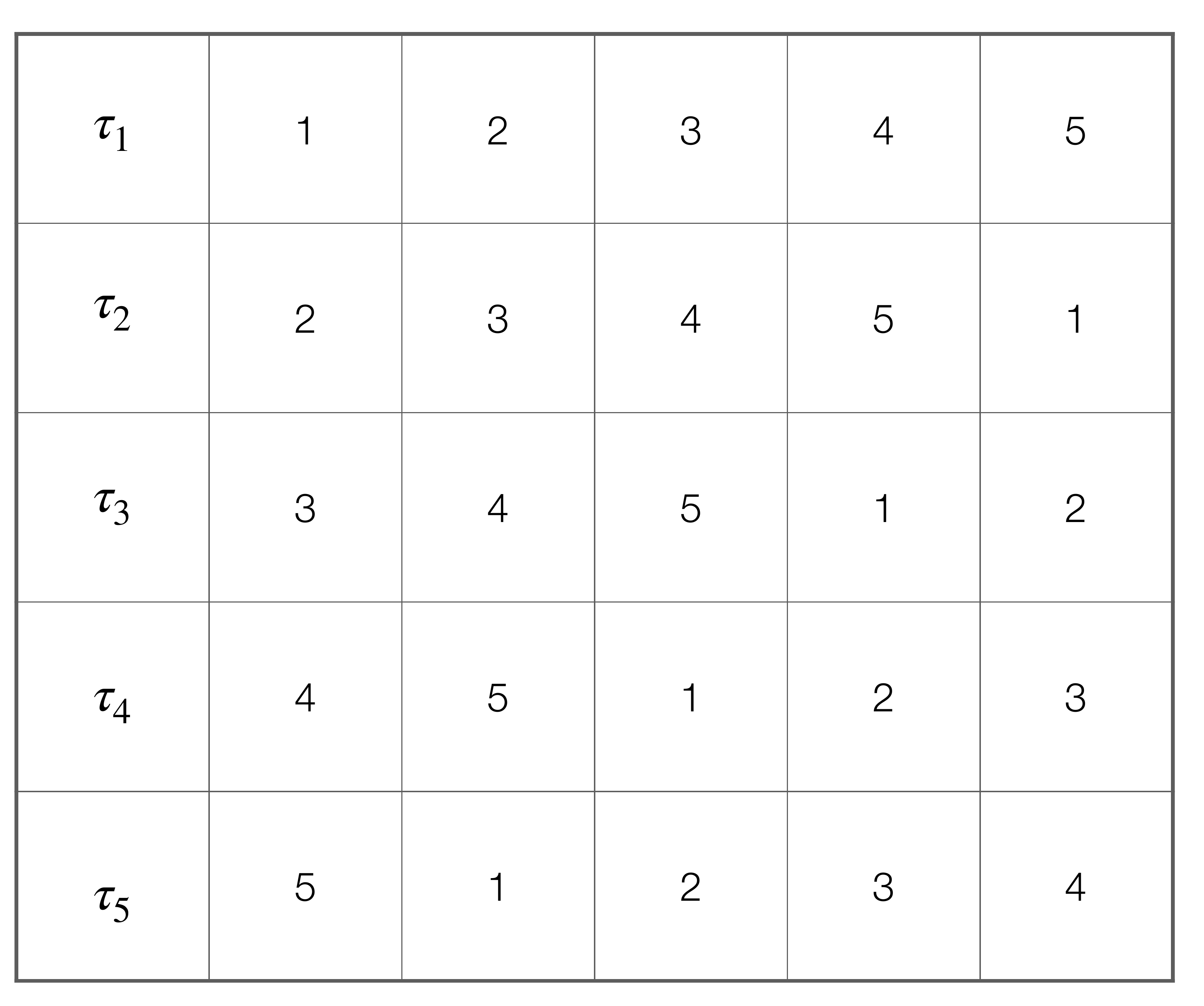}
	\caption{A table containing the permutations $\tau_1, \ldots, \tau_k$ for $k = 5$.}
	\label{fig:tau-perm}
\end{figure}

\begin{proof}[Proof of \Cref{lem:coloring-k-packing}]
The vertex set $V'$ of $G'$ is:
\[
V' = V'_1 \sqcup V'_2 \sqcup \ldots \sqcup V'_k, 
\]
where each $V_i'$ is a copy of $V$. 
We use $\vij{i}{j, \ell}_a$ for $a \in [k]$ to denote the copy of vertex $\vij{i}{j, \ell} \in V$ in set $V'_a$. 

The edge set of $G'$ will be composed of $(t \cdot r \cdot k)$-many $k$-cliques. 
For $i \in [t], j\in [r]$ and $\ell \in [k]$, we define the clique $\cliq{i, j, \ell}$ as the following set of edges. 
\[
	\cliq{i,j,\ell} = \{(\vij{i}{j, \tau_{\ell}(a)}_{a}, \vij{i}{j, \tau_{\ell}(b)}_b) \mid a, b \in [k] \textnormal{ with } a \neq b \}.
\]
This forms a $k$-clique between the following vertices:
\[
	\{ \vij{i}{j, \tau_{\ell}(1)}_{1}, \vij{i}{j, \tau_{\ell}(2)}_2, \ldots, \vij{i}{j, \tau_{\ell}(k)}_k\}.
\]
The edge set $E'$ will be the union of all $\cliq{i,j, \ell}$ for $i \in [t], j \in[r]$ and $\ell \in [k]$.
It is easy to see that $G'$ is $k$-colorable. Each copy $V'_i$ can be colored with one color, as there are no edges between any two vertices in the same copy. 

Observe that for any $i \in[t]$ the vertices of $\cliq{i, j, \ell}$ are disjoint from each other for all $j \in [r]$ and $\ell \in [k]$. For two vertices $\vij{i}{j, \tau_{\ell}(a)}_{a}$ and $\vij{i}{j, \tau_{\ell'}(a)}_a$ in the same layer $V'_{a}$ with $\ell \neq \ell'$, the values of $\tau_{\ell}(a)$ and $\tau_{\ell'}(a)$ are distinct. 

We define the partition of $E'$ into $t$ sets as follows. For $i \in [t]$, 
\[
	E'_i = \bigcup_{j \in [r], \ell \in [k]} \cliq{i, j, \ell}.
\]
As we argued that cliques $\cliq{i,j, \ell}$ are disjoint for $j \in [r], \ell \in [k]$, we can see that $E'_i$ is a $k$-cluster of size $r \cdot k$. It remains to prove the inducedness condition required in cluster-packing graphs. 

Consider any two vertices $\vij{i}{j, \tau_{\ell}(a)}_a$ and $\vij{i}{j', \tau_{\ell'}(b)}_b$ inside the subgraph induced by $E'_i$. Assume towards a contradiction that an edge exists between these two vertices, and it comes from some partition $E'_{\istar}$ for $\istar \in [t]$, with $\istar \neq i$. 
Then, these vertices must be of the form:
\[
	\vij{i}{j, \tau_{\ell}(a)}_a = \vij{\istar}{\jstarb, \tau_{\ellstar}(a)}_a \qquad  \vij{i}{j', \tau_{\ell'}(b)}_b = \vij{\istar}{\jstarb, \tau_{\ellstar}(b)}_b,
\]
for some $\jstarb \in [r], \ellstar \in [k]$. 
By definition of $E'$, if an edge exists between these two vertices, an edge also exists between vertices $\vij{\istar}{\jstarb, \tau_{\ellstar}(a)}$ and $\vij{\istar}{\jstarb, \tau_{\ellstar}(b)}$ in the original graph $G$, from the partition $E_{\istar}$. However, these vertices are present in the subgraph induced by partition $E_i$, and this violates the inducedness property of $E_i$, giving a contradiction. 
\end{proof}

\begin{figure}[h!]
 \centering
	\includegraphics[scale=0.35]{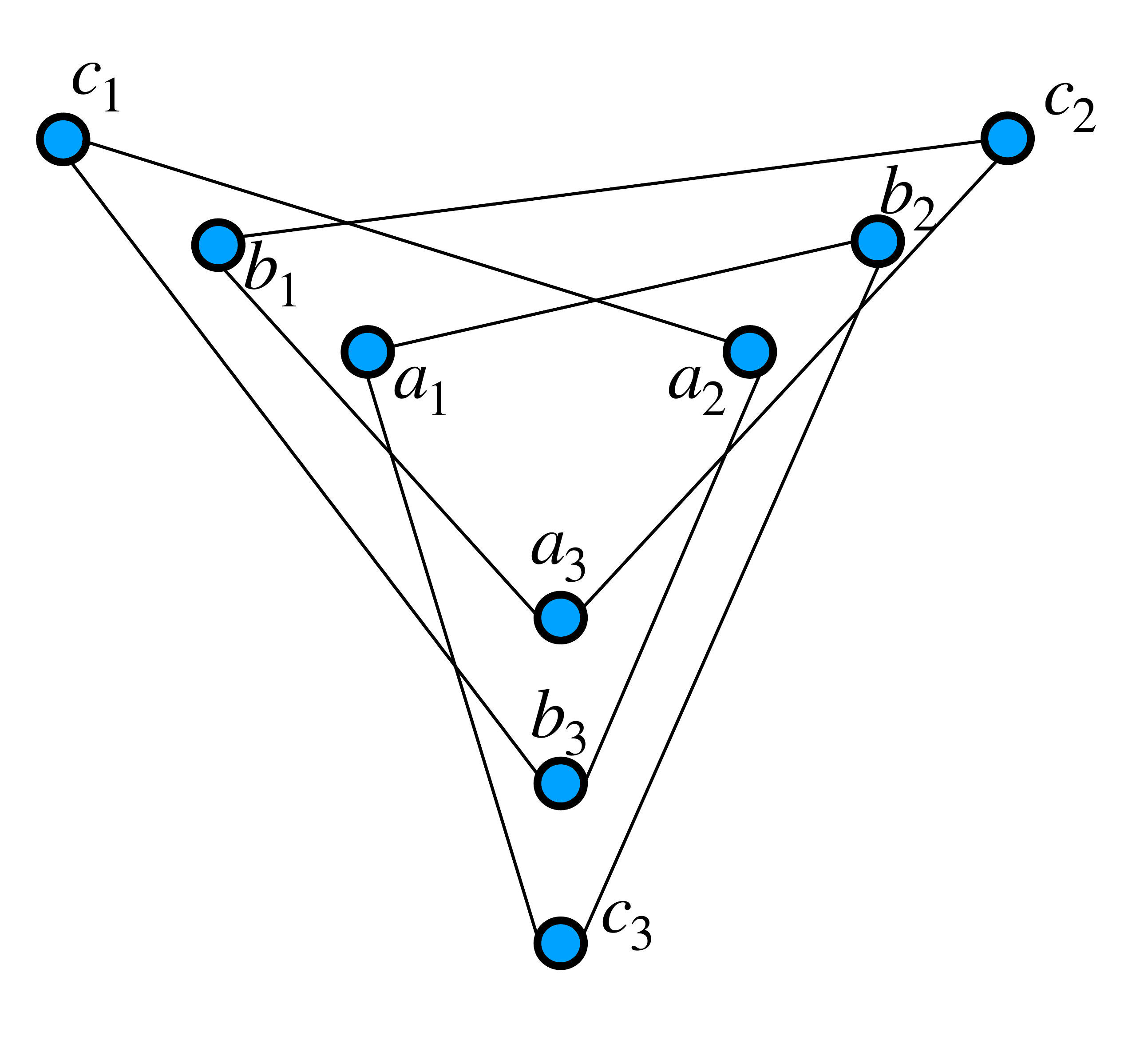}
\caption{An illustration of one partition of edges in graph $G$ and $G'$ when $r = 1$ and $k =3$. Three triangles created in graph $G'$ because of one triangle $\set{a,b,c}$ in $G$.}
\label{fig:k-color-construct}
\end{figure}

\newcommand{\cF}{\mathcal{F}}

\subsection{Threshold of $r=n/k^2$ in Cluster Packing Graphs}\label{app:cpg-threshold} 

Our \Cref{prop:cpg-large-r} shows that we can have almost dense (or at least not too sparse) cluster packing graphs with cluster size of $r = n/2k^2 \cdot (1-o(1))$. Using a slightly different construction, we can 
achieve a similar result even when $r = n/k^2 \cdot (1-o(1))$, namely, a factor $2$ improvement. This is by borrowing an idea
from~\cite{GoelKK12} that achieved the same improvement for RS graphs, by increasing induced matching size in RS graphs from $\sim n/6$ in~\cite{FischerLNRRS02} to $\sim n/4$ in~\cite{GoelKK12} while keeping the density of the graph 
almost the same. We briefly sketch the idea below but omit the details. 

\paragraph{Improving~\Cref{prop:cpg-large-r} for $r = n/k^2 \cdot (1-o(1))$.}
In \Cref{prop:cpg-large-r}, the value of $r$ we obtain is $n/2k^2 \cdot (1-o(1))$, from \Cref{claim:lower-bound-r}. For any set $S$ from the family $\cF$, this corresponds to all the points in $\Li{1}$ with color $c_1$ that satisfy $x_i + 2k \leq p $ for all $i \in S$. 
There were at most $p$ groups in layer $\Li{1}$, and the best lower bound we could get is $n/2k^2$. 

To improve the value of $r$, we make the size of groups which are assigned the color `white' less than the sizes of all the other groups. 
More formally, if we want $r $ to be $n/k^2 \cdot (1-\epsilon - o(1))$ for some constant $\epsilon \in (0,1)$, we pick the sets $S$ to have size $\epsilon \cdot d/6$. The grouping of all the layers $\Li{i}$ will now be asymmetric between even groups (of size $w = \card{S} = \epsilon \cdot d/3$) and odd groups (of size $w' = d/3$). 
For $j \in [p/(2+\epsilon)]$,
 \begin{align*}
\LijS{i}{2j-1}{S} &:= \set{x \in \Li{i} \mid  (j-1)\cdot (w + w') \leq w_S(x) < (j-1)\cdot (w+w') + w'}, \\
\LijS{i}{2j}{S}&:= \set{x \in \Li{i} \mid (j-1) \cdot (w+w') + w' \leq w_S(x) < j \cdot (w + w')} 
\end{align*}
The colors associated with $S$ are assigned similarly, starting from group $\LijS{1}{1}{S}$, and cycling through $c_1, c_2, \ldots, c_k$ for odd groups, and `white' for even groups. 
The lines and edges are also defined similary. 
As the size of `white' groups is $w = \card{S}$, each line will still have colours $c_1, c_2, \ldots, c_k$ as in \Cref{obs:colorful-line}.

Let us argue that the number of clusters is large. 
We want to lower bound the total number of points $x \in \Li{1}$ such that $x_i + 2k \leq p$ for all $i \in S$ and $x_i$ belongs to a group with color $c_1$. 
The total number of odd groups is at least $p^d/(w+w')$. Among these at least $k$ have color $c_1$ as the colors are assigned cyclically. 
The total number of points in $\Li{1}$ with color $c_1 $ is at least, 
\[
	\frac{p^d}{w + w'} \cdot \frac1{k} \cdot w' = \frac{n}{k^2} \cdot (\frac{w'}{w + w'}) = \frac{n}{k^2} \cdot (1- \frac{\epsilon}{1+\epsilon}) \geq \frac{n}{k^2} \cdot (1-o(1)).
\]

The total number of points in $\Li{1}$ with at least some $i \in S$ with $x_i + 2k > p$ will still be $o(p^d/k) = o(n/k^2)$ fraction by the same arguments. Hence, the sizes of the $k$-clusters is at least $(n/k^2) \cdot (1-\epsilon - o(1))$.

Finally, to ensure the inducedness property, we still want a large set family with low intersection. As the sizes of the sets are parametrized by $\epsilon$, we can find a family of $2^{\Omega_{\epsilon}(d)}$ sets, each of size $\epsilon \cdot (d/3)$, with pairwise intersection at most $(\epsilon d/3)^2 \cdot c$ for some appropriate constant $c > 1$ (but still small enough so that inducedness property holds) using similar arguments as in \Cref{claim:num-of-sets}. We will get an $(r,t,k)$-cluster packing graph with parameters:
\[
	r = \frac{n}{k^2} \cdot(1-\epsilon - o(1)) \qquad t = n^{\Omega_{\epsilon}(1/\log \log n)},
\]
for any constant $ \epsilon \in (0,1)$.


\paragraph{Optimality of $r=n/k^2$.} We can also ask if it is possible to increase $r$ even further, say, have $r = \Theta(n/k)$ for instance. In the following, we use a simple argument 
to show we cannot hope to increase $n/k^2$ by any constant factor and still achieve a non-sparse cluster packing graphs. This is a reminiscent of the threshold of $n/4$ established by~\cite{FoxHS17} for RS graphs. 

\begin{proposition}\label{lem:tight-structure}
	For any $\delta \in (0,1)$, any $n$-vertex $(r, k, t)$-cluster packing graph with parameter 
	\[
	r = \frac{n}{k^2} \cdot (1+\delta)
	\]
	 can have at most $t=O(\frac{k}{\delta})$ many induced $k$-clusters.
\end{proposition}

To prove~\Cref{lem:tight-structure}, we need the following lemma, which provides a version of Corradi's lemma~\cite[Lemma 2.1]{Jukna:2001} tailored to our setting. 
This lemma is (somewhat) standard and we postpone its proof to~\Cref{app:missing} (we are not aware of this exact formulation, hence providing a simple proof of it for completeness). 

\begin{lemma}\label{lem:set-limited-size}
Let $n$ be a positive integer and $\varepsilon \in (0, 1/2)$. Let $\mathcal{F}$ be a family of subsets of $[n]$ such that for every $S \in \mathcal{F}$, we have $\card{S} = \varepsilon \cdot n$, and for all distinct $S, T \in \mathcal{F}$, we have 
$\left|S \cap T \right| \leq \varepsilon^2 \cdot (1-\delta)\cdot n$ for some $\delta\in (0, 1)$.
Then, $\left|\mathcal{F}\right| = O(\frac{1}{\delta\cdot \varepsilon})$.
\end{lemma}

To continue, we also need the following basic observation about cluster packing graphs. 

\begin{claim}\label{claim:kCluster-intersection-size}
	In any $(r, t, k)$-cluster packing graph with $k$-clusters $H_1, \ldots, H_t$, for all $i\ne j\in [t]$, 
	\[
	\left|H_i\cap H_j\right| \leq r.
	\]
\end{claim}
\begin{proof}
	Suppose for contradiction that for some $i\ne j\in [t]$, we have
	$
	\left|H_i\cap H_j\right| \geq r+1.
	$
	Let $C_1, \ldots, C_r$ be the $k$-cliques of $H_i$, and let $H_i\cap H_j =\{v_1, \ldots, v_m\}$ with $m\geq r+1$. Since $\{v_1, \ldots, v_m\} \subseteq H_i = \bigcup_{l\in [r]}C_l$, by the pigeonhole principle, there exists some $i^*\in [r]$ such that
	 \[
	 \left|C_{i^*}\cap (\{v_1, \ldots, v_m\})\right| \geq 2.
	 \]
	Let $u, v$ be two of vertices in $C_{i^*}$ that also belong to $H_j$. Since $H_i$ and $H_j$ are edge-disjoint and $(u, v)\in E(H_i)$, it must be that $(u, v) \notin E(H_j)$, so $u$ and $v$ cannot be in the same clique in $H_j$. Thus, there exists cliques $C$ and $C'$ in $H_j$ such that $u\in C$ and $v\in C'$. However, $H_j$ is an induced subgraph, and there is no edge between $C$ and $C'$ which is a contradiction as $(u, v)$ is an edge in the ambient graph. 
	Therefore, our assumption is false, and we conclude that $\left|H_i \cap H_j\right| \leq r$.
\end{proof}

We can now conclude the proof of~\Cref{lem:tight-structure} easily. 

\begin{proof}[Proof of~\Cref{lem:tight-structure}]
	We can view the $k$-clusters $H_1, \ldots, H_t$ as subsets of $[n]$. We know that 
	$\left|H_i\right|=r\cdot k$ for all $i\in [t]$, and by~\Cref{claim:kCluster-intersection-size}, we know that $\left|H_i\cap H_j\right|\leq r$ for all $i\ne j\in [t]$.  
    Let $\varepsilon := \frac{r\cdot k}{n}$. Assuming $r = \frac{n}{k^2} \cdot (1+\delta)$, we see that
    \begin{align*}
        \left|H_i\cap H_j\right|
        &\leq r = r \cdot \underbrace{\frac{1}{1+\delta} \cdot (1+\delta) \cdot \frac{n}{k^2} \cdot \frac{k^2}{n^2} \cdot n}_{=1}\\
        & \leq \frac{(r\cdot k)^2}{n} \cdot \frac1{1+\delta} = \varepsilon^2\cdot \frac{1}{1+\delta}\cdot n  \tag{by the choice of $r$ and $\varepsilon$} \\
        &\leq (1-\delta/2) \cdot \varepsilon^2 \cdot n. \tag{as $\delta \in (0,1)$} 
    \end{align*}
    
    Therefore, $r \leq \varepsilon^2 \cdot (1-\delta/2)\cdot n$ for any $\delta\in (0, 1)$. Thus, by~\Cref{lem:set-limited-size}, we have 
    \[
    t = O\paren{\frac{1}{\delta\cdot \varepsilon}} = O\paren{ \frac{1}{\delta} \cdot \frac{n}{r \cdot k}} = O\paren{\frac{1}{\delta} \cdot \frac{n \cdot k^2}{n \cdot (1+\delta) \cdot k}} = O\paren{\frac{k}{\delta}},
    \]
    concluding the proof. 
\end{proof}

\subsection{Comparison to RS-graphs}\label{app:cpg-parameters}

 RS graphs have found numerous applications in various models of computation that are closely related to the topics of this paper, such as:  
 \begin{itemize}
 	\item Property Testing \cite{Alon01,FischerLNRRS02}, ... 
 	\item Streaming Algorithms and Lower Bounds \cite{GoelKK12,Kapralov13,AssadiKLY16,AssadiR20,AssadiBKL23,AssadiS23}, ...
 	\item Dynamic Graph Algorithms \cite{AssadiBKL23,BehnezhadG24,AssadiKK25}, ... .
\end{itemize}
This list of citations is by no means comprehensive, and more references can be found within these works. Hence, we provide some more context for cluster-packing graphs and connections to RS-graphs, as it is reasonable to believe that they may find more applications. 

\subsubsection*{Graph Constructions}
There are two main regimes of parameters for RS-graphs that are of interest to us:
\begin{enumerate}[label=$(\roman*)$]
	\item Sublinear matching size: here $r = o(n)$, and we want the total density $t \cdot r$ to be large. \cite{AlonMS12} showed that graphs with $\binom{n}{2} - o(n^2)$ edges with $r = n^{1-o(1)}$ size matchings exist.  
	Traditionally, and the origin of these graphs, the question of interest is to set $t=\Theta(n)$ and attempt to increase $r$ as much as possible~\cite{RuzsaS78}. 
	\item Linear matching size: here $r = \Omega(n)$, and we want to maximize the number of matchings $t$. The current constructions stand at $t = n^{\Omega(1/\log \log n)}$ for any $r < n/4$. \cite{FischerLNRRS02,GoelKK12}.
\end{enumerate}

\noindent
Similarly, density of $(r,t,k)$-cluster packing graphs can be studied in two regimes. 
\begin{enumerate}[label=$(\roman*)$]
	\item Suboptimal $k$-cluster size: here $r = o(n/k^2)$, and we aim to maximize the density $r \cdot t$ of the clusters. 
	This regime can be viewed as an analogue of the \emph{sublinear matching} size regime of RS-graphs.
	The construction of \cite{AssadiKNS24} (presented in \Cref{app:cpg-dense} for completeness) falls in this regime, and it is based on the RS-graph construcion of \cite{RuzsaS78} for the sublinear matching size regime. 
	\item Optimal $k$-cluster size: here $r = \Omega(n/k^2)$, and we want to maximize the number of $k$-clusters of this size. This regime can be seen as an analogue of the \emph{linear matching} size regime of RS-graphs. The construction presented in \Cref{sec:sparse-cluster-packing} of this work falls into this regime, and it is in turn based on the construction of \cite{FischerLNRRS02} for RS-graphs with linear matching size. 
\end{enumerate}

 \Cref{app:cpg-threshold} shows that $r$ cannot be $\omega(n/k^2)$ while still having many $k$-clusters, and hence these are the only two parameter regimes for which dense $(r,t,k)$-cluster packing graphs can exist. 

\subsubsection*{Upper Bounds on Densities}

Bounding the densities of RS-graphs has been a notoriously challenging open problem. State-of-art upper bounds are based on the graph removal lemma by \cite{Fox12}. They prove that the total density of the graphs can be at most $n^2/2^{\Omega(\logstar{n})}$
as long as $r$ or $t$ are just moderately large. Even in the linear matching size regime, we do not know stronger lower bounds (only for values of $r$ very close to $n/4$, we have 
a better upper bound of $t = O(n/\log{n})$ established in~\cite{FoxHS17}; note that all these bounds are quite far from existing constructions). 

These upper bounds continue to hold for $(r,t,k)$-cluster packing graphs, even for $k = 2$. 
More specifically, in the linear size regime for RS-graphs and for $(r,t,k)$-cluster packing graphs with optimal $k$-cluster size, these \emph{same bounds} hold on the value of $t$:
\begin{itemize}
	\item For RS-graphs:
	\begin{align*}
		n^{\Omega(1/\log \log n)} \underset{\textnormal{\cite{FischerLNRRS02}}}{\leq} t \underset{\textnormal{\cite{FoxHS17}}}{\leq} \frac{n}{2^{\Omega(\logstar{n})}},
	\end{align*}
	\item For cluster packing graphs:
	\begin{align*}
		n^{\Omega(1/\log \log n)} \underset{\textnormal{this work}}{\leq} t \underset{\textnormal{\cite{FoxHS17}}}{\leq} \frac{n}{2^{\Omega(\logstar{n})}}.
	\end{align*}
\end{itemize}

It is interesting to see if we can prove stronger upper bounds on $t$ in cluster packing graphs, with a dependence on $k$. This in turn may potentially pave the path for making progress on density of RS graphs. 

\clearpage


\section{Omitted Proofs of Standard Results}\label{app:missing}

We present the proofs of some standard results missing from the main body of the paper here. 

\subsection{Proof of~\Cref{claim:num-of-sets}}\label{sec:claim:num-of-sets}

\begin{claim*}
        There exist $t=2^{\Omega(d)}$ sets $S_1, \ldots, S_t\subseteq [d]$ such that for all $i \neq j \in [d]$, 
        \[
        		\card{S_i} = \frac{d}{3} \quad \textnormal{and} \quad \card{S_i \cap S_j} \leq \frac{d}{7}.
	\]
\end{claim*}

We first need the following standard concentration inequality. A \emph{hypergeometric} random variable with parameters $N$, $K$, and $M$ is a discrete random variable
in $\IN$ distributed as follows: we have $N$ elements, $K$ of them are marked as ``good'', and we sample
$M$ elements uniformly at random and without replacement and count the number of good samples.
We use a standard result on the concentration of hypergeometric random variables.

\begin{proposition}[cf.~{\cite[Theorem 2.10]{JansonLR11}}]\label{prop:hypergeometric-concentration-ineq}
    Suppose $X$ is a hypergeometric random variable with parameters $N$,$K$,$M$ and thus the expectation $\Exp[X] = M\cdot K/N$. Then, for any $t\geq 0$    
    \[
    \Pr(X - \Exp[X] \geq t\cdot M) \leq \exp \left(-2t^2\cdot M \right).
    \]
\end{proposition}

We can now prove~\Cref{claim:num-of-sets}. 

\begin{proof}[Proof of~\Cref{claim:num-of-sets}]
    Pick the sets $S_1, \ldots, S_t$ independently and uniformly at random from the collection of subsets of $[d]$ of size $\frac{d}{3}$. Fix some pair $i \ne j \in [t]$. Conditioning on the elements of $S_i$, the process of drawing $S_j$ can be viewed as a hypergeometric process, where the ``good" elements are those in $S_i$. 
	Then the random variable $X = \left|S_i \cap S_j\right|$ follows a hypergeometric distribution with parameters $d$, $\frac{d}{3}$, and $\frac{d}{3}$. Hence, we have
	$$\Exp[\left|S_i\cap S_j\right|]=\frac{d}{9}.$$
	By~\Cref{prop:hypergeometric-concentration-ineq}, we have
	\begin{align*}
		\Pr\left(\left|S_i\cap S_j\right|>\frac{d}{7}\right)
		&=\Pr\left(\left|S_i\cap S_j\right| - \Exp{\left|S_i\cap S_j\right|} > \frac{2}{63}\cdot d \right)\\
		&\leq \exp\left(-2(\frac{2}{63})^2\cdot d\right)\\
		&\leq 2^{-d/10^3}.
	\end{align*}
	Therefore, for $t = 2^{\Omega(d)}$, by the union bound over all $\binom{t}{2}$ pairs, we have
	\begin{align*}
		&\Pr\left(\forall i\ne j\in[t], \left|S_i\cap S_j\right|\leq \frac{d}{7} \right)\\
		&=1-\Pr\left(\exists i\ne j\in[t], \left|S_i\cap S_j\right|\leq \frac{d}{7} \right) >0.
	\end{align*}
	Thus, such a family of sets $S_1, \ldots, S_t$ exists by the probabilistic argument.
\end{proof}

\subsection{Proof of~\Cref{prop:robust-index}}\label{sec:prop:robust-index}

\begin{proposition*}
	For any $a,b,m \geq 1$, any one-way communication protocol $\prot$ for $Index^{a,b}_m$ in the robust communication model with error probability at most $\delta = 2^{-(a+b)}/4$ has
	\[
		CC(\pi) = \Omega(m) \cdot (1-H_2(\delta)) = \Omega(m \cdot 4^{-(a+b)})~bits. 
	\]
\end{proposition*}

As stated earlier,~\cite[Theorem 5.3]{ChakrabartiCM08} only proves~\Cref{prop:robust-index} for constant values of $a,b \geq 1$, but explicitly tracking the parameters in their proofs imply the above bounds as well for 
all $a,b \geq 1$ that may depend on $m$ also. We briefly sketch the proof. 

\begin{proof}[Proof Sketch.]
By their proof, using $\pi$ we can get protocol a $\pi_{\text{Index}}$ for the standard Index problem (where Alice gets $x \in \set{0,1}^m$ and Bob gets $i \in [m]$) with error probability at most $\frac12 - \frac{2^{-(a+b)}}{4}$ for constant $a, b$ such that
\[
CC(\pi) \geq CC( \pi_\text{Index}).
\]
By their Lemma 2.5 we know 
\[
CC(\pi_{\text{Index}}) \geq (1-H_2\paren{\frac12 - \frac{2^{-(a+b)}}{4}})\cdot m.
\]
So if we use the same protocol for any $a, b\geq 1$, we get
\begin{align*}
	CC(\pi) 
	&\geq CC( \pi_\text{Index})
	\\
	&\geq (1-H_2\paren{\frac12 - \frac{2^{-(a+b)}}{4}})\cdot m
	\\
	&\geq (\frac{4^{-(a+b)}}{16}) \cdot m.
\end{align*}
since $1-H_2((1/2)-\delta) \geq \delta^2$ for all $\delta \leq 1/4$ (which is the case in~\Cref{prop:robust-index} as $a+b \geq 0$). Therefore, $CC(\pi)=\Omega(m \cdot 4^{-(a+b)})$, 
concluding the proof. 
\end{proof}

\subsection{Proof of~\Cref{lem:set-limited-size}}\label{sec:lem:set-limited-size}

\begin{lemma*}
Let $n$ be a positive integer and $\varepsilon \in (0, 1/2)$. Let $\mathcal{F}$ be a family of subsets of $[n]$ such that for every $S \in \mathcal{F}$, we have $\card{S} = \varepsilon \cdot n$, and for all distinct $S, T \in \mathcal{F}$, we have 
$\left|S \cap T \right| \leq \varepsilon^2 \cdot (1-\delta)\cdot n$ for some $\delta\in (0, 1)$.
Then, $\left|\mathcal{F}\right| = O(\frac{1}{\delta\cdot \varepsilon})$.
\end{lemma*}
\begin{proof}
    Let $\mathcal{F}=\{S_1, \ldots, S_t\}$.
    For each $i\in [n]$, define $x_i$ to be the number of sets $S\in \mathcal{F}$ such that $i\in S$. Then we have
    \begin{align}
    \sum_{i=1}^n x_i = \sum_{j=1}^t |S_j| = t \cdot \varepsilon n. \label{eq:sum-xi}
    \end{align}
    For each $i \in [n]$, the element $i$ appears in ${x_i \choose 2}$ intersections among the $S_j$'s. Thus, summing over all $i \in [n]$:
    \begin{align}
    \sum_{\substack{S \neq T \\ S,T \in \mathcal{F}}} |S \cap T| = \sum_{i=1}^n {x_i \choose 2} \leq {t \choose 2} \cdot \varepsilon^2 \cdot (1-\delta) \cdot n. \label{eq:sum-intersections}
    \end{align}
    Combining~\Cref{eq:sum-xi} and~\Cref{eq:sum-intersections} we have
    \begin{align*}
        \sum_{i\in [n]} x_i^2- \sum_{i\in [n]}x_i
        &=\sum_{i\in [n]} x_i^2 - t\cdot \varepsilon n \leq (t^2-t)\cdot \varepsilon^2 \cdot (1-\delta)\cdot n.
    \end{align*}
    By Cauchy–Schwarz inequality we have 
    \[
    (\sum_{i\in [n]} x_i^2) \geq \frac{(\sum_{i\in [n]} x_i)^2}{n}=\frac{(t\cdot \varepsilon n)^2}{n} = t^2\cdot\varepsilon^2 n .
    \]
    Therefore, 
    $
    t^2\cdot \varepsilon^2 n - t \cdot \varepsilon n \leq (t^2-t)\cdot \varepsilon^2 \cdot (1-\delta)\cdot n
    $,
    which result in 
    \[
    t \leq \frac{1}{\varepsilon\cdot \delta} = O(\frac{1}{\varepsilon\cdot \delta}),
    \]
    concluding the proof.
\end{proof}

\clearpage

\section{A Schematic Organization for the Proof of \Cref{thm:adv-res}}\label{sec:schema-org}


This a schematic organization of the flow of main components of the proof of \Cref{thm:adv-res}. Each arrow points to the main component(s) used in the proof of originating component. 

\bigskip

\begin{center}
\hspace{-1.5cm}
\begin{tikzpicture}
	\node[draw, rounded corners, minimum width=5cm, minimum height=2cm, align=center] (box1)
	{
		{\textbf{\Cref{thm:adv-res}}} \\
		(Main result) \\[0.5cm] 
		Introduced in {\Cref{sec:adversarial}} \\ 
		Proved in \Cref{subsec:thm-adv-res-proof}
	};

	\node[draw, rounded corners, minimum width=5cm, minimum height=2cm, align=center] (box2) [below=1.5cm of box1]
	{
		\textbf{\Cref{lem:adv-lb-cluster}} \\
		(Lower bound in terms of \\
		cluster packing graph parameters)  \\[0.5cm] 
		Introduced in {\Cref{subsec:p-player-back}} \\ 
		Proved in \Cref{subsec:player-elim}
	};
	
		\node[draw, rounded corners, minimum width=5cm, minimum height=2cm, align=center] (box-clust1) [left=0.5cm of box2]
	{
		\textbf{\Cref{prop:cpg-dense}} \\
		Dense Cluster Packing Graphs \\[0.5cm] 
		Introduced in {\Cref{sec:cpg-definition}} \\ 
		Proved implicitly in \cite{AssadiKNS24}, \\
		also presented in \Cref{app:cpg-dense}
	};
	
		\node[draw, rounded corners, minimum width=5cm, minimum height=2cm, align=center] (box-clust2) [right=0.5cm of box2]
	{
		\textbf{\Cref{prop:cpg-large-r}} \\
		Cluster Packing Graphs with\\
		Large Cluster Size \\[0.5cm] 
		Introduced in {\Cref{sec:cpg-definition}} \\ 
		Proved in \Cref{sec:sparse-cluster-packing}
	};
	
	\node[draw, rounded corners, minimum width=5cm, minimum height=2cm, align=center] (box3) [below=1.5cm of box2]
	{
		\textbf{\Cref{lem:inductive-hardness-Gp}} \\
		(Player Elimination Protocol)  \\[0.5cm] 
		Introduced in {\Cref{subsec:player-elim}} \\ 
		Proved in \Cref{subsec:player-elim}
	};
	
	\node[draw, rounded corners, minimum width=5cm, minimum height=2cm, align=center] (box4) [below=1.5cm of box3]
	{
		\textbf{\Cref{lem:first-msg}} \\
		(Message of First Player \\
		Reveals Low Information) \\[0.5cm] 
		Introduced in {\Cref{subsec:player-elim}} \\ 
		Proved in \Cref{subsec:adv-specbit}
	};
	
	\node[draw, rounded corners, minimum width=5cm, minimum height=2cm, align=center] (box5) [below left =2.8cm and 1cm of box4.center] 
	{
		\textbf{\Cref{clm:set-int-final}} \\
		(Distribution of Special Set) \\[0.5cm] 
		Introduced in {\Cref{subsec:adv-specset}} \\ 
		Proved in \Cref{subsec:adv-specset}
	};
	
	\node[draw, rounded corners, minimum width=5cm, minimum height=2cm, align=center] (box6) [below right=2.8cm and 1cm of box4.center]
	{
		\textbf{\Cref{clm:index-hard}} \\
		(Distribution of Special Bit) \\[0.5cm] 
		Introduced in \Cref{subsec:adv-specbit} \\ 
		Proved in \Cref{subsec:adv-specbit}
	};

	\draw[->, line width=1pt]
	(box1) -- (box2);
	
	\draw[->, line width=1pt]	
	(box2) -- (box3);
	
	\draw[->, line width=1pt]	
	(box3) -- (box4);
	
	\draw[->, line width=1pt]	
	(box4) -- (box5);

	\draw[->, line width=1pt]	
	(box4) -- (box6);

\draw[<-, line width=1pt]
(box-clust1) -- (box1);

\draw[<-, line width=1pt]
(box-clust2) -- (box1);	
	
\end{tikzpicture}
\end{center}

\clearpage

\newcommand{\rZ}{\ensuremath{\rv{Z}}}
\newcommand{\cU}{\ensuremath{\mathcal{U}}}

\section{Background on Information Theory}\label{app:info}

We now briefly introduce some definitions and facts from information theory that are used in our proofs. We refer the interested reader to the text by Cover and Thomas~\cite{CoverT06} for an excellent introduction to this field, 
and the proofs of the statements used in this Appendix. 

For a random variable $\rA$, we use $\supp{\rA}$ to denote the support of $\rA$ and $\distribution{\rA}$ to denote its distribution. 
When it is clear from the context, we may abuse the notation and use $\rA$ directly instead of $\distribution{\rA}$, for example, write 
$A \sim \rA$ to mean $A \sim \distribution{\rA}$, i.e., $A$ is sampled from the distribution of random variable $\rA$. 

\begin{itemize}[leftmargin=10pt]
\item We denote the \emph{Shannon Entropy} of a random variable $\rA$ by
$\en{\rA}$, which is defined as: 
\begin{align}
	\en{\rA} := \sum_{A \in \supp{\rA}} \Pr\paren{\rA = A} \cdot \log{\paren{1/\Pr\paren{\rA = A}}} \label{eq:entropy}
\end{align} 
\noindent
\item The \emph{conditional entropy} of $\rA$ conditioned on $\rB$ is denoted by $\en{\rA \mid \rB}$ and defined as:
\begin{align}
\en{\rA \mid \rB} := \Ex_{B \sim \rB} \bracket{\en{\rA \mid \rB = B}}, \label{eq:cond-entropy}
\end{align}
where 
$\en{\rA \mid \rB = B}$ is defined in a standard way by using the distribution of $\rA$ conditioned on the event $\rB = B$ in \Cref{eq:entropy}.

\item The \emph{mutual information} of two random variables $\rA$ and $\rB$ is denoted by
$\mi{\rA}{\rB}$ and is defined:
\begin{align}
\mi{\rA}{\rB} := \en{A} - \en{A \mid  B} = \en{B} - \en{B \mid  A}. \label{eq:mi}
\end{align}
\noindent
\item The \emph{conditional mutual information} $\mi{\rA}{\rB \mid \rC}$ is $\en{\rA \mid \rC} - \en{\rA \mid \rB,\rC}$ and hence by linearity of expectation:
\begin{align}
	\mi{\rA}{\rB \mid \rC} = \Ex_{C \sim \rC} \bracket{\mi{\rA}{\rB \mid \rC = C}}. \label{eq:cond-mi}
\end{align} 
\end{itemize}

\subsection{Useful Properties of Entropy and Mutual Information}\label{sec:prop-en-mi}

We shall use the following basic properties of entropy and mutual information throughout. 

\begin{fact}\label{fact:it-facts}
  Let $\rA$, $\rB$, $\rC$, and $\rD$ be four (possibly correlated) random variables.
   \begin{enumerate}
  \item \label{part:uniform} $0 \leq \en{\rA} \leq \log{\card{\supp{\rA}}}$. The right equality holds
    iff $\distribution{\rA}$ is uniform.
  \item \label{part:info-zero} $\mi{\rA}{\rB \mid \rC} \geq 0$. The equality holds iff $\rA$ and
    $\rB$ are \emph{independent} conditioned on $\rC$.
    \item \label{part:info-entropy} $\mi{\rA}{\rB \mid \rC} \leq \en{\rB}$ for any random variables $\rA, \rB, \rC$. 
  \item \label{part:cond-reduce} \emph{Conditioning on a random variable reduces entropy}:
    $\en{\rA \mid \rB,\rC} \leq \en{\rA \mid  \rB}$.  The equality holds iff $\rA \perp \rC \mid \rB$.
  \item \label{part:chain-rule} \emph{Chain rule for mutual information}: $\mi{\rA,\rB}{\rC \mid \rD} = \mi{\rA}{\rC \mid \rD} + \mi{\rB}{\rC \mid  \rA,\rD}$.
  \item \label{part:data-processing} \emph{Data processing inequality}: for a function $f(\rA)$ of $\rA$, $\mi{f(\rA)}{\rB \mid \rC} \leq \mi{\rA}{\rB \mid \rC}$. 
   \end{enumerate}
\end{fact}


\noindent
We also use the following two standard propositions, regarding the effect of conditioning on mutual information.

\begin{proposition}\label{prop:info-increase}
  For random variables $\rA, \rB, \rC, \rD$, if $\rA \perp \rD \mid \rC$, then, 
  \[\mi{\rA}{\rB \mid \rC} \leq \mi{\rA}{\rB \mid  \rC,  \rD}.\]
\end{proposition}
 \begin{proof}
  Since $\rA$ and $\rD$ are independent conditioned on $\rC$, by
  \itfacts{cond-reduce}, $\HH(\rA \mid  \rC) = \HH(\rA \mid \rC, \rD)$ and $\HH(\rA \mid  \rC, \rB) \ge \HH(\rA \mid  \rC, \rB, \rD)$.  We have,
	 \begin{align*}
	  \mi{\rA}{\rB \mid  \rC} &= \HH(\rA \mid \rC) - \HH(\rA \mid \rC, \rB) = \HH(\rA \mid  \rC, \rD) - \HH(\rA \mid \rC, \rB) \\
	  &\leq \HH(\rA \mid \rC, \rD) - \HH(\rA \mid \rC, \rB, \rD) = \mi{\rA}{\rB \mid \rC, \rD}. \qed
	\end{align*}
	
\end{proof}

\begin{proposition}\label{prop:info-decrease}
  For random variables $\rA, \rB, \rC,\rD$, if $ \rA \perp \rD \mid \rB,\rC$, then, 
  \[\mi{\rA}{\rB \mid \rC} \geq \mi{\rA}{\rB \mid \rC, \rD}.\]
\end{proposition}
 \begin{proof}
 Since $\rA \perp \rD \mid \rB,\rC$, by \itfacts{cond-reduce}, $\HH(\rA \mid \rB,\rC) = \HH(\rA \mid \rB,\rC,\rD)$. Moreover, since conditioning can only reduce the entropy (again by \itfacts{cond-reduce}), 
  \begin{align*}
 	\mi{\rA}{\rB \mid  \rC} &= \HH(\rA \mid \rC) - \HH(\rA \mid \rB,\rC) \geq \HH(\rA \mid \rD,\rC) - \HH(\rA \mid \rB,\rC) \\
	&= \HH(\rA \mid \rD,\rC) - \HH(\rA \mid \rB,\rC,\rD) = \mi{\rA}{\rB \mid \rC,\rD}. \qed
 \end{align*}

\end{proof}

\subsection{Measures of Distance Between Distributions}\label{sec:prob-distance}

We use two main measures of distance (or divergence) between distributions, namely the \emph{Kullback-Leibler divergence} (KL-divergence) and the \emph{total variation distance}. 

\paragraph{KL-divergence.} For two distributions $\mu$ and $\nu$ over the same probability space, the \textbf{Kullback-Leibler (KL) divergence} between $\mu$ and $\nu$ is denoted by $\kl{\mu}{\nu}$ and defined as: 
\begin{align}
\kl{\mu}{\nu}:= \Ex_{a \sim \mu}\Bracket{\log\frac{\mu(a)}{{\nu}(a)}}. \label{eq:kl}
\end{align}
We also have the following relation between mutual information and KL-divergence. 
\begin{fact}\label{fact:kl-info}
	For random variables $\rA,\rB,\rC$, 
	\[\mi{\rA}{\rB \mid \rC} = \Ex_{(B,C) \sim {(\rB,\rC)}}\Bracket{ \kl{\distribution{\rA \mid \rB=B,\rC=C}}{\distribution{\rA \mid \rC=C}}}.\] 
\end{fact}


\paragraph{Total variation distance.} We denote the \textbf{total variation distance} between two distributions $\mu$ and $\nu$ on the same 
support $\Omega$ by $\tvd{\mu}{\nu}$, defined as: 
\begin{align}
\tvd{\mu}{\nu}:= \max_{\Omega' \subseteq \Omega} \paren{\mu(\Omega')-\nu(\Omega')} = \frac{1}{2} \cdot \sum_{x \in \Omega} \card{\mu(x) - \nu(x)}.  \label{eq:tvd}
\end{align}
\noindent
We use the following basic properties of total variation distance. 
\begin{fact}\label{fact:tvd-small}
	Suppose $\mu$ and $\nu$ are two distributions for $\event$, then, 
	$
	{\mu}(\event) \leq {\nu}(\event) + \tvd{\mu}{\nu}.
$
\end{fact}


We also have the following (chain-rule) bound on the total variation distance of joint variables.

\begin{fact}\label{fact:tvd-chain-rule}
	For any distributions $\mu$ and $\nu$ on $n$-tuples $(X_1,\ldots,X_n)$, 
	\[
		\tvd{\mu}{\nu} \leq \sum_{i=1}^{n} \Exp_{X_{<i} \sim \mu} \tvd{\mu(X_i \mid X_{<i})}{\nu(X_i \mid X_{<i})}. 
	\]
\end{fact}

We also have the following ``over conditioning'' property. 

\begin{fact}\label{fact:tvd-over-conditioning}
	For any random variables $\rX,\rY,\rZ$, 
	\[
		\tvd{\rX}{\rY} \leq \tvd{\rX\rZ}{\rY\rZ} = \Exp_{Z} \tvd{(\rX \mid \rZ=Z)}{(\rY \mid \rZ=Z)}.  
	\]
\end{fact}
\begin{proof}
First, we prove the equality between the second term and the third term in the statement. 
\begin{align*}
	\tvd{\rX\rZ}{\rY\rZ} &=\frac12 \cdot \sum_{W, Z} \Pr[(W, Z)] \card{\Pr[\rX\rZ = (W, Z)] - \Pr[\rY\rZ=(W, Z)]}  \\
	&=\frac12 \cdot \sum_{W, Z}  \card{\Pr[\rZ = Z] \cdot (\Pr[\rX = W \mid  \rZ = Z] - \Pr[\rY = W \mid \rZ= Z])} \\
	&= \sum_{Z} \Pr[\rZ =Z] \cdot \frac12 \sum_{W}  \card{\Pr[\rX = W \mid  \rZ = Z] - \Pr[\rY = W \mid \rZ= Z]} \\
	&= \sum_{Z} \Pr[\rZ =Z] \cdot\tvd{(\rX \mid \rZ= Z)}{(\rY \mid \rZ = Z)} \\
	&= \Exp_{Z} \tvd{(\rX \mid \rZ= Z)}{(\rY \mid \rZ = Z)}.
\end{align*}

Now we prove the inequality between the first term and the third term. 
\begin{align*}
		\tvd{\rX}{\rY} &= \frac12 \cdot \sum_{W} \card{\Pr[\rX = W] - \Pr[\rY = W]} \\
		&= \frac12 \cdot \sum_{W} \card{\sum_{Z} \Pr[\rZ = Z](\Pr[\rX = W \mid \rZ = Z] - \Pr[\rY = W \mid \rZ = Z])} \\
		& \leq  \frac12 \cdot \sum_{W}  \sum_{Z} \Pr[\rZ = Z] \card{\Pr[\rX = W \mid \rZ = Z] - \Pr[\rY = W \mid \rZ = Z]} \\
		&= \sum_{Z}   \Pr[\rZ = Z] \cdot \paren{\frac12 \cdot \sum_{W} \card{\Pr[\rX = W \mid \rZ = Z] - \Pr[\rY = W \mid \rZ = Z]}}  \\
		&= \Exp_{Z} \tvd{\rX \mid \rZ= Z}{\rY \mid \rZ = Z}. \qedhere
\end{align*}
\end{proof}

We also have the following data processing inequality for total variation distance. 

\begin{fact}\label{fact:tvd-data-processing}
	Suppose $\rX$ and $\rY$ are two random variables with the same support $\Omega$ and $f: \Omega \rightarrow \Gamma$ is a fixed random process.  Then,
	\[
		\tvd{f(\rX)}{f(\rY)} \leq \tvd{\rX}{\rY}. 
	\]
\end{fact}

The following Pinsker's inequality bounds the total variation distance between two distributions based on their KL-divergence.

\begin{fact}[Pinsker's inequality]\label{fact:pinskers}
	For any distributions $\mu$ and $\nu$, 
	$
	\tvd{\mu}{\nu} \leq \sqrt{\frac{1}{2} \cdot \kl{\mu}{\nu}}.
	$ 
\end{fact}

\begin{fact}\label{fact:Fanos-inequality}(Fano's inequality) Let $A, B$ be random variables and $f$ be a function that given $A$ predicts a value for $B$. If $\Pr(f(A)\ne B)\leq \delta$, then 
	\[
	\en{B\mid A}\leq H_2(\delta)+\delta\cdot \left(\log \left|B\right|-1\right).
	\]
	If B is binary, then the bound improves to $\en{B \mid A} \leq H_2(\delta)$.
\end{fact}

\end{document}